\documentclass[envcountsame,runningheads]{llncs}
\usepackage[T1]{fontenc}
\usepackage{graphicx}
\usepackage{booktabs, multirow}

\usepackage{hyphenat}

\usepackage{stmaryrd}
\usepackage{amsfonts}
\usepackage{amsmath, nccmath, amssymb}
\usepackage{amsfonts}
\usepackage{color, ifthen} 
\usepackage{mathtools}
\usepackage{graphicx}
\usepackage{cite}
\usepackage{siunitx}

\usepackage{enumitem} 
\usepackage{wrapfig}
\usepackage[all]{xy} \SelectTips{cm}{}  

\usepackage[bookmarks=true, hidelinks, pdfencoding=auto]{hyperref}
\usepackage{pdfpages}
\usepackage{numprint}
\usepackage[table]{xcolor}

\usepackage{comment}
\specialcomment{auxproof}
{\mbox{}\newline\textbf{BEGIN: AUX-PROOF}\dotfill\newline}
{\mbox{}\newline\textbf{END: AUX-PROOF}\dotfill\newline}
\excludecomment{auxproof}

\graphicspath{{../inkscape/}}

\newcommand{\rloop}[2][-]{\save \POS!R(.7) \ar@(ru,rd)^#1{#2} \restore}
\newcommand{\lloop}[2][-]{\save \POS!L(.7) \ar@(lu,ld)_#1{#2} \restore}
\newcommand{\uloop}[2][-]{\save \POS!U(.7) \ar@(lu,ru)^(.8){#2} \restore}

\newcommand\defeq{\stackrel{\text{\tiny def}}{=}}

\newcommand{\Nat}{\mathbb{N}}
\newcommand{\states}{\mathcal{S}}

\newcommand{\tr}[1]{\stackrel{#1}{\to}}
\newcommand{\trz}[2]{{\stackrel{#1}{\to}}_{#2}}
\newcommand{\Htrz}[3]{{\stackrel{#1}{\to}}_{#2}^{#3}}
\newcommand{\ttr}[1]{\stackrel{#1}{\to^*}}
\newcommand{\ttp}[1]{\stackrel{#1}{\to^+}}

\newcommand{\myparagraph}[1]{\vspace*{.3em}\noindent\textbf{#1}\;}

\usepackage{algorithm2e}

\newcommand{\cvec}[4]{\!{\tiny{ \!\begin{bmatrix} #1 \\ #2 \\ #3 \\#4 \end{bmatrix}\!}}\!}

\def\APDR{{\texttt{AdjointPDR}}\xspace}
\def\ADPDR{{\texttt{AdjointPDR}$^\downarrow$}\xspace}

\newcommand{\negation}[1]{\vec{neg(#1)}}

\newif\ifdraft\draftfalse

\ifdraft
\usepackage{todonotes}
\reversemarginpar
\newcommand\kori[1]{\textcolor{blue}{#1}}
\newcommand\korit[1]{\todo[color=green!40]{#1 --kori}}
\newcommand\todorb[1]{\todo[color=blue!40]{RB: #1}}

\newcommand{\conf}[1]{}
\newcommand{\todoil}[1]{\todo[inline,caption={}]{#1}}
\else
\usepackage[disable]{todonotes}
\newcommand\kori[1]{#1}
\newcommand\korit[1]{}
\newcommand\todorb[1]{}
\excludecomment{auxproof}
\newcommand{\conf}[1]{}
\newcommand{\todoil}[1]{}
\fi
\usepackage{amsfonts}
\usepackage{galois}
\presetkeys{todonotes}{inline}{}
\usepackage{listings}
\definecolor{dkblue}{rgb}{0,0.1,0.6}
\definecolor{dkgreen}{rgb}{0,0.35,0}
\definecolor{dkviolet}{rgb}{0.3,0,0.5}
\definecolor{dkred}{rgb}{0.5,0,0}
\lstdefinelanguage{NT}{
mathescape=true,
texcl=false,
morekeywords=[1]{while, do, if, then, else, for, all},
morekeywords=[2]{return, continue, goto},
morecomment=[s]{/*}{*/},
showstringspaces=false,
morestring=[b]",
morestring=[d],
tabsize=3,
extendedchars=false,
sensitive=true,
breaklines=false,
basicstyle=\ttfamily,
captionpos=b,
columns=[l]fixed,
identifierstyle={\color{black}},
keywordstyle=[1]{\color{dkviolet}},
keywordstyle=[2]{\color{dkred}},
stringstyle=\ttfamily,
commentstyle={\ttfamily\color{dkblue}},
literate=
    {true}{{{\color{dkgreen}$true$}}}3
    {false}{{{\color{dkgreen}$false$}}}4
    {choose}{{{\color{dkred}choose}}}6
    {st}{{{\color{dkred}such that}}}{9}
    {And}{{{\color{dkgreen}$and$}}}4
     {case}{{{\color{dkviolet}case}}}4
     {of}{{{\color{dkviolet}of}}}3
    {endcase}{{{\color{dkviolet}endcase}}}3
}[keywords,comments,strings]
\lstnewenvironment{codeNT}{\lstset{language=NT}}{}

\makeatletter
	\def\@citecolor{blue}%
	\def\@urlcolor{blue}%
	\def\@linkcolor{blue}%
	
	\def\orcidID#1{\smash{\href{http://orcid.org/#1}{\protect\raisebox{-1.25pt}{\protect\includegraphics{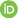}}}}}
\makeatother

\newsavebox\lstbox

\usepackage{graphicx}                   %
\usepackage{hyperref}                   %
\usepackage[firstpage]{draftwatermark}  %

\SetWatermarkAngle{0}

\SetWatermarkText{\raisebox{12.8cm}{%
  \hspace{-6cm}%
}}

\begin{document}
\title{Exploiting Adjoints in Property Directed Reachability Analysis
	\thanks{Research supported
    by MIUR PRIN Project 201784YSZ5 \emph{ASPRA},
    by JST ERATO HASUO Metamathematics for Systems Design Project JPMJER1603,
    by JST CREST Grant JPMJCR2012,
    by JSPS DC KAKENHI Grant 22J21742 and
    by EU Next-GenerationEU (NRRP) SPOKE~10, Mission~4, Component~2, Investment N.~1.4, CUP N.~I53C22000690001.
}}
\titlerunning{Exploiting Adjoints in PDR}
\author{Mayuko Kori\inst{1,2}\orcidID{0000-0002-8495-5925}
\and Flavio Ascari\inst{3}\orcidID{0000-0003-4624-9752}
\and Filippo Bonchi\inst{3}\orcidID{0000-0002-3433-723X}
\and Roberto Bruni\inst{3}\orcidID{0000-0002-7771-4154}
\and Roberta Gori\inst{3}\orcidID{0000-0002-7424-9576}
\and Ichiro Hasuo\inst{1,2}\orcidID{0000-0002-8300-4650}
 }
\authorrunning{M.Kori  et al.}
\institute{National Institute of Informatics, Tokyo, Japan
\and
The Graduate University for Advanced Studies\\ (SOKENDAI), Hayama, Japan\\
\email{\{mkori,hasuo\}@nii.ac.jp} \and%
Dipartimento di Informatica, Universit\`a di Pisa, Pisa, Italy
\email{flavio.ascari@phd.unipi.it} \email{\{filippo.bonchi,roberto.bruni,roberta.gori\}@unipi.it}
}
\maketitle              %

\begin{abstract}
	We formulate, in lattice-theoretic terms, two novel algorithms inspired by Bradley's property directed reachability algorithm.
	For finding safe invariants or counterexamples, the first algorithm exploits over-approximations of both forward and backward transition relations, expressed abstractly by the notion of adjoints.
	In the absence of adjoints, one can use the second algorithm, which exploits lower sets and their principals.
	As a notable example of application, we consider quantitative reachability problems for Markov Decision Processes.
\end{abstract}

\section{Introduction}\label{sec:intro}

\emph{Property directed reachability analysis} (PDR)
refers to a class of verification algorithms for solving safety problems of transition systems~\cite{Bradley11,EenMB11}.
Its essence %
consists of 1) interleaving the construction of an \emph{inductive invariant} (a \emph{positive chain}) with that of a \emph{counterexample} (a \emph{negative sequence}), and 2)  making the two sequences \emph{interact}, with one  narrowing down the search space for the other.

PDR algorithms have shown impressive performance both in hardware and software verification, leading to active research~\cite{SeufertS18,SeufertS19,Gurfinkel2015IC3PA,HoderB12}  going far beyond its original scope.
For instance, an abstract domain~\cite{cousot21} capturing the over-approximation exploited by  PDR has been recently introduced in~\cite{DBLP:journals/pacmpl/FeldmanSSW22}, while PrIC3~\cite{BatzJKKMS20} extended PDR for quantitative verification of probabilistic systems.

To uncover the abstract principles behind PDR and its extensions, Kori et al. proposed LT-PDR~\cite{KoriCAV22}, a generalisation of PDR in terms of lattice/category theory.
LT-PDR can be instantiated using domain-specific \emph{heuristics} to create effective algorithms for different kinds of systems such as Kripke structures, Markov Decision Processes (MDPs), and Markov reward models.
However, the theory in~\cite{KoriCAV22} does not offer guidance on devising concrete heuristics.

\myparagraph{Adjoints in PDR.}
Our approach shares the same vision of LT-PDR, but we identify different principles: \emph{adjunctions} are the core of our toolset.

\begin{wrapfigure}[3]{r}{0pt}
\begin{minipage}{.16\textwidth}

 \vspace*{-2,5em}
   \begin{math}
    \xymatrix{
      A \ar@/_1.5ex/[r]_-{g}^-\bot &C\ar@/_1.5ex/[l]_-{f}
    }
  \end{math}
\end{minipage}
\end{wrapfigure}
An adjunction $f \dashv g$ is one of the central concepts in category theory~\cite{maclane:71}.
It is prevalent  in various fields of computer science, too, such as
abstract interpretation~\cite{cousot21} and functional programming~\cite{Levy2004}. Our use of adjoints in this work comes in the following two flavours.
\begin{itemize}
 \item (forward-backward adjoint) $f$ describes the \emph{forward semantics} of a transition system, while $g$ is the \emph{backward} one, where we typically have  $A=C$. %
 \item (abstraction-concretization adjoint) $C$ is a concrete semantic domain, and $A$ is an abstract one, much like in abstract interpretation. An adjoint enables us to convert a fixed-point problem in $C$ to that in $A$.
\end{itemize}

\myparagraph{Our Algorithms.}
The problem we address is the standard lattice theoretical formulation of safety problems, namely whether the least fixed point of a continuous map $b$ over a complete lattice $(L,\sqsubseteq)$ is below a given element $p\in L$. In symbols $\mu b\sqsubseteq_{?} p$. We present two algorithms. 

\begin{wrapfigure}[3]{r}{0pt}
\begin{minipage}{.16\textwidth}
 \vspace*{-2,5em}
   \begin{math}
    \xymatrix{
      L \ar@/_1.5ex/[r]_-{g}^-\bot &L\ar@/_1.5ex/[l]_-{f}
    }
  \end{math}
\end{minipage}
\end{wrapfigure}
The first one, named
{\APDR}, assumes to have an element $i\in L$ and two adjoints $f \dashv g\colon L \to L$, representing respectively initial states, forward semantics and backward semantics (see right) such that $b(x)=f(x)\sqcup i$ for all $x\in L$.
Under this assumption, we have the following equivalences (they follow from the Knaster-Tarski theorem, see~\S{}\ref{sec:preliminaries}):
$$\mu b\sqsubseteq p
\quad \Leftrightarrow\quad
\mu (f\sqcup i)\sqsubseteq p
\quad \Leftrightarrow\quad
i \sqsubseteq \nu (g \sqcap p),$$ 

\noindent
where $\mu (f\sqcup i)$  and $\nu (g \sqcap p)$ are, by the Kleene theorem, the limits of the \emph{initial} and \emph{final} chains illustrated below.
\[\bot \sqsubseteq i \sqsubseteq f(i)\sqcup i \sqsubseteq  \cdots \qquad\qquad \qquad  \cdots \sqsubseteq g(p)\sqcap p \sqsubseteq p \sqsubseteq \top\]
As positive chain, PDR exploits an over-approximation of the initial chain: it is made greater to accelerate convergence; still it has to be below $p$.

The distinguishing feature of {\APDR} is to take as a negative sequence (that is a sequential construction of potential counterexamples) an over\hyp{}approximation of the final chain. This crucially differs from the negative sequence of LT-PDR, namely an under-approximation of the computed positive chain.

We %
prove that \APDR{} is sound (Theorem \ref{th:soundness}) and does not loop (Proposition \ref{prop:progres}) but since, the problem $\mu b \sqsubseteq_? p$
is not always decidable, we cannot prove termination. Nevertheless, {\APDR} allows for a formal theory of heuristics that are essential when instantiating the algorithm to concrete problems. %
The theory  prescribes the choices to obtain the boundary executions, using initial and final chains (Proposition \ref{prop:negativesequencefinalchain}); it thus
identifies a class of heuristics guaranteeing termination when answers are negative (Theorem \ref{thm:negativetermination}).

\APDR's assumption of a forward-backward adjoint $f \dashv g$, however, does not hold very often, especially in probabilistic settings. Our second algorithm \ADPDR{} circumvents this problem by extending the lattice for the negative sequence, from $L$ to the lattice $L^{\downarrow}$ of \emph{lower sets} in $L$. 

\begin{wrapfigure}[4]{r}{0pt}
\begin{minipage}{.35\textwidth}

 \vspace*{-2em}
   \begin{math}
    \xymatrix{
     L \lloop{b} \ar@/_1.5ex/[r]_-{(-)^\downarrow}^-\bot
      &L^\downarrow \rloop{b^\downarrow
\, \dashv\, b^\downarrow_r
} \ar@/_1.5ex/[l]_-{\bigsqcup}
    }
  \end{math}
\end{minipage}
\end{wrapfigure}
 Specifically,
 by using the second form of adjoints, namely an abstraction-concretization pair,
 the problem $\mu b \sqsubseteq_{?} p$ in $L$ can be translated to an equivalent problem on $b^{\downarrow}$ in $L^\downarrow$, for which an adjoint $b^\downarrow
 \dashv b^\downarrow_r$ is guaranteed. This allows one to run \APDR{} in the lattice $L^\downarrow$. We then notice that the search for a positive chain can be conveniently restricted to principals in $L^\downarrow$, which have representatives in $L$. The resulting algorithm, using $L$ for positive chains and $L^\downarrow$ for negative sequences, is \ADPDR.

The use of lower sets for the negative sequence is a key advantage.
It not only avoids the restrictive assumption on forward-backward adjoints $f\dashv g$, but also enables a more thorough search for counterexamples. {\ADPDR} can simulate step-by-step LT-PDR (Theorem \ref{th:LT-PDR-instance-ADPDR}), while the reverse is not possible due to a single negative sequence in {\ADPDR} potentially representing multiple (Proposition~\ref{prop:multipleLTPDR}) or even all (Proposition~\ref{prop:LTPDRfinal}) negative sequences in LT-PDR.

\myparagraph{Concrete Instances.} Our lattice-theoretic algorithms yield many concrete instances: the original IC3/PDR~\cite{Bradley11,EenMB11} as well as 
Reverse PDR~\cite{SeufertS17} are instances of \APDR{} with $L$
being the powerset of the state space; 
since LT-PDR can be simulated by \ADPDR{}, the latter generalizes all instances in~\cite{KoriCAV22}. %

As a notable instance, we apply \ADPDR to MDPs, specifically to decide if the maximum reachability probability \cite{BaierK} is below a given threshold. Here the lattice $L=[0,1]^S$ is that of fuzzy predicates over the state space $S$.
Our theory provides guidance to devise two heuristics, for which we  prove negative termination (Corollary \ref{cor:ADPDRtermination}).
We present its implementation  in Haskell, and its experimental evaluation, where comparison is made against existing probabilistic PDR algorithms (PrIC3~\cite{BatzJKKMS20}, LT-PDR~\cite{KoriCAV22}) and a non-PDR one (Storm~\cite{DehnertJK017}). The performance of \ADPDR is encouraging---it supports the potential of PDR algorithms in probabilistic model checking. The experiments also indicate the importance of having a variety of heuristics, and thus the value of our adjoint framework that helps coming up with those.

Additionally, we found that abstraction features of Haskell allows us to code lattice-theoretic algorithms almost literally ($\sim$100 lines). Implementing a few heuristics takes another $\sim$240 lines. This way, we found that mathematical abstraction can directly help easing implementation effort.

\begin{auxproof}

\subsection*{draft}
\myparagraph{Property Directed Reachability.}
Property Directed Reachability Analysis, introduced in \cite{Bradley11,EenMB11} and shortly called PDR,
is a model checking algorithm solving safety problems of transition systems, e.g.~whether reachable states are always in safe states.
\todo[noinline]{TODO}
Although a naive algorithm first calculate reachable states and then compare them to safe states,
as the name of ``Property Directed'', PDR integrates these calculation.
In detail, they tries to over-approximate  reachable states as it is under safe states.
PDR is known for its performance both in hardware and software verification.

There are several studies that extends PDR to target many kinds of systems.
Batz et al.~\cite{BatzJKKMS20} suggested PrIC3, which is a PDR algorithm for Markov Decision Processes.
It is a first attempt to extend PDR for quantitative systems as far as we know.
Kori et al.~\cite{KoriCAV22} suggested LT-PDR, and they further extends the area of target systems
to general ones written in lattice theory including MDPs and Markov Reward Models.
Instead of generality in LT-PDR, the algorithm has some parts called heuristics that users need to define.
Though the choice of heuristics heavily affects the performance,
in their paper they didn't clarify how to define better heuristics and left it as a future work.

\todo[noinline]{position}
\begin{wrapfigure}[4]{r}{0pt}
  \vspace{-30pt}
  \begin{math}
    \xymatrix{
      A \ar@/_1.5ex/[r]_-{g}^-\bot &C\ar@/_1.5ex/[l]_-{f}
    }
  \end{math}
  \vspace{30pt}
\end{wrapfigure}
\myparagraph{Adjoint.}
Adjoint $f \dashv g$ is one of the central concepts in category theory.
It's prevalent also in computer science especially when using lattice theory.
For example, in abstract interpretation
a concrete domain $C$ and an abstract domain $A$ are connected by an adjoint as shown on the right.
This adjoint enables us to convert fixed-point problems from $C$ to $A$ so that we can easily calculate fixed-points.
Another is an adjoint between forward semantics $f$ and backward semantics $g$ of transition systems; we used it in LT-PDR.

\myparagraph{Our algorithm: \APDR.}
In this paper we exploit various adjoint structures in LT-PDR,
and propose our algorithms called \APDR and \ADPDR.
Let us first see the base algorithm \APDR
and its variant \ADPDR solving wider range of problems.

\ADPDR solves $\mu b \sqsubseteq p$ in a lattice $L$ with $b(x) = f(x) \sqcup i$ with $f \dashv g$.
It is a generalization of safety problems for transition systems.
Intuitively, $\mu b \sqsubseteq p$ in $L$ means whether all reachable states $\mu b$ are in safe states $p$.
It yields safety problems for several systems by changing the lattice $L$.

\todo[inline]{(Ichiro) My impression is that you speak of LT-PDR too much. Instead of ``we can, LT-PDR can't,'' you can just say ``we can.'' Comparison with LT-PDR should be on a much higher level (heuristics are left open in LT-PDR; adjoints suggest concrete heuristics here).}
It yields two equivalent safety problems $\mu (f \sqcup i) \sqsubseteq p$ and $i \sqsubseteq \nu (g \sqcap p)$.
\APDR successfully utilizes both of them
though LT-PDR do only the former one.
LT-PDR and \APDR look similar and they have two sequences for positive and negative side respectively.
However, thanks to the existence of $i \sqsubseteq \nu (g \sqcap p)$,
in \APDR we can use all potential counterexamples made from $\nu (g \sqcap p)$ as the sequence for negative side
and get heuristics with a termination property for negative side
though it's unclear for LT-PDR.

In the paper~\cite{KoriCAV22},
we pointed out one crucial problem of the adjointness assumption on $b$:
there are many cases including MDPs and MRMs that $b$ does not satisfy the assumption.
In this paper we overcome the problem by using adjoints with lower sets.
We find that
for any lattice $L$ and $b: L \to L$,
lower sets of $L$ consists a complete lattice $L^\downarrow$,
and along the downward closure $(-)^\downarrow \colon L \to L^\downarrow$, which is a right adjoint,
we can transform the problem $\mu b \sqsubseteq p$ in $L$ into an equivalent one in $L^\downarrow$.
This conversion gives us another algorithm \ADPDR, which is a special variant of \APDR using lower sets.
As the above discussion, it is applicable to any problem $\mu b \sqsubseteq p$ in $L$.

\todo[inline]{Speak about the comparison with existing IC3/PDR algorithms.}
\APDR generalizes IC3/PDR~\cite{Bradley11,EenMB11} by letting $L\coloneqq \mathcal{P}S$.
\ADPDR simulates LT-PDR. (One can see that \APDR deals with multiple Kleene sequences at once, see \S{}\ref{ssec:LTPDRvsADPDR})

\myparagraph{Concrete Instances.}
As an  example, letting $L=[0, 1]^S$,
we instantiate \ADPDR to solve max reachability problems with thresholds? for Markov Decision Processes.
Exploiting adjoints to identify the heuristics,
we introduce several heuristics with a termination property.
We implement \ADPDR with these heuristics in Haskell
and assess the performance by comparing it to PrIC3, LT-PDR, and Storm.

\end{auxproof}

\myparagraph{Related Work.}%
Reverse PDR~\cite{SeufertS17} applies PDR from unsafe states using a backward transition relation $\mathbf{T}$ and tries to prove that initial states are unreachable.
Our right adjoint $g$ is also backward, but it differs from $\mathbf{T}$ in the presence of nondeterminism: roughly, $\mathbf{T}(X)$ is the set of states which \emph{can} reach $X$ in one step,
while $g(X)$ are states which \emph{only} reach $X$ in one step.
\textit{fb}PDR~\cite{SeufertS18,SeufertS19} runs PDR and Reverse PDR in parallel with shared information.
Our work uses both forward and backward directions (the pair $f\dashv g$), too, but approximate differently: Reverse PDR over-approximates the set of states that can reach an unsafe state, while we over-approximate the set of states that only reach safe states.

The comparison with  LT-PDR~\cite{KoriCAV22} is extensively discussed in Section~\ref{ssec:LTPDRvsADPDR}.
PrIC3~\cite{BatzJKKMS20} extended PDR to MDPs, which are our main experimental ground: Section~\ref{sec:experiments} compares the performances of PrIC3, LT-PDR and {\ADPDR}.

We remark that PDR has been applied to other settings,  such as software model checking using theories and SMT-solvers~\cite{CimattiG12,LangeNNK20} or automated planning~\cite{Suda14}.
Most of them (e.g., software model checking) fall already in the generality of LT-PDR and thus they can be embedded in our framework.

It is also worth to mention that, in the context of abstract interpretation, the use of adjoints to construct initial and final chains and exploit the interaction between their approximations has been investigated in several works, e.g.,~\cite{DBLP:conf/sara/Cousot00}.

\myparagraph{Structure of the paper.}
After recalling some preliminaries in Section~\ref{sec:preliminaries}, we present {\APDR} in Section~\ref{sec:APDR} and {\ADPDR} in Section~\ref{sec:downset}. In Section~\ref{sec:MDP} we introduce the heuristics for the max reachability problems of MDPs, that are experimentally tested in Section~\ref{sec:experiments}. 

\section{Preliminaries and Notation}\label{sec:preliminaries}

We assume that the reader is familiar with lattice theory, see, e.g.,~\cite{DBLP:books/daglib/0023601}.
We use $(L,\sqsubseteq)$, $(L_1,\sqsubseteq_1)$, $(L_2,\sqsubseteq_2)$ to range over complete lattices and $x,y,z$ to range over their elements. We omit subscripts and order relations whenever clear from the context. As usual, $\bigsqcup$ and $\bigsqcap$ denote  least upper bound and  greatest lower bound, $\sqcup$ and $\sqcap$ denote join and meet, $\top$ and $\bot$ top and bottom.
Hereafter we will tacitly assume that all maps are monotone.
Obviously, the identity map $id\colon L\to L$ and the composition $f\comp g \colon L_1\to L_3$ of two monotone maps $g\colon L_1\to L_2$ and  $f\colon L_2\to L_3$ are monotone. For a map $f\colon L \to L$, we inductively define $f^0=id$ and $f^{n+1}=f\comp f^n$. Given $l \colon L_1 \to L_2$ and $r\colon L_2\to L_1$, we say that $l$ is the \emph{left adjoint} of $r$, or equivalently that $r$ is the \emph{right adjoint} of $l$, written $l\dashv r$, when it holds that $l(x)\sqsubseteq_2 y$ if{}f $x \sqsubseteq_1 r(y)$ for all $x\in L_1$ and $y\in L_2$.
Given a map $f\colon L\to L$, the element $x\in L$ is a \emph{post-fixed point}
if{}f  $x\sqsubseteq f(x)$, a \emph{pre-fixed point} if{}f $f(x)\sqsubseteq x$ and a \emph{fixed point} if{}f $x=f(x)$. Pre, post and fixed points form complete lattices: we write $\mu f$ and $\nu f$ for the least and greatest fixed point.

Several problems relevant to computer science can be reduced to check if $\mu b \sqsubseteq p$
for a monotone map $b\colon L \to L$ on a complete lattice $L$.
The Knaster-Tarski fixed-point theorem characterises $\mu b$ as the least upper bound of all pre-fixed points of $b$ and $\nu b$ as the
greatest lower bound of all its post-fixed points:
\begin{equation*}\label{eq:KNfpthm}
\mu b= \bigsqcap \{ x  \mid b(x) \sqsubseteq x \} \qquad \qquad \nu b= \bigsqcup \{ x  \mid x \sqsubseteq b(x) \}\enspace .
\end{equation*}
This immediately leads to two proof principles, illustrated below:
\begin{equation}\label{eq:coinductionproofprinciple}
\begin{array}{c}
    \exists x, \;  b(x) \sqsubseteq x \sqsubseteq p \\
    \hline
 \mu b\sqsubseteq p
\end{array}
\qquad
\qquad
\begin{array}{c}
    \exists x, \; i \sqsubseteq x\sqsubseteq b(x)\\
    \hline
    i \sqsubseteq \nu b
\end{array} \tag{KT}
\end{equation}
By means of \eqref{eq:coinductionproofprinciple}, one can prove $\mu b \sqsubseteq p$ by finding some pre-fixed point $x$, often called \emph{invariant}, such that $x \sqsubseteq p$.
However, automatically finding invariants might be rather complicated, so most of the algorithms rely on another fixed-point theorem, usually attributed to Kleene. %
It characterises $\mu b$ and $\nu b$ as the least upper bound and the greatest lower bound, of the \emph{initial} and \emph{final chains}:
\begin{align*}%
&\bot \sqsubseteq b(\bot) \sqsubseteq b^2(\bot) \sqsubseteq \cdots \quad \text{and}\quad
\cdots \sqsubseteq b^2(\top) \sqsubseteq b(\top) \sqsubseteq \top. \quad\text{That is,}
\\
& \mu b = \bigsqcup_{n\in \Nat} b^n(\bot), \qquad \qquad \nu b = \bigsqcap_{n\in \Nat} b^n(\top). 
\tag{Kl}
\label{eq:Kleenefpthm}
\end{align*}
The assumptions are stronger than for Knaster-Tarski: for the leftmost statement, it requires the map $b$ to be \emph{$\omega$-continuous} (i.e., it preserves $\bigsqcup$ of $\omega$-chains) and, for the rightmost  \emph{$\omega$-co-continuous} (similar but for  $\bigsqcap$). Observe that every left adjoint is continuous 
and every right adjoint is co-continuous (see e.g.~\cite{maclane:71}).

As explained in~\cite{KoriCAV22}, property directed reachability (PDR) algorithms~\cite{Bradley11} exploits \eqref{eq:coinductionproofprinciple} to try to prove the inequation and \eqref{eq:Kleenefpthm} to refute it.
In the algorithm we introduce in the next section, we further assume that $b$ is of the form $f \sqcup i$ for some element $i \in L$ and map $f\colon L \to L$, namely $b(x)= f(x) \sqcup i$ for all $x \in L$. Moreover we require $f$ to have a right adjoint $g\colon L \to L$. In this case %
\begin{equation}\label{eq:iff}
\mu (f \sqcup i) \sqsubseteq p \qquad \text{ if{}f } \qquad i \sqsubseteq \nu (g \sqcap p)
\end{equation}

\noindent
(which is easily shown using the Knaster-Tarski theorem)
 and $(f \sqcup i)$ and $(g \sqcap p)$ are guaranteed to be (co)continuous. Since $f \dashv g$ and left and right adjoints preserve, resp., arbitrary joins and meets, then for all $n \in \Nat$
\begin{equation}\label{eq:chainsadjoint}
\textstyle
(f\sqcup i)^{n} (\bot) = \bigsqcup_{j< n} f^j(i) \qquad  (g\sqcap p)^{n} (\top) = \bigsqcap_{j< n} g^j(p)
\end{equation}
which by \eqref{eq:Kleenefpthm} provide useful characterisations of least and greatest fixed points.
\begin{equation}\label{eq:kleeneadjoint}
\textstyle
\mu (f \sqcup i) = \bigsqcup_{n\in \Nat} f^n(i) \qquad \qquad \nu (g \sqcap p) = \bigsqcap_{n\in \Nat} g^n(p) \tag{Kl${\dashv}$}
\end{equation}

We conclude this section with an example that we will often revisit. It also provides a justification for the intuitive terminology that we sporadically use.

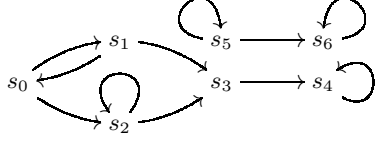
\begin{figure}[t]
\begin{displaymath}
    \xymatrix@R=5pt{
     & s_1 \ar@/^1ex/[rd] \ar@/^1ex/[ld] &  s_5 \ar@(l,u)\ar[r] & s_6 \ar@(r,u)\\
      s_0 \ar@/^1ex/[ur] \ar@/_1ex/[rd] & &s_3 \ar[r]& s_4\ar@(rd,ur) \\
      & s_2 \ar@(ur,ul) \ar@/_1ex/[ru] &
      }
   \end{displaymath}
\caption{The transition system of Example~\ref{eg:simple}, with $S = \{ s_0, \dots s_6 \}$ and $I=\{s_0\}$.}\label{fig:LTS}
\end{figure}

\begin{example}[Safety problem for transition systems] \label{eg:simple}
A \emph{transition system} consists of a triple $(S, I, \delta)$
where $S$ is a set of states, $I \subseteq S$ is a set of initial states, and $\delta\colon S \to \mathcal{P}S$ is a transition relation.
Here $\mathcal{P}S$ denotes the powerset of $S$, which forms a complete lattice ordered by inclusion $\subseteq$.
By defining $F\colon \mathcal{P}S \to \mathcal{P}S$ as $F(X) \defeq \bigcup_{s \in X} \delta(s)$ for all $X\in  \mathcal{P}S$,
one has that  $\mu (F \cup I)$ is the set of all states reachable from $I$.
    Therefore, for any $P \in  \mathcal{P}S$, representing some safety property,
    $\mu (F \cup I) \subseteq P$  holds iff
   all reachable states are safe.
    It is worth to remark that
   $F$ has a right adjoint $G\colon \mathcal{P}S \to \mathcal{P}S$ defined for all $X\in \mathcal{P}S$ as $G(X) \defeq \{s \mid \delta(s) \subseteq X\}$. Thus by~\eqref{eq:iff}, $\mu (F \cup I) \subseteq P$ iff $I \subseteq \nu (G \cap P)$.

Consider the transition system in Fig.~\ref{fig:LTS}.
Hereafter we write $S_{j}$ for the set of states $\{s_0, s_1, \dots, s_j\}$ and we fix the set of safe states to be $P = S_5$. It is immediate to see that $\mu(F \cup I)=S_4 \subseteq P$. Automatically, this can be checked with the initial chains of $(F \cup I)$ or with the final chain of $(G \cap P)$ displayed below on the left and on the right, respectively.

\[\emptyset \subseteq  I \subseteq S_2 \subseteq S_3 \subseteq S_4 \subseteq S_4 \subseteq \cdots \qquad \quad \qquad  
\cdots \subseteq S_4 \subseteq S_4 \subseteq P \subseteq S
\]

\noindent
The $(j+1)$-th element of the initial chain contains all the states that can be reached by $I$ in at most $j$ transitions, while $(j+1)$-th element of the final chain contains all the states that in at most $j$ transitions reach safe states only.
   \end{example}

\section{Adjoint PDR}\label{sec:APDR}

\begin{figure}[t]
{\scriptsize{
\begin{minipage}{.4\linewidth}
\begin{align}
\quad x_0 = \bot \tag{I0}\label{eq:x0bot} \\
1\leq k \leq n \tag{I1} \label{eq:invi} \\
\forall j\in[0, n-2]\text{, }x_j \sqsubseteq x_{j+1} \tag{I2}\label{eq:positivechain}
 \end{align}

\begin{align}
 i \sqsubseteq x_1  \tag{P1} \label{eq:Ix1}\\
 x_{n-2} \sqsubseteq p  \tag{P2}\label{eq:xP}\\
\forall j\in[0, n-2]\text{, }f(x_j) \sqsubseteq x_{j+1} \tag{P3}\label{eq:positiveF} \\
\forall j\in[0, n-2]\text{, }x_j \sqsubseteq g(x_{j+1}) \tag{P3a} \label{eq:positiveG}
\end{align}
\end{minipage}
\begin{minipage}{.6\linewidth}
\begin{align}
\text{If }\vec{y}\neq \varepsilon\text{ then }p \sqsubseteq y_{n-1}  \tag{N1}\label{eq:Pepsilon} \\
\forall j\in[k,n-2]\text{, }g(y_{j+1}) \sqsubseteq y_j  \tag{N2}\label{eq:negativeG}
\end{align}

\begin{align}
\forall j \in [k, n - 1] \text{, } x_j \not\sqsubseteq y_j \tag{PN} \label{eq:positivenegative}\\
\forall j \in [0, n-1] \text{, } (f \sqcup i)^j (\bot) \sqsubseteq x_j \sqsubseteq (g \sqcap p)^{n-1-j} (\top) \tag{A1} \label{eq:positiveinitialfinal}\\
\forall j \in [1, n-1] \text{, } x_{j-1} \sqsubseteq g^{n-1-j}(p) \tag{A2} \label{positivefinal}\\
\forall j\in[k,n-1]\text{, }g^{n-1-j}(p) \sqsubseteq y_j \tag{A3} \label{negativefinal}
\end{align}
\end{minipage}
}}
\caption{Invariants of {\APDR}.}
\label{fig:invariants}
\end{figure}

In this section we present {\APDR}, an algorithm that
takes in input a tuple $(i,f,g,p)$ with $i,p\in L$ and $f\dashv g \colon L\to L$ and, if it terminates, it returns  true whenever $\mu (f \sqcup i) \sqsubseteq p$ and false otherwise.

\smallskip

The algorithm manipulates two sequences of elements of $L$:
$\vec{x} \defeq  x_0, \dots, x_{n-1}$  of length $n$
and $\vec{y}\defeq  y_k, \dots y_{n-1} $ of length $n-k$. These satisfy, through the executions
of {\APDR}, the invariants in Fig.~\ref{fig:invariants}. Observe that, by \eqref{eq:positiveinitialfinal}, $x_j$ over-approximates the $j$-th element of the initial chain, namely $(f\sqcup i)^j(\bot) \sqsubseteq x_j$, while, by \eqref{negativefinal}, the $j$-indexed element $y_j$ of $\vec{y}$ over-approximates $g^{n-j-1}(p)$ that, borrowing the terminology of Example~\ref{eg:simple}, is the set of states which are safe in $n-j-1$ transitions. %
Moreover, by~\eqref{eq:positivenegative},
the element $y_j$ witnesses that $x_j$ is unsafe,
 i.e., that
$x_j \not\sqsubseteq g^{n-1-j}(p)$ or equivalently %
$f^{n-j-1}(x_j) \not\sqsubseteq p$.
Notably, $\vec{x}$ is a positive chain and $\vec{y}$ a negative sequence, according to the definitions below.

  \begin{definition}[positive chain] \label{def:posi_seq}
    A \emph{positive chain} for %
    $\mu (f \sqcup i) \sqsubseteq p$
    is a finite chain $x_0 \sqsubseteq \dots \sqsubseteq x_{n-1}$ in $L$
    of length $n \geq 2$
    which satisfies \eqref{eq:Ix1}, \eqref{eq:xP}, \eqref{eq:positiveF} in Fig.~\ref{fig:invariants}.
It is \emph{conclusive} if $x_{j+1} \sqsubseteq x_j$ for some $j \leq n-2$.
 \end{definition}

In a conclusive positive chain, $x_{j+1}$ provides an invariant for $f\sqcup i$ and thus, by \eqref{eq:coinductionproofprinciple}, $\mu (f \sqcup i) \sqsubseteq p$ holds. So, when $\vec{x}$ is conclusive, {\APDR} returns true.

\begin{definition}[negative sequence] \label{def:neg_seq}
    A \emph{negative sequence} for %
    $\mu (f \sqcup i) \sqsubseteq p$  is a finite sequence $ y_k, \dots, y_{n-1}$ in $L$
    with  $1 \leq k \leq n$
    which satisfies \eqref{eq:Pepsilon} and \eqref{eq:negativeG} in Fig.~\ref{fig:invariants}.
    It is \emph{conclusive} if $k=1$ and $i \not \sqsubseteq y_1$.
  \end{definition}

When $\vec{y}$ is conclusive, {\APDR} returns false as $y_1$ provides a counterexample: \eqref{eq:Pepsilon} and \eqref{eq:negativeG} entail \eqref{negativefinal} and thus $i \not \sqsubseteq y_1\sqsupseteq g^{n-2}(p)$. By~\eqref{eq:kleeneadjoint},  $g^{n-2}(p) \sqsupseteq  \nu (g \sqcap p)$ and thus $i \not \sqsubseteq \nu (g \sqcap p)$. By~\eqref{eq:iff}, $\mu (f \sqcup i) \not \sqsubseteq p$.

\begin{figure}[t]
\begin{lrbox}{\lstbox}\begin{minipage}{\textwidth}
\centering
\underline{{\APDR} $(i,f,g,p)$}
\begin{codeNT}
<INITIALISATION>
  $( \vec{x} \| \vec{y} )_{n,k}$ := $(\bot,\top\|\varepsilon)_{2,2}$
<ITERATION>						         % $\vec{x},\vec{y}$  not conclusive
  case $( \vec{x} \| \vec{y} )_{n,k}$ of
	   $\vec{y}=\varepsilon$ And $x_{n-1} \sqsubseteq p$     :                    %(Unfold)
			$( \vec{x} \| \vec{y} )_{n,k}$ := $( \vec{x}, \top \| \varepsilon )_{n+1,n+1}$
	   $\vec{y}=\varepsilon$ And $x_{n-1} \not \sqsubseteq p$    :                     %(Candidate)
			choose $z\in L$ st  $x_{n-1} \not \sqsubseteq z$ And  $p \sqsubseteq z$;
			$( \vec{x} \| \vec{y} )_{n,k}$ := $( \vec{x} \| z )_{n,n-1}$
	   $\vec{y} \neq \varepsilon$ And $f(x_{k-1}) \not \sqsubseteq y_k$ :                        %(Decide)
			choose $z \in L$ st $x_{k-1} \not \sqsubseteq z$ And $g(y_k) \sqsubseteq z$;
			$(\vec{x} \| \vec{y} )_{n,k}$ := $(\vec{x} \| z , \vec{y} )_{n,k-1}$
	   $\vec{y} \neq \varepsilon$ And $f(x_{k-1}) \sqsubseteq y_k$ :                        %(Conflict)
			choose $z \in L$ st $z \sqsubseteq y_k$ And $(f \sqcup i)(x_{k-1} \sqcap z) \sqsubseteq z$;
			$(\vec{x} \| \vec{y} )_{n,k}$ := $(\vec{x} \sqcap_k z \| \mathsf{tail}(\vec{y}) )_{n,k+1}$
  endcase
<TERMINATION>
	if $\exists j\in [0,n-2]\,.\, x_{j+1} \sqsubseteq x_j$ then return true		 % $\vec{x}$ conclusive
	if $i \not \sqsubseteq y_1$ then return false							% $\vec{y}$ conclusive
\end{codeNT}
\end{minipage}\end{lrbox}
\centering
\scalebox{.8}{\usebox\lstbox}
\caption{{\APDR} algorithm checking $\mu (f \sqcup i) \sqsubseteq p$.}\label{fig:naive}
\end{figure}

The pseudocode of the algorithm is displayed in Fig.~\ref{fig:naive}, where we write $( \vec{x} \| \vec{y} )_{n,k}$ to compactly represents the state of the algorithm:
the pair $(n,k)$ is called the \emph{index} of the state, with $\vec{x}$ of length $n$ and $\vec{y}$ of length $n-k$.
When $k=n$, $\vec{y}$ is the empty sequence $\varepsilon$. For any $z\in L$, we write $ \vec{x},z$ for the chain $ x_0, \dots, x_{n-1}, z$ of length $n+1$ and $ z,\vec{y}$ for the sequence $ z,y_k, \dots y_{n-1}$ of length $n-(k-1)$.
Moreover, we write $\vec{x}\sqcap_j z$  for the chain
$ x_0 \sqcap z, \dots, x_j \sqcap z, x_{j+1},\dots , x_{n-1}$. Finally, $\mathsf{tail}(\vec{y})$ stands for the tail of $\vec{y}$, namely $y_{k+1}, \dots y_{n-1}$ of length $n-(k+1)$.

The algorithm starts in the initial state $s_0\defeq( \bot, \top \| \varepsilon )_{2,2}$ and, unless one of $\vec{x}$ and $\vec{y}$ is conclusive, iteratively applies one of the four mutually exclusive rules: (Unfold), (Candidate), (Decide) and (Conflict). %
The rule (Unfold) extends the positive chain by one element when the
negative sequence is empty and the positive chain is under $p$;
since the element introduced by (Unfold) is $\top$, its application
typically triggers rule (Candidate) that starts the negative sequence
with an over-approximation of $p$.
Recall that the role of $y_j$ is to witness that $x_j$ is unsafe.
After (Candidate) either (Decide) or (Conflict) are possible:
if $y_k$ witnesses that, besides $x_k$, also $f(x_{k-1})$ is unsafe, then (Decide) is used to
further
extend the negative sequence to witness that $x_{k-1}$ is
 unsafe; otherwise, the rule (Conflict) improves the precision
of the positive chain in such a way that $y_k$ no longer witnesses
$x_k\sqcap z$ unsafe and, thus, the negative sequence is
shortened.

Note that, in (Candidate), (Decide) and (Conflict), the element $z\in L$ is chosen among a set of possibilities, thus {\APDR} is nondeterministic.

To illustrate the executions of the algorithm, we adopt a labeled transition system notation. Let $\states \defeq \{( \vec{x} \| \vec{y} )_{n,k} \mid n \geq 2 \text{, }k\leq n\text{, }  \vec{x}\in L^n \text{ and } \vec{y}\in L^{n-k}\}$ be the set of all possible states of {\APDR}.
We call $( \vec{x} \| \vec{y} )_{n,k} \in \states$ \emph{conclusive} if $\vec{x}$ or
$\vec{y}$ are such.
When $s \in \states$ is not conclusive, we write $s \trz{D}{}$ to mean that
$s$
satisfies the guards in the rule (Decide), and $s \trz{D}{z} s'$ to mean
that, being (Decide) applicable, {\APDR} moves from
state $s$ to $s'$ by choosing $z$.
Similarly for the other rules: the labels $\mathit{Ca}$, $\mathit{Co}$
and $U$ stands for (Candidate), (Conflict) and (Unfold),
respectively. %
When irrelevant we omit to specify labels and choices and we just write $s \tr{} s'$.
As usual %
$\ttp{}$ stands for the transitive closure of $\tr{}$ while $\ttr{}$ stands for the reflexive and transitive closure of $\tr{}$.

\begin{example}\label{ex:simple-ts}
Consider the safety problem in Example~\ref{eg:simple}. Below we illustrate two possible computations of {\APDR} that differ for the choice of $z$ in (Conflict).
The first run is conveniently represented as the following series of transitions.
\begin{lrbox}{\lstbox}\begin{minipage}{\textwidth}
{\footnotesize \begin{align*}
		&( \emptyset, S \| \varepsilon )_{2,2}
		\tr{\mathit{Ca}}_{P} ( \emptyset, S \| P )_{2,1}
		\tr{\mathit{Co}}_{I} ( \emptyset, I \| \varepsilon )_{2,2}
		\tr{U} ( \emptyset, I, S \| \varepsilon )_{3,3}
		\tr{\mathit{Ca}}_{P} ( \emptyset, I, S \| P )_{3,2} \\[-.4em]
		\tr{\mathit{Co}}_{S_2} & ( \emptyset, I, S_2 \| \varepsilon )_{3,3}
		\tr{U} %
		\tr{\mathit{Ca}}_{P}  ( \emptyset, I, S_2, S \| P )_{4,3}
		\tr{\mathit{Co}}_{S_3} ( \emptyset, I, S_2, S_3 \| \varepsilon )_{4,4}
		\tr{U} %
		\tr{\mathit{Ca}}_{P}  ( \emptyset, I, S_2, S_3, S \| P )_{5,4}\\[-.4em]
		\tr{\mathit{Co}}_{S_4} & ( \emptyset, I, S_2, S_3, S_4 \| \varepsilon )_{5,5}
		\tr{U} %
		\tr{\mathit{Ca}}_{P} ( \emptyset, I, S_2, S_3, S_4, S \| P )_{6,5}
		\tr{\mathit{Co}}_{S_4} ( \emptyset, I, S_2, S_3, S_4, S_4 \| \varepsilon )_{6,6}
	\end{align*}}%
\end{minipage}\end{lrbox}

\[\scalebox{.93}{\usebox\lstbox}\]

\noindent
The last state returns true since $x_4 = x_5=S_4$.  Observe that the elements of $\vec{x}$, with the exception of the last element $x_{n-1}$, are those of the initial chain of $(F \cup I)$, namely, $x_j$ is the set of states reachable in at most $j-1$ steps. In the second computation, the elements of $\vec{x}$ are roughly  those of the final chain of $(G \cap P)$. More precisely, after (Unfold) or (Candidate), $x_{n-j}$ for $j<n-1$ is the set of states which only reach safe states within $j$ steps.
\begin{lrbox}{\lstbox}\begin{minipage}{\textwidth}
{\footnotesize \begin{align*}
		&( \emptyset, S \| \varepsilon )_{2,2}
		\tr{\mathit{Ca}}_{P} ( \emptyset, S \| P )_{2,1}
		\tr{\mathit{Co}}_{P} ( \emptyset, P \| \varepsilon )_{2,2} \\[-.4em]
		\tr{U} %
		\tr{\mathit{Ca}}_{P} & ( \emptyset, P, S \| P )_{3,2}
		\tr{D}_{S_4} ( \emptyset, P, S \| S_4, P )_{3,1}
		\tr{\mathit{Co}}_{S_4} ( \emptyset, S_4, S \| P )_{3,2}
		\tr{\mathit{Co}}_{P} ( \emptyset, S_4, P \| \varepsilon )_{3,3} \\[-.4em]
		\tr{U} %
		\tr{\mathit{Ca}}_{P} & ( \emptyset, S_4, P, S \| P )_{4,3}
		\tr{D}_{S_4} ( \emptyset, S_4, P, S \| S_4, P )_{4,2}
		\tr{\mathit{Co}}_{S_4} ( \emptyset, S_4, S_4, S \| P )_{4,3}
	\end{align*}}%
\end{minipage}\end{lrbox}

\[\scalebox{.93}{\usebox\lstbox}\]

\noindent
Observe that, by invariant \eqref{eq:positiveinitialfinal}, the values of $\vec{x}$ in the two runs are, respectively, the least and the greatest values for all possible computations of {\APDR}.
\end{example}

Theorem~\ref{th:soundness}.1 follows by invariants \eqref{eq:positivechain}, \eqref{eq:Ix1}, \eqref{eq:positiveF} and \eqref{eq:coinductionproofprinciple}; %
Theorem~\ref{th:soundness}.2 by \eqref{eq:Pepsilon}, \eqref{eq:negativeG} and \eqref{eq:kleeneadjoint}. Note that both results hold for any choice of $z$.

\begin{theorem}[Soundness]\label{th:soundness}
\emph {\APDR} is sound. Namely,
\begin{enumerate}%
\item If \emph{\APDR} returns true then $\mu (f \sqcup i) \sqsubseteq p$.
\item If \emph{\APDR} returns false then $\mu (f \sqcup i) \not \sqsubseteq p$.
\end{enumerate}
\end{theorem}

\subsection{Progression}\label{sec:progression}

It is necessary to prove that in any step of the execution, if the algorithm does not return true or false, then it can progress to a new state, not yet visited. To this aim we must deal with the subtleties of the non-deterministic choice of the element $z$ in (Candidate), (Decide) and (Conflict). The following proposition ensures that, for any of these three rules, there is always a possible choice.

\begin{proposition}[Canonical choices]\label{prop:CanonicalChoice}
The following are always possible:

\noindent \begin{minipage}{.5\linewidth}
\begin{enumerate}\setcounter{enumi}{0}
\item in (Candidate) $z=p$;
\item in (Decide) $z= g(y_k)$;
\end{enumerate}
\end{minipage}
\begin{minipage}{.5\linewidth}
\begin{enumerate}\setcounter{enumi}{2}
\item in (Conflict) $z = y_k$;
\item in (Conflict) $z = (f \sqcup i)(x_{k-1})$.
\end{enumerate}
\end{minipage}
Thus, for all non-conclusive $s\in \states$, if $s_0 \ttr{} s $ then $s \tr{}$.
\end{proposition}
Then, Proposition~\ref{prop:progres} ensures that {\APDR} always traverses new states.
\begin{proposition}[Impossibility of loops]\label{prop:progres}
	If $s_0 \ttr{} s \ttp{ } s'$, then $s\neq  s'$.
\end{proposition}

Observe that the above propositions  entail that {\APDR} terminates whenever the lattice $L$ is finite, since the set of reachable states is finite in this case.

\begin{example}
For $(I,F,G,P)$ as in Example \ref{eg:simple}, {\APDR} behaves essentially as IC3/PDR~\cite{Bradley11}, solving reachability problems for transition systems with finite state space $S$. Since
the lattice $\mathcal{P}S$ is also finite, %
{\APDR} always terminates.
\end{example}

\subsection{Heuristics}\label{sec:heuristics}

The nondeterministic choices of the algorithm can be resolved by using heuristics. Intuitively, a heuristic chooses for any states $s\in\states$ an element $z\in L$ to be possibly used in (Candidate), (Decide) or (Conflict), so it is just a function $h\colon \states \to L$.
When defining a heuristic, we will avoid to specify its values on conclusive states or in those performing (Unfold), as they are clearly irrelevant.

With a heuristic, one can instantiate {\APDR} by making the choice of $z$ as prescribed by $h$. Syntactically, this means to erase  from the code of Fig.~\ref{fig:naive} the three lines of \texttt{choose} and replace them by
$z \texttt{:= } h(\,( \vec{x} \| \vec{c} )_{n,k}\,)$. We call {\APDR}$_h$ the resulting deterministic algorithm and write $s \Htrz{}{h}{} s'$ to mean that {\APDR}$_h$ moves from state $s$ to $s'$. We let $\states^h\defeq \{s\in \states \mid s_0\Htrz{}{h}{*} s\}$ be the sets of all states reachable by {\APDR}$_h$.

\begin{definition}[legit heuristic]
A heuristic $h\colon \states \to L$ is called \emph{legit} whenever for all $s,s'\in \states^h$,
if $s \Htrz{}{h}{}s'$ then $s\tr{}s'$. %
\end{definition}
When $h$ is legit, the only execution of the deterministic algorithm {\APDR}$_h$ is one of the possible executions of the non-deterministic algorithm {\APDR}.

The canonical choices provide two legit heuristics:
first, we call \emph{simple} any legit heuristic $h$ that chooses $z$ in (Candidate) and (Decide) as in Proposition \ref{prop:CanonicalChoice}: %
\begin{equation}\label{eq:simple}
( \vec{x} \| \vec{y} )_{n,k} \mapsto
\begin{cases*}
		p & if $( \vec{x} \| \vec{y} )_{n,k} \tr{ \mathit{Ca} }$ \\
		g(y_k)  & if $( \vec{x} \| \vec{y} )_{n,k} \tr{ D }$
\end{cases*}
\end{equation}
Then, if the choice in (Conflict) is like in Proposition \ref{prop:CanonicalChoice}.4, we call $h$ \emph{initial}; if it is like in Proposition \ref{prop:CanonicalChoice}.3, %
we call $h$ \emph{final}. Shortly, the two legit heuristics are: %
\[
\begin{array}{r|ll}
\quad\emph{simple initial} \quad\
& 	\quad\eqref{eq:simple} \text{ and }( \vec{x} \| \vec{y} )_{n,k} \mapsto
	(f\sqcup i)(x_{k-1})   &\quad\mbox{if $( \vec{x} \| \vec{y} )_{n,k} \in \mathit{Co}$}\quad
\\[5pt]
\hline
\\[-7pt]
\quad\emph{simple final} \quad\
& \quad\eqref{eq:simple} \text{ and }
	( \vec{x} \| \vec{y} )_{n,k} \mapsto
	y_k   &\quad\mbox{if $( \vec{x} \| \vec{y} )_{n,k} \in \mathit{Co}$}\quad\\
\end{array}
\]
Interestingly, with any simple heuristic, the sequence $\vec{y}$ takes a familiar shape:
\begin{proposition}\label{prop:negativesequencefinalchain}
Let $h\colon \states \to L$ be any simple heuristic. For all $( \vec{x} \| \vec{y} )_{n,k} \in \states^h$, invariant~\eqref{negativefinal} holds as an equality, namely  for all $j\in[k,n-1]$,
$y_j=g^{n-1-j}(p)$.
\end{proposition}
By the above proposition and~\eqref{negativefinal}, the negative sequence $\vec{y}$ occurring in the execution of {\APDR}$_h$, for a simple heuristic $h$, is the least amongst all the negative sequences occurring in any execution of {\APDR}.

Instead, %
 invariant \eqref{eq:positiveinitialfinal} informs us that the positive chain $\vec{x}$ is always in between the initial chain of $f\sqcup i$ and the final chain of $g \sqcap p$. Such values of $\vec{x}$ are obtained  by, respectively, simple initial and simple final heuristic.

\begin{example}
Consider the two runs of {\APDR} in Example~\ref{ex:simple-ts}. The first one exploits the simple initial heuristic and indeed, the positive chain $\vec{x}$ coincides with the initial chain.
Analogously, the second run uses the simple final heuristic.
\end{example}

\subsection{Negative Termination}\label{sec:termination}

When the lattice $L$ is not finite, {\APDR} may not terminate, since checking $\mu (f\sqcup i) \sqsubseteq p$ is not always decidable. In this section, we show that the use of certain heuristics can guarantee termination whenever $\mu (f \sqcup i) \not \sqsubseteq p$. %

The key insight is the following: if $\mu (f \sqcup i) \not \sqsubseteq p$ then by~\eqref{eq:Kleenefpthm}, there should exist some $\tilde{n}\in \Nat$ such that $(f \sqcup i)^{\tilde{n}} (\bot) \not \sqsubseteq p$. By \eqref{eq:positiveinitialfinal}, the rule (Unfold) can be applied only when $(f \sqcup i)^{n-1} (\bot) \sqsubseteq x_{n-1} \sqsubseteq p$. Since (Unfold) increases $n$ and $n$ is never decreased by other rules, then (Unfold) can be applied at most $\tilde{n}$ times.

The elements of negative sequences are introduced by rules (Candidate) and (Decide).
If we guarantee that for any index $(n,k)$ the heuristic in such cases returns a finite number of values for $z$, then one can prove termination. %
To make this formal, we fix
$\mathit{CaD}^h_{n,k} \defeq \{ ( \vec{x} \| \vec{y} )_{n,k}\in \states^h \mid ( \vec{x} \| \vec{y} )_{n,k}\tr{\mathit{Ca}} \text{ or } ( \vec{x} \| \vec{y} )_{n,k}\tr{D}\}$, i.e., the set of all $(n,k)$-indexed  states reachable by {\APDR}$_h$ that trigger (Candidate) or (Decide), and $h(\mathit{CaD}^h_{n,k})\defeq \{h(s) \mid s\in \mathit{CaD}^h_{n,k}\}$, i.e., the set of all possible values returned by $h$ in such states. %

\begin{theorem}[Negative termination]\label{thm:negativetermination}
Let $h$ be a legit heuristic. If $h(\mathit{CaD}^h_{n,k})$ is finite for all $n,k$ and $\mu(f\sqcup i) \not \sqsubseteq p$, then \emph{\APDR}$_h$ terminates.
\end{theorem}

\begin{corollary}\label{cor:negativetermiantion}
	Let $h$ be a simple heuristic.
	If  $\mu(f\sqcup i) \not \sqsubseteq p$, then \emph{\APDR}$_h$ terminates.
\end{corollary}

Note that this corollary ensures negative termination whenever we use the canonical choices in (Candidate) and (Decide) \emph{irrespective of the choice for}  (Conflict), therefore it holds for both simple initial and simple final heuristics.

\section{Recovering Adjoints with Lower Sets}\label{sec:downset}

In the previous section, we have introduced an algorithm for checking $\mu b \sqsubseteq p$ whenever $b$ is of the form $f\sqcup i$ for an element $i\in L$ and a left-adjoint $f\colon L \to L$. This, unfortunately, is not the case for several interesting problems, like the max reachability problem~\cite{BaierK} that we will illustrate in Section~\ref{sec:MDP}.

The next result informs us that, under standard assumptions, one can transfer the problem of checking $\mu b \sqsubseteq p$
to lower sets, where adjoints can always be defined.
Recall that, for a lattice $(L,\sqsubseteq)$, a \emph{lower set} is a subset $X\subseteq L$ such that if $x\in X$ and $x'\sqsubseteq x$ then $x'\in X$; the set of lower sets of $L$ forms a complete lattice $(L^\downarrow, \subseteq)$
with joins and meets given by  union and intersection; as expected $\bot$ is $\emptyset$ and $\top$ is $L$.
Given $b\colon L\to L$, one can define two functions $b^\downarrow, b^\downarrow_r \colon L^\downarrow \to L^\downarrow$ as $b^\downarrow(X) \defeq b(X)^\downarrow$
and $b^\downarrow_r(X) \defeq \{x \mid b(x) \in X\}$. It holds that $b^\downarrow\, \dashv\, b^\downarrow_r$.

  \begin{equation}\label{eq:lowersetadjunction}
    \xymatrix{
      (L, \sqsubseteq) \lloop{b} \ar@/_1.5ex/[r]_-{(-)^\downarrow}^-\bot
      &(L^\downarrow, \subseteq) \rloop{b^\downarrow\, \dashv\, b^\downarrow_r} \ar@/_1.5ex/[l]_-{\bigsqcup}
    }
  \end{equation}
In the diagram above, $(-)^\downarrow\colon x \mapsto \{x' \mid x' \sqsubseteq x\}$
  and $\bigsqcup \colon L^\downarrow \to L$ maps a lower set $X$ into $\bigsqcup \{x\mid x\in X\}$. The maps $\bigsqcup$ and $(-)^\downarrow$ form a \emph{Galois insertion}, namely $\bigsqcup \dashv (-)^\downarrow$ and $\bigsqcup (-)^\downarrow = id$, and thus one can think of~\eqref{eq:lowersetadjunction} in terms of \emph{abstract interpretation} \cite{cousot77,cousot21}: $L^\downarrow$ represents the concrete domain, $L$ the abstract domain and $b$ is a sound abstraction of $b^\downarrow$. Most importantly, it turns out that $b$ is  \emph{forward-complete}~\cite{GRS00,BonchiGGP18} w.r.t. $b^\downarrow$, namely the following equation holds.
\begin{equation}\label{eq:EMlaw}
(-)^\downarrow \circ b = b^\downarrow \circ (-)^\downarrow
\end{equation}

\begin{proposition} \label{prop:prob_down_up}
Let $(L,\sqsubseteq)$ be a complete lattice, $p\in L$ and $b \colon L \to L$ be a $\omega$-continuous map. Then $\mu b \sqsubseteq p$ iff $\mu (b^\downarrow \cup \bot^\downarrow) \subseteq p^\downarrow$.
\end{proposition}

By means of Proposition~\ref{prop:prob_down_up}, we can thus solve $\mu b \sqsubseteq p$ in $L$ by running {\APDR} on $(\bot^\downarrow, b^\downarrow,b_r^{\downarrow}, p^{\downarrow})$.
Hereafter, we tacitly assume that $b$ is $\omega$-continuous.

\subsection{{\ADPDR}: Positive Chain in $L$, Negative Sequence in $L^\downarrow$}\label{ssec:ADPDR}

\begin{figure}[t]
\begin{lrbox}{\lstbox}\begin{minipage}{\textwidth}
\centering
\underline{{\ADPDR} $(b,p)$}
\begin{codeNT}
<INITIALISATION>
  $( \vec{x} \| \vec{Y} )_{n,k}$ := $(\emptyset,\bot,\top\|\varepsilon)_{3,3}$
<ITERATION>
  case $( \vec{x} \| \vec{Y} )_{n,k}$ of								% $\vec{x}, \vec{Y}$ not conclusive
	   $\vec{Y}=\varepsilon$ And $x_{n-1} \sqsubseteq p$     :                    %(Unfold)
			$( \vec{x} \| \vec{Y} )_{n,k}$ := $( \vec{x}, \top \| \varepsilon )_{n+1,n+1}$
	   $\vec{Y}=\varepsilon$ And $x_{n-1} \not \sqsubseteq p$    :                     %(Candidate)
			choose $Z\in L^{\downarrow}$ st  $x_{n-1} \not \in Z$ And  $p \in Z$;
			$( \vec{x} \| \vec{Y} )_{n,k}$ := $( \vec{x} \| Z )_{n,n-1}$
	   $\vec{Y} \neq \varepsilon$ And $b(x_{k-1}) \not \in Y_k$ :                        %(Decide)
			choose $Z\in L^{\downarrow}$ st $x_{k-1} \not \in Z$ And $b^{\downarrow}_r(Y_k) \subseteq Z$;
			$(\vec{x} \| \vec{Y} )_{n,k}$ := $(\vec{x} \| Z , \vec{Y} )_{n,k-1}$
	   $\vec{Y} \neq \varepsilon$ And $b(x_{k-1}) \in Y_k$ :                        %(Conflict)
			choose $z \in L$ st $z \in Y_k$ And $b(x_{k-1} \sqcap z) \sqsubseteq z$;
			$(\vec{x} \| \vec{Y} )_{n,k}$ := $(\vec{x} \sqcap_k z \| \mathsf{tail}(\vec{Y}) )_{n,k+1}$
  endcase
<TERMINATION>
	if $\exists j\in [0,n-2]\,.\, x_{j+1} \sqsubseteq x_j$ then return true		% $\vec{x}$ conclusive
	if $Y_1=\emptyset$ then return false							% $\vec{Y}$ conclusive
\end{codeNT}
\end{minipage}\end{lrbox}
\centering
\scalebox{.8}{\usebox\lstbox}
\caption{The algorithm {\ADPDR} for checking $\mu b \sqsubseteq p$: the elements of negative sequence are in $L^\downarrow$, while those of the positive chain are in $L$, with the only exception of $x_0$ which is constantly the bottom lower set $\emptyset$. For $x_0$, we fix $b(x_0) = \bot$.}
\label{fig:downclosed}
\end{figure}

While {\APDR} on $(\bot^\downarrow, b^\downarrow,b_r^{\downarrow}, p^{\downarrow})$ might be computationally expensive, it is the first step toward the definition of an efficient algorithm
that exploits a convenient form of the positive chain.

A lower set $X\in L^{\downarrow}$ is said to be a \emph{principal} if $X=x^\downarrow$ for some $x\in L$. Observe that the top of the lattice $(L^\downarrow, \subseteq)$ is a principal, namely $\top^\downarrow$, and that the meet (intersection) of two principals $x^\downarrow$ and $y^\downarrow$ is the principal $(x\sqcap y)^\downarrow$.

Suppose now that, in (Conflict), {\APDR}$(\bot^\downarrow, b^\downarrow,b_r^{\downarrow}, p^{\downarrow})$ always chooses principals rather than arbitrary lower sets.
This suffices to guarantee that all the elements of $\vec{x}$ are principals (with the only exception of $x_0$ which is constantly the bottom element of $L^\downarrow$ that, note, is $\emptyset$ and not $\bot^\downarrow$). In fact, the elements of $\vec{x}$ are all obtained by (Unfold), that adds the principal $\top^\downarrow$, and by (Conflict), that  takes their meets with the chosen principal.

Since principals are in bijective correspondence with the elements of $L$, by imposing to {\APDR}$(\bot^\downarrow, b^\downarrow,b_r^{\downarrow}, p^{\downarrow})$ to choose  a principal in (Conflict), we obtain an algorithm, named {\ADPDR}, where the elements of the positive chain are drawn from $L$, while the negative sequence is taken in $L^{\downarrow}$. The algorithm is reported in Fig.~\ref{fig:downclosed} where we use the notation $( \vec{x} \| \vec{Y} )_{n,k}$ to emphasize that the elements of the negative sequence are lower sets of elements in $L$.

All definitions and results illustrated in Section \ref{sec:APDR} for {\APDR}  are inherited\footnote{Up to a suitable renaming: the domain is $(L^\downarrow, \subseteq)$ instead of $(L,\sqsubseteq)$, the parameters are $\bot^\downarrow, b^\downarrow,b_r^{\downarrow}, p^{\downarrow}$ instead of $i, f,g, p$ and the negative sequence is $\vec{Y}$ instead of $\vec{y}$.} by {\ADPDR}, with the only exception of Proposition \ref{prop:CanonicalChoice}.3. The latter does not hold, as it prescribes a choice for (Conflict) that may not be a principal. In contrast, the choice in Proposition \ref{prop:CanonicalChoice}.4 is, thanks to \eqref{eq:EMlaw}, a principal. This means in particular that the simple initial heuristic is always applicable. %

\begin{theorem}\label{th:ADPDR}
All results in Section \ref{sec:APDR},  but Prop.~\ref{prop:CanonicalChoice}.3, hold for \emph{\ADPDR}.
\end{theorem}

\subsection{{\ADPDR} simulates LT-PDR}\label{ssec:LTPDRvsADPDR}
The closest approach to {\APDR} and {\ADPDR} is the lattice-theoretic extension of the original PDR, called LT-PDR~\cite{KoriCAV22}. While these algorithms exploit essentially the same positive chain to find an invariant, the main difference lies in the sequence used to witness the existence of some counterexamples. %

\begin{definition}[Kleene sequence, from~\cite{KoriCAV22}]
A sequence $\vec{c}= c_k,\dots, c_{n-1}$ of elements of $L$ is a \emph{Kleene sequence}
if the conditions \emph{(C1)} and \emph{(C2)} below hold.
It is \emph{conclusive} if also condition \emph{(C0)} holds.
\[
\emph{(C0) } c_1 \sqsubseteq b(\bot),
\qquad
\emph{(C1) } c_{n-1} \not \sqsubseteq p,
\qquad
\emph{(C2) } \forall j\in[k,n-2].~c_{j+1} \sqsubseteq b(c_j)\text{.}
\]
\end{definition}

LT-PDR tries to construct an under-approximation $c_{n-1}$ of $b^{n-2}(\bot)$ that violates the property $p$. The Kleene sequence is constructed by trial and error, starting by some arbitrary choice of $c_{n-1}$.

{\APDR} crucially differs from LT-PDR in the search for counterexamples: LT-PDR under-approximates the final chain while {\APDR} over\hyp{}approximates it. The algorithms are thus incomparable.
However, we can draw a formal correspondence between {\ADPDR} and LT-PDR by showing that  {\ADPDR} simulates LT-PDR, but cannot be simulated by LT-PDR.
In fact,
{\ADPDR} exploits the existence of the adjoint to start from an over\hyp{}approximation $Y_{n-1}$ of $p^\downarrow$ and computes backward an over-approximation of the set of  safe states.
 Thus, the key difference comes from the strategy to look for a counterexample:
  to prove $\mu b \not \sqsubseteq p$,
  {\ADPDR} tries to find $Y_{n-1}$ satisfying $p \in Y_{n-1}$ and $\mu b \not \in Y_{n-1}$
  while
  LT-PDR tries to find $c_{n-1}$ s.t. $c_{n-1} \not \sqsubseteq p$ and $c_{n-1} \sqsubseteq \mu b$.

\medskip

Theorem~\ref{th:LT-PDR-instance-ADPDR} below states that any execution of LT-PDR can be mimicked by {\ADPDR}.
The proof exploits a map from LT-PDR's Kleene sequences  $\vec{c}$ to {\ADPDR}'s negative sequences $\negation{c}$ of a particular form.
Let $(L^{\uparrow}, \supseteq)$  be the complete lattice of upper sets, namely subsets $X \subseteq L$ such that $X=X^\uparrow \defeq \{x'\in L \mid \exists x\in X \,. \, x\sqsubseteq x'\}$.
There is an isomorphism $\neg \colon {(L^\uparrow, \supseteq)} \stackrel{\cong}{\longleftrightarrow} (L^\downarrow, \subseteq)$ mapping each $X\subseteq S$ into its complement.
For a Kleene sequence $\vec{c} = c_k,\dots, c_{n-1}$ of LT-PDR,  the sequence
$\negation{c} \defeq \lnot (\{ c_k \}^{\uparrow}), \dots, \lnot (\{ c_{n-1} \}^{\uparrow})$
is a negative sequence, in the sense of Definition \ref{def:neg_seq}, for {\ADPDR}. Most importantly, the assignment $\vec{c} \mapsto \negation{c}$ extends to a function, from the states of LT-PDR to those of {\ADPDR}, that is proved to be a \emph{strong simulation} \cite{Mil89}.

\begin{theorem}%
\label{th:LT-PDR-instance-ADPDR}
	\emph{\ADPDR} simulates LT-PDR.
\end{theorem}

Remarkably, {\ADPDR}'s negative sequences are not limited to the images of LT-PDR's Kleene sequences: they are more general  than the complement of the upper closure of a singleton.
In fact, a single negative sequence of {\ADPDR} can represent \emph{multiple} Kleene sequences of LT-PDR at once.
Intuitively, this means that a single execution of {\ADPDR} can correspond to multiple runs of LT-PDR.
We can make this formal by means of the following result.

\begin{proposition}\label{prop:multipleLTPDR}
Let $\{\vec{c^m}\}_{m\in M}$ be a family of Kleene sequences.
Then its pointwise intersection $\bigcap_{m\in M} \negation{c^m}$ is a negative sequence.
\end{proposition}

The above intersection is pointwise in the sense that,  for all $j\in {[k,n-1]}$, it holds $(\bigcap_{m\in M} \negation{c^m})_j \defeq \bigcap_{m\in M} (\negation{c^m})_j = \lnot(\{ c_j^m \mid m \in M \}^{\uparrow})$: intuitively, this is (up to $\negation{\cdot}$) a set containing all the $M$ counterexamples.
Note that, if the negative sequence of {\ADPDR} makes \eqref{negativefinal} hold as an equality, as it is possible with any simple heuristic (see Proposition~\ref{prop:negativesequencefinalchain}), then its complement contains \emph{all} Kleene
sequences possibly computed by LT-PDR.

\begin{proposition}\label{prop:LTPDRfinal}
Let $\vec{c}$ be a Kleene sequence and
 $\vec{Y}$ be the negative sequence s.t. $Y_j= (b_r^\downarrow)^{n-1-j}(p^\downarrow)$ for all $j \in [k,n-1]$.
Then  $c_j \in \neg(Y_j)$ for all $j \in [k,n-1]$.
\end{proposition}

While the previous result suggests that simple heuristics are always the best in theory, as they can carry all counterexamples, this is  often not the case in practice, since
they might be computationally hard and outperformed by  some smart over-approximations. An example is given by \eqref{eq:secondterminatingheuristics} in the next section.

\section{Instantiating {\ADPDR} for MDPs}\label{sec:MDP}
In this section we illustrate how to use {\ADPDR} to address the max reachability problem \cite{BaierK} for Markov Decision Processes.

A \emph{Markov Decision Process} (MDP) is a tuple  $(A, S, s_\iota, \delta)$ where $A$ is a set of labels, $S$ is a set of states, $s_\iota \in S$ is an initial state, and $\delta \colon S\times A  \to \mathcal{D}S + 1$ is a transition function.
  Here $\mathcal{D}S$ is the set of probability distributions over $S$, namely functions $d\colon S\to [0, 1]$ such that $\sum_{s\in S} d(s)=1$, and $\mathcal{D}S + 1$ is the disjoint union of $\mathcal{D}S$ and $1=\{*\}$. The transition function $\delta$ assigns to every label $a\in A$ and to every state $s\in S$ either a distribution of states or $* \in 1$. We assume that both $S$ and $A$ are finite sets and
that the set $\mathit{Act}(s)\defeq \{ a\in A \mid \delta(s,a)\neq *\}$ of actions enabled at $s$ is non-empty for all states.

Intuitively, the \emph{max reachability problem} requires to check whether the probability of reaching some bad states $\beta \subseteq S$ is less than or equal to a given threshold $\lambda \in [0, 1]$.
Formally, it can be expressed in lattice theoretic terms, by considering the lattice $([0, 1]^S,\leq)$ of all functions $d\colon S\to [0,1]$, often called frames, ordered pointwise. The max reachability problem consists in checking $\mu b \leq p$ for $p\in[0,1]^S$ and $b \colon [0, 1]^S \to [0, 1]^S$, defined for all $d\in [0, 1]^S $ and $s\in S$, as
  \begin{displaymath}
      p(s)\defeq  \begin{cases}
        \lambda  &\text{ if } s=s_\iota, \\
        1 &\text{ if }  s \neq s_\iota,
      \end{cases}
      \qquad
 b(d)(s) \defeq  \begin{cases}
        1  &\text{ if } s \in \beta, \\
        \displaystyle \max_{a \in \mathit{Act}(s)} \sum_{s'\in S} d(s') \cdot \delta(s, a)(s') &\text{ if } s \notin \beta .
      \end{cases}
  \end{displaymath}
The reader is referred to \cite{BaierK} for all details. %

\smallskip

Since $b$ is not of the form $f\sqcup i$ for a left adjoint $f$ (see e.g. \cite{KoriCAV22}), rather than using {\APDR}, one can exploit {\ADPDR}.
Beyond the simple initial heuristic, which is always applicable and enjoys negative termination, we illustrate now two additional heuristics that are experimentally tested in Section~\ref{sec:experiments}.

The two novel heuristics make the same choices in (Candidate) and (Decide). They exploit functions $\alpha \colon S \to A$, also known as memoryless schedulers,
and the function $b_{\alpha} \colon [0, 1]^S \to [0, 1]^S$ defined for all $d\in [0, 1]^S $ and $s\in S$ as follows:
  \begin{displaymath}
      b_{\alpha}(d)(s) \defeq \begin{cases}
        1  &\text{ if } s \in \beta, \\
        \sum_{s'\in S} d(s') \cdot \delta(s, \alpha(s))(s') &\text{ otherwise}.
      \end{cases}
  \end{displaymath}
Since for all $D\in ([0,1]^S)^\downarrow$, $b^\downarrow_r (D) = \{d \mid b(d) \in D\}
= \bigcap_{\alpha}\{d \mid b_{\alpha} (d)\in D\}$ %
and since {\ADPDR} executes (Decide) only when $b(x_{k-1}) \notin Y_k$, there should exist some $\alpha$ such that $b_{\alpha} (x_{k-1})\notin Y_k$. One can thus fix
\begin{equation}\label{eq:secondterminatingheuristics}
( \vec{x} \| \vec{Y} )_{n,k} \mapsto
\begin{cases*}
		p^\downarrow & if $( \vec{x} \| \vec{Y} )_{n,k} \tr{\mathit{Ca}}$ \\
		\{d \mid b_{\alpha}(d) \in Y_k\}  & if $( \vec{x} \| \vec{Y} )_{n,k} \tr{D}$
\end{cases*}
\end{equation}
Intuitively, such choices are smart refinements of those in \eqref{eq:simple}: for (Candidate) they are exactly the same; for (Decide) rather than taking $b^\downarrow_r (Y_k)$, we consider a larger lower-set determined by the labels chosen by $\alpha$. %
This allows to represent each $Y_j$ as a set of $d\in [0, 1]^S $ satisfying a \emph{single} linear inequality, while using $b^\downarrow_r (Y_k)$ would yield a systems of possibly exponentially many inequalities
(see Example~\ref{ex:shortheuristicforDecide} below).
Moreover, from Theorem~\ref{thm:negativetermination}, it follows that such choices ensures negative termination.

\begin{corollary}\label{cor:ADPDRtermination}
Let $h$ be a legit heuristic defined for (Candidate) and (Decide) as in \eqref{eq:secondterminatingheuristics}.
If $\mu b \not \leq p$, then \emph{\ADPDR}$_h$ terminates.
\end{corollary}

\begin{example}\label{ex:shortheuristicforDecide}
Consider the maximum reachability problem with threshold $\lambda = \frac{1}{4}$ and $\beta= \{s_3\}$
 for the following MDP on alphabet $A=\{a,b\}$ and $s_\iota=s_0$.
  \begin{displaymath}
    \xymatrix@C=30pt{
      s_2\ar@/^1ex/[r]^{b,1} & s_0 \ar@(dr,dl)|{b,\frac{1}{3}} \ar@/^1ex/[l]^{a,\frac{1}{2}\; b,\frac{2}{3}} \ar@/_1ex/[r]_{a,\frac{1}{2}} &s_1 \ar@/_1ex/[l]_{a,\frac{1}{2}} \ar[r]^{a,\frac{1}{2}} &s_3 \ar@(rd, ru)_{a, 1}},
  \end{displaymath}
Hereafter we write $d\in [0,1]^S$ as column vectors with four entries $v_0\dots v_3$
and we will use $\cdot$ for the usual matrix multiplication.
With this notation, the lower set $p^\downarrow \in ([0,1]^S)^\downarrow$ and $b\colon [0,1]^S \to [0,1]^S$ can be written as
\[ p^\downarrow = \{\,\cvec{v_0}{v_1}{v_2}{v_3} \mid {\tiny{ \!\begin{bmatrix} 1 & 0 & 0 &0 \end{bmatrix}\!}} \cdot \cvec{v_0}{v_1}{v_2}{v_3} \leq {\tiny{ \!\begin{bmatrix} \frac{1}{4}  \end{bmatrix}\!} } \} \quad \text{ and } \quad b (\, \cvec{v_0}{v_1}{v_2}{v_3} \,) =\cvec{\max(\frac{v_1+v_2}{2}, \frac{v_0+2v_2}{3})}{\frac{v_0+v_3}{2}}{v_0}{1}
\textrm{.}  \]

Amongst the several memoryless schedulers, only two are relevant for us:
$\zeta \defeq ( s_0 \mapsto a ,\; s_1 \mapsto a ,\; s_2 \mapsto b ,\; s_3 \mapsto a )$  and $\xi \defeq (s_0 \mapsto b ,\; s_1 \mapsto a ,\; s_2 \mapsto b ,\; s_3 \mapsto a)$.
By using the definition of $b_\alpha \colon [0,1]^S \to [0,1]^S$, we have that
\[ b_\zeta (\, \cvec{v_0}{v_1}{v_2}{v_3} \,) =\cvec{\frac{v_1+v_2}{2}}{\frac{v_0+v_3}{2}}{v_0}{1}  \qquad \text{ and } \qquad b_\xi (\, \cvec{v_0}{v_1}{v_2}{v_3} \,) =\cvec{ \frac{v_0+2v_2}{3}}{\frac{v_0+v_3}{2}}{v_0}{1}\textrm{.} \]

It is immediate to see that the problem has negative answer, since using $\zeta$ in $4$ steps or less, $s_0$ can reach $s_3$ already with probability $\frac{1}{4}+\frac{1}{8}$.

\begin{figure}[t]
\[\begin{array}{rc|c|c}
\mathcal{F}^0 & \defeq &  \{\,\cvec{v_0}{v_1}{v_2}{v_3} \mid {\tiny{ \!\begin{bmatrix} 1 & 0 & 0 &0 \end{bmatrix}\!}} \cdot \cvec{v_0}{v_1}{v_2}{v_3} \leq {\tiny{ \!\begin{bmatrix} \frac{1}{4}  \end{bmatrix}\!} } \} & \{\,\cvec{v_0}{v_1}{v_2}{v_3} \mid {\tiny{ \!\begin{bmatrix} 1 & 0 & 0 &0 \end{bmatrix}\!}} \cdot \cvec{v_0}{v_1}{v_2}{v_3} \leq {\tiny{ \!\begin{bmatrix} \frac{1}{4}  \end{bmatrix}\!} } \,\} \\
\mathcal{F}^1 & \defeq &   \{\,\cvec{v_0}{v_1}{v_2}{v_3} \mid {\tiny{ \!\begin{bmatrix} 0 & 1 & 1 &0 \\ 1&0 &2 &0 \end{bmatrix}\!}} \cdot \cvec{v_0}{v_1}{v_2}{v_3} \leq {\tiny{ \!\begin{bmatrix} \frac{1}{2} \\ \frac{3}{4} \end{bmatrix}\!} } \}   & \{\,\cvec{v_0}{v_1}{v_2}{v_3} \mid {\tiny{ \!\begin{bmatrix} 0 & \frac{1}{2} & \frac{1}{2} &0 \end{bmatrix}\!}} \cdot \cvec{v_0}{v_1}{v_2}{v_3} \leq {\tiny{ \!\begin{bmatrix} \frac{1}{4}  \end{bmatrix}\!} } \,\}  \\
\mathcal{F}^2 & \defeq & \{\,\cvec{v_0}{v_1}{v_2}{v_3} \mid {\tiny{ \!\begin{bmatrix} 3 & 0 & 0 & 1 \\ 2&1 &1 &0 \\ 4&0&2&0\end{bmatrix}\!}} \cdot \cvec{v_0}{v_1}{v_2}{v_3} \leq {\tiny{ \!\begin{bmatrix}1 \\ \frac{3}{2} \\ \frac{9}{4}  \end{bmatrix}\!} } \}  & \{\,\cvec{v_0}{v_1}{v_2}{v_3} \mid {\tiny{ \!\begin{bmatrix} \frac{3}{4} & 0 & 0 & \frac{1}{4} \end{bmatrix}\!}} \cdot \cvec{v_0}{v_1}{v_2}{v_3} \leq {\tiny{ \!\begin{bmatrix} \frac{1}{4}  \end{bmatrix}\!} } \,\}   \\
\mathcal{F}^3 & \defeq & \{\,\cvec{v_0}{v_1}{v_2}{v_3} \mid {\tiny{ \!\begin{bmatrix} 0 & \frac{3}{2} & \frac{3}{2} & 0 \\ 1&0 &2 &0 \\ \frac{3}{2}&1&1&\frac{1}{2} \\ \frac{13}{6}&0&\frac{4}{3}&\frac{1}{2}
\\ 2&2&2&0 \\ \frac{10}{3}&0&\frac{8}{3}&0 \end{bmatrix}\!}} \cdot \cvec{v_0}{v_1}{v_2}{v_3} \leq {\tiny{ \!\begin{bmatrix}0 \\ 0 \\ \frac{3}{2} \\ \frac{3}{2} \\\frac{9}{4} \\ \frac{9}{4}  \end{bmatrix}\!} } \} & \{\,\cvec{v_0}{v_1}{v_2}{v_3} \mid {\tiny{ \!\begin{bmatrix} 0 & \frac{3}{8} & \frac{3}{8} & 0 \end{bmatrix}\!}} \cdot \cvec{v_0}{v_1}{v_2}{v_3} \leq {\tiny{ \!\begin{bmatrix} 0  \end{bmatrix}\!} } \,\}\\
\mathcal{F}^4 & \defeq & \{\,\cvec{v_0}{v_1}{v_2}{v_3} \mid  {\tiny{ \!\begin{bmatrix} 1 & 0 & 0 & 0 \\ 0&1 &0 &0 \\ 0&0&1&0\\ 0&0&0&1\end{bmatrix}\!}} \cdot \cvec{v_0}{v_1}{v_2}{v_3} \leq {\tiny{ \!\begin{bmatrix}0 \\ 0 \\ 0\\0  \end{bmatrix}\!} }    \}  = \{ \, \cvec{0}{0}{0}{0} \, \}& \{\,\cvec{v_0}{v_1}{v_2}{v_3} \mid {\tiny{ \!\begin{bmatrix} \frac{9}{16} & 0 & 0 & \frac{3}{16} \end{bmatrix}\!}} \cdot \cvec{v_0}{v_1}{v_2}{v_3} \leq {\tiny{ \!\begin{bmatrix} 0  \end{bmatrix}\!} } \,\} \\
\mathcal{F}^5 & \defeq  & \emptyset  &\emptyset
\end{array}\]
\caption{The elements of the negative sequences computed by {\ADPDR} for the MDP in  Example~\ref{ex:shortheuristicforDecide}. In the central column, these elements are computed by means of the simple initial heuristics, that is $\mathcal{F}^i=(b_r^\downarrow)^i(p^\downarrow)$. In the rightmost column, these elements are computed
using the heuristic in \eqref{eq:secondterminatingheuristics}. In particular $\mathcal{F}^i = \{d\mid b_\zeta(d) \in \mathcal{F}^{i-1} \}$ for $i\leq 3$, while for $i\geq 4$ these are computed as $\mathcal{F}^i = \{d\mid b_\xi(d) \in \mathcal{F}^{i-1} \}$.
}\label{fig:exDecideheuristics}
\end{figure}

To illustrate the advantages of \eqref{eq:secondterminatingheuristics}, we run {\ADPDR} with the simple initial heuristic and with the heuristic that only differs for the choice in (Decide), taken as in \eqref{eq:secondterminatingheuristics}. For both heuristics, the first iterations are the same: several repetitions of (Candidate), (Conflict) and (Unfold) exploiting elements of the positive chain that form
the initial chain (except for the last element $x_{n-1}$). %
{\footnotesize
\begin{align*}
  (\emptyset\cvec{0}{0}{0}{0}\cvec{1}{1}{1}{1} \| \varepsilon)_{3,3}
  \tr{\mathit{Ca}} %
  \tr{\mathit{Co}} (\emptyset\cvec{0}{0}{0}{0}\cvec{0}{0}{0}{1} \| \varepsilon)_{3,3}
  \tr{U}\tr{\mathit{Ca}} %
  \tr{\mathit{Co}} %
  \tr{U}\tr{\mathit{Ca}} %
  \tr{\mathit{Co}} %
   \tr{U}\tr{\mathit{Ca}}  %
   \tr{\mathit{Co}} %
    \tr{U}\tr{\mathit{Ca}}  (\emptyset\cvec{0}{0}{0}{0}\cvec{0}{0}{0}{1} \cvec{0}{\frac{1}{2}}{0}{1}\cvec{\frac{1}{4}}{\frac{1}{2}}{0}{1} \cvec{\frac{1}{4}}{\frac{5}{8}}{\frac{1}{4}}{1} \cvec{1}{1}{1}{1} \|  p^\downarrow )_{7,6} \textrm{.}
\end{align*} }

In the latter state the algorithm has to perform (Decide), since $b(x_5) \notin p^\downarrow$.
Now the choice of $z$ in (Decide) is different for the two heuristics: the former uses $b_r^\downarrow(p^\downarrow) = \{d \mid b(d) \in p^\downarrow \}$, the latter uses $\{d \mid b_\zeta(d) \in p^\downarrow \}$. Despite the different choices, both the heuristics proceed with $6$ steps of (Decide):
{\footnotesize
\begin{align*}
 (\emptyset\cvec{0}{0}{0}{0}\cvec{0}{0}{0}{1} \cvec{0}{\frac{1}{2}}{0}{1}\cvec{\frac{1}{4}}{\frac{1}{2}}{0}{1} \cvec{\frac{1}{4}}{\frac{5}{8}}{\frac{1}{4}}{1} \cvec{1}{1}{1}{1} \|  \mathcal{F}^0 )_{7,6} \tr{D}
 \tr{D} %
  \tr{D} %
    \tr{D} %
 \tr{D}  (\emptyset\cvec{0}{0}{0}{0}\cvec{0}{0}{0}{1} \cvec{0}{\frac{1}{2}}{0}{1}\cvec{\frac{1}{4}}{\frac{1}{2}}{0}{1} \cvec{\frac{1}{4}}{\frac{5}{8}}{\frac{1}{4}}{1} \cvec{1}{1}{1}{1} \|  \mathcal{F}^5,\mathcal{F}^4, \mathcal{F}^3,\mathcal{F}^2,\mathcal{F}^1, \mathcal{F}^0 )_{7,1}
\end{align*} }

The element of the negative sequence $\mathcal{F}^i$ are illustrated in Fig.~\ref{fig:exDecideheuristics} for both the heuristics. In both cases, $\mathcal{F}^5=\emptyset$ and thus {\ADPDR} returns false.

To appreciate the advantages provided by \eqref{eq:secondterminatingheuristics}, it is  enough to compare the two columns for the $\mathcal{F}^i$ in Fig.~\ref{fig:exDecideheuristics}:  in the central column, the number of inequalities defining $\mathcal{F}^i$ significantly  grows, while in the rightmost column is always $1$.
\end{example}

\kori{
Whenever $Y_k$ is generated by a single linear inequality, we observe that
$Y_k=\{d\in [0,1]^S \mid \sum_{s\in S}(r_s \cdot d(s)) \leq r \}$
for suitable non-negative real numbers $r$ and $r_s$ for all $s\in S$.
The convex set $Y_k$ is generated by finitely many $d\in [0,1]^S$ enjoying a convenient property: $d(s)$ is different from $0$ and $1$ only for at most one $s\in S$. The set of its generators, denoted by $\mathcal{G}_k$, can thus be easily computed.
We exploit this property to resolve the choice for (Conflict).
}
We consider  its sub set $\mathcal{Z}_k\defeq \{d \in \mathcal{G}_k \mid b(x_{k-1}) \leq d\}$ and define $z_{B},  z_{01}\in[0,1]^S$ for all $s\in S$ as
\begin{equation}\label{eq:heuristicsZConflict}
\hspace{-10pt}z_{B}(s) \!\defeq\!
\begin{cases*}
		(\bigwedge \mathcal{Z}_k)(s)  & if $r_s \neq 0$, $\mathcal{Z}_k\neq\emptyset$ \\
		b(x_{k-1})(s) & otherwise
\end{cases*}
\hspace{-4pt}z_{01}(s) \!\defeq\!
\begin{cases*}
		\lceil z_{B}(s)\rceil  & if $r_s = 0, \mathcal{Z}_k\neq\emptyset$\\
		z_{B}(s) & otherwise
\end{cases*}
\end{equation}
where, for $u\in[0,1]$, $\lceil u \rceil$ denotes $0$ if $u=0$ and $1$ otherwise. We call \verb|hCoB|  and \verb|hCo01| the heuristics defined as in \eqref{eq:secondterminatingheuristics} for (Candidate) and (Decide) and as $z_{B}$, respectively $z_{01}$, for (Conflict).
The heuristics \verb|hCo01| can be seen as a Boolean modification of
\verb|hCoB|, rounding up positive values to $1$ to accelerate convergence.

\begin{proposition}\label{prop:genlegit}
The heuristics \verb|hCoB|  and \verb|hCo01| are legit.
\end{proposition}

By Corollary~\ref{cor:ADPDRtermination},  {\ADPDR} terminates for negative answers with both \verb|hCoB|  and \verb|hCo01|. We conclude this section with a last example.
  \begin{example}\label{ex:MDPpositive}
Consider the following MDP with alphabet $A=\{a,b\}$ and $s_\iota=s_0$
  \begin{displaymath}
    \xymatrix{
      s_2 \ar@(ld,lu)^{a, 1} & s_0 \ar@(lu,ru)^{a, 1} \ar[l]^{b, \frac{1}{2}} \ar@/_1ex/[r]_{b, \frac{1}{2}} &s_1 \ar@/_1ex/[l]_{a, \frac{1}{3}} \ar[r]^{a, \frac{2}{3}} &s_3 \ar@(rd,ru)_{a, 1}    }
  \end{displaymath}
and the max reachability problem with threshold $\lambda = \frac{2}{5}$ and $\beta=\{s_3\}$.
The lower set $p^\downarrow \in ([0,1]^S)^\downarrow$ and $b\colon [0,1]^S \to [0,1]^S$ can be written as
\[ p^\downarrow = \{\,\cvec{v_0}{v_1}{v_2}{v_3} \mid {\tiny{ \!\begin{bmatrix} 1 & 0 & 0 &0 \end{bmatrix}\!}} \cdot \cvec{v_0}{v_1}{v_2}{v_3} \leq {\tiny{ \!\begin{bmatrix} \frac{2}{5}  \end{bmatrix}\!} } \}
\quad \text{ and }  \quad
b (\, \cvec{v_0}{v_1}{v_2}{v_3} \,) =\cvec{\max(v_0, \frac{v_1+v_2}{2})}{\frac{v_0+2\cdot v_3}{3}}{v_2}{1}
 \]

\noindent
With the simple initial heuristic, {\ADPDR} does not terminate. %
With the heuristic \verb|hCo01|, it returns true in 14 steps, %
while with \verb|hCoB| in 8. %
The first 4 steps, common to both \verb|hCoB| and \verb|hCo01|, are illustrated below.

\begin{lrbox}{\lstbox}\begin{minipage}{\textwidth}
 {\footnotesize \begin{align*}
    & (\emptyset\cvec{0}{0}{0}{0}\cvec{1}{1}{1}{1} \| \varepsilon)_{3,3}
    \tr{\mathit{Ca}} (\emptyset\cvec{0}{0}{0}{0}\cvec{1}{1}{1}{1} \| p^\downarrow )_{3,2} \tr{\mathit{Co}} (\emptyset\cvec{0}{0}{0}{0}\cvec{\frac{2}{5}}{0}{0}{1} \| \varepsilon)_{3,3} & {\tiny b(\cvec{0}{0}{0}{0})=\cvec{0}{0}{0}{1} \quad \mathcal{Z}_2 = \{\cvec{\frac{2}{5}}{0}{0}{1}, \cvec{\frac{2}{5}}{1}{0}{1},\cvec{\frac{2}{5}}{0}{1}{1}, \cvec{\frac{2}{5}}{1}{1}{1} \} } \\
    \tr{U}  \tr{\mathit{Ca}} & (\emptyset\cvec{0}{0}{0}{0}\cvec{\frac{2}{5}}{0}{0}{1} \cvec{1}{1}{1}{1}\| p^\downarrow)_{4,3} \tr{\mathit{Co}} \;\vline\vline (\emptyset\cvec{0}{0}{0}{0}\cvec{\frac{2}{5}}{0}{0}{1} \cvec{\frac{2}{5}}{1}{0}{1} \| \varepsilon)_{4,4}  \;\vline\vline\; (\emptyset\cvec{0}{0}{0}{0}\cvec{\frac{2}{5}}{0}{0}{1} \cvec{\frac{2}{5}}{\frac{4}{5}}{0}{1} \| \varepsilon)_{4,4}  & {\tiny b(\cvec{\frac{2}{5}}{0}{0}{1})= \cvec{\frac{2}{5}}{\frac{4}{5}}{0}{1} \quad \mathcal{Z}_3 = \{ \cvec{\frac{2}{5}}{1}{0}{1}, \cvec{\frac{2}{5}}{1}{1}{1} \} }%
 \end{align*}}
\end{minipage}\end{lrbox}
\[\noindent\scalebox{.95}{\usebox\lstbox}\]

\noindent
Observe that in the first (Conflict) $z_B = z_{01}$, while in the second $z_{01}(s_1)=1$ and $z_{B}(s_1)=\frac{4}{5}$, leading to the two different states prefixed by vertical lines.

\end{example}

\section{Implementation and Experiments}\label{sec:experiments}

We first developed, using Haskell and exploiting its abstraction features,  a common template that accommodates both {\APDR} and {\ADPDR}.
It is a  program parametrized by two lattices---used for positive chains and negative sequences, respectively---and by a heuristic.

For our experiments, we instantiated the template to \ADPDR{} for MDPs (letting $L=[0,1]^{S}$), with three different heuristics: \verb|hCoB| and \verb|hCo01| from Proposition~\ref{prop:genlegit};  and  \verb|hCoS|  introduced below. Besides the template ($\sim$100 lines),  we needed $\sim$140 lines to account for \verb|hCoB| and \verb|hCo01|, and additional $\sim$100 lines to further obtain \verb|hCoS|. All this indicates a clear benefit of our abstract theory: a general template can itself be coded succinctly;  instantiation to concrete problems is easy, too, thanks to an explicitly specified interface of heuristics.

Our implementation accepts MDPs  expressed in a symbolic format
inspired by Prism models~\cite{KwiatkowskaNP11}, in which states are variable valuations and  transitions are described by symbolic functions (they can be segmented with symbolic guards $\{\text{guard}_i\}_{i}$).
We use rational arithmetic
(\verb|Rational| in Haskell) for probabilities
to limit the impact of rounding errors.

\myparagraph{Heuristics.}
The three heuristics (\verb|hCoB|, \verb|hCo01|, \verb|hCoS|) use the same choices in (Candidate) and (Decide),
as defined in \eqref{eq:secondterminatingheuristics},
but different ones in (Conflict).

The third heuristics \verb|hCoS| is a \emph{symbolic} variant of
\verb|hCoB|; it relies on our symbolic model format.
It uses $z_{S}$ for $z$ in (Conflict), where $z_{S}(s)=z_{B}(s)$ if $r_{s}\neq 0$ or $\mathcal{Z}_k=\emptyset$. The definition of $z_{S}(s)$ otherwise is notable: we use a piecewise affine function $(t_{i}\cdot s + u_{i})_{i}$ for $z_{S}(s)$,
where the affine functions $(t_{i}\cdot s + u_{i})_{i}$ are guarded by the same guards $\{\text{guard}_i\}_{i}$ of the MDP's transition function. We let the SMT solver Z3~\cite{MouraB08} search for the values of the coefficients $t_{i}, u_{i}$, so that $z_{S}$ satisfies the requirements of (Conflict) (namely $b(x_{k-1})(s) \leq z_{S}(s) \leq 1$ for each $s\in S$ with $r_s=0$), together with the condition
 $b (z_{S}) \leq z_{S}$ for faster convergence.
If the search is unsuccessful, we give up
\verb|hCoS| and fall back on the heuristic \verb|hCoB|.

As a task common to the three heuristics, we need to calculate $\mathcal{Z}_k = \{d \in \mathcal{G}_k \mid b(x_{k-1}) \leq d\}$ in (Conflict) (see~(\ref{eq:heuristicsZConflict})). Rather than computing the whole set $\mathcal{G}_k$ of generating points of the linear inequality that defines $Y_{k}$, we implemented an ad-hoc algorithm that crucially exploits the condition $b(x_{k-1}) \leq d$ for pruning.

\myparagraph{Experiment Settings.}
We conducted the experiments on Ubuntu 18.04 and AWS t2.xlarge (4 CPUs, 16 GB memory, up to 3.0 GHz Intel Scalable Processor). We used several Markov chain (MC)
 benchmarks
and a couple of MDP ones.

\begin{table}[tb]\scriptsize
  \caption{Experimental results on MC benchmarks.
 $|S|$ is the number of states, $P$ is the
reachability probability (calculated by manual inspection), $\lambda$ is the threshold in the problem $P\le_{?} \lambda$ (shaded if the answer is no). The other columns show the average execution time in seconds; TO is timeout (\SI{900}{\second}); MO is out-of-memory.
For \ADPDR{} and LT-PDR
we used the \texttt{tasty-bench} Haskell package and repeated executions until std.\ dev.\ is $<$ 5\% (at least three execs). For PrIC3 and Storm, we made five executions. Storm's execution does not depend on $\lambda$: it seems to answer queries of the form $P\le_{?} \lambda$ by calculating $P$. We observed a wrong answer for the entry with $(\dagger)$ (Storm, sp.-num., Haddad-Monmege); see the discussion of RQ2.
} \label{tb:mc}
\begin{lrbox}{\lstbox}\begin{minipage}{\textwidth}
  \npdecimalsign{.}
  \nprounddigits{3}
	       \begin{tabular}{ccccccccccccccc}
    	       \toprule
    	       Benchmark & $|S|$ &
    	       $P$
    	       & $\lambda$
    	       &\multicolumn{3}{c}{\ADPDR}
    	       &%
     	       LT-PDR
    	       &\multicolumn{4}{c}{PrIC3}
    	       &\multicolumn{3}{c}{Storm}
    	       \\\cmidrule(lr){5-7} \cmidrule(lr){9-12} \cmidrule(lr){13-15}
    	       &&&&\verb|hCoB| &\verb|hCo01| &\verb|hCoS|
    	       & &none & lin. & pol. & hyb. &sp.-num. &sp.-rat. &sp.-sd. \\
    	       \midrule\midrule
    	       \multirow{4}{*}{Grid} & \multirow{2}{*}{$10^2$} & \multirow{2}{*}{0.033} & 0.3 & 0.013 &0.022 & 0.659 & 0.343 & 1.383 &23.301 &MO &MO &\multirow{2}{*}{0.010} &\multirow{2}{*}{0.010} &\multirow{2}{*}{0.010} \\
               &   &   & 0.2 & 0.013 &0.031 & 0.657 & 0.519 & 1.571 &26.668 &TO &MO \\
    	       \cmidrule[0.5pt](r){2-15}
     	       & \multirow{2}{*}{$10^3$} & \multirow{2}{*}{$<$0.001} & 0.3 & 1.156 &2.187 & 5.633 & 126.441 & TO &TO &TO &MO &\multirow{2}{*}{0.010} &\multirow{2}{*}{0.017} &\multirow{2}{*}{0.011} \\
               &   &   & 0.2 & 1.146 &2.133 & 5.632 & 161.667 & TO &TO &TO &MO \\
    	       \midrule
    	       \multirow{3}{*}{BRP} & \multirow{3}{*}{$10^3$} & \multirow{3}{*}{0.035} & 0.1 & 12.909 &7.969 & 55.788 & TO & TO &TO &MO &MO &\multirow{3}{*}{0.012} &\multirow{3}{*}{0.018} &\multirow{3}{*}{0.011} \\
               &   &   & \cellcolor{gray!25}0.01  & 1.977 &8.111 & 5.645 & 21.078 & 60.738 &626.052 &524.373 &823.082  \\
               &   &   & \cellcolor{gray!25}0.005 & 0.604 &2.261  & 2.709 & 1.429 & 12.171 &254.000 &197.940 &318.840  \\
      	       \midrule
    	       \multirow{8}{*}{\parbox{4em}{Zero- Conf}} & \multirow{4}{*}{$10^2$} & \multirow{4}{*}{0.5} & 0.9 &1.217 &68.937  & 0.196 & TO & 19.765 &136.491 &0.630 &0.468 &\multirow{4}{*}{0.010} &\multirow{4}{*}{0.018} &\multirow{4}{*}{0.011} \\
               &   &   & 0.75 & 1.223 & 68.394 & 0.636 & TO & 19.782 &132.780 &0.602 &0.467  \\
               &   &   & 0.52 & 1.228 & 60.024 & 0.739& TO & 19.852 &136.533 &0.608 &0.474  \\
               &   &   & \cellcolor{gray!25}0.45 &$<$0.001 & 0.001 & 0.001 & $<$0.001 & 0.035 &0.043 &0.043 &0.043 \\
    	       \cmidrule[0.5pt](r){2-15}
     	       & \multirow{4}{*}{$10^4$} & \multirow{4}{*}{0.5} & 0.9 &MO & TO & 7.443 & TO & TO &TO &0.602 &0.465 &\multirow{4}{*}{0.037} &\multirow{4}{*}{262.193} &\multirow{4}{*}{0.031} \\
               &   &   & 0.75 & MO & TO & 15.223 & TO & TO &TO &0.599 &0.470  \\
               &   &   & 0.52 & MO & TO & TO & TO &TO &TO &0.488 & 0.475 \\
               &   &   & \cellcolor{gray!25}0.45 &0.108 & 0.119 & 0.169 & 0.016 & 0.035 &0.040 &0.040 &0.040  \\
    	       \midrule
    	       \multirow{4}{*}{Chain} & \multirow{4}{*}{$10^3$} & \multirow{4}{*}{0.394} & 0.9 & 36.083 & TO & 0.478 & TO & 269.801 &TO &0.938 &0.686 &\multirow{4}{*}{0.010} &\multirow{4}{*}{0.014} &\multirow{4}{*}{0.011} \\
               &   &   & 0.4 & 35.961 & TO & 394.955 & TO & 271.885 &TO &0.920 &TO\\
               &   &   & \cellcolor{gray!25}0.35 & 101.351 & TO & 454.892 & 435.199 & 238.613 &TO &TO &TO  \\
               &   &   & \cellcolor{gray!25}0.3 & 62.036 & 463.981 & 120.557 & 209.346 & 124.829 &746.595 &TO &TO \\
    	       \midrule
    	       \multirow{4}{*}{\parbox{4em}{Double- Chain}} & \multirow{4}{*}{$10^3$} & \multirow{4}{*}{0.215} & 0.9 & 12.122 & 7.318 & TO & TO & TO &TO &1.878 &2.053 &\multirow{4}{*}{0.011} &\multirow{4}{*}{0.018} &\multirow{4}{*}{0.010} \\
               &   &   & 0.3 & 12.120 & 20.424 & TO & TO & TO &TO &1.953 &2.058 \\
               &   &   & 0.216 & 12.096 & 19.540 & TO & TO & TO &TO &172.170 &TO \\
               &   &   & \cellcolor{gray!25}0.15 & 12.344 & 16.172 & TO & 16.963 & TO &TO &TO &TO \\
    	       \midrule
    	       \multirow{4}{*}{\parbox{4em}{Haddad- Monmege}} & \multirow{2}{*}{$41$} & \multirow{2}{*}{0.7} & 0.9 & 0.004 & 0.009 &8.528 & TO & 1.188 &31.915 &TO &MO &\multirow{2}{*}{0.011} &\multirow{2}{*}{0.011} &\multirow{2}{*}{1.560} \\
               &   &   & 0.75 & 0.004 & 0.011 & 2.357 & TO & 1.209 &32.143 &TO &712.086 \\
    	       \cmidrule[0.5pt](r){2-15}
     	       & \multirow{2}{*}{$10^3$} & \multirow{2}{*}{0.7} & 0.9 & 59.721 & 61.777 &TO & TO & TO &TO &TO &TO &\multirow{2}{*}{0.013 $(\dagger)$} &\multirow{2}{*}{0.043} &\multirow{2}{*}{TO} \\
               &   &   & 0.75 & 60.413 & 63.050 & TO & TO & TO &TO &TO &TO \\
    	       \bottomrule
  	       \end{tabular}
\end{minipage}\end{lrbox}
\scalebox{.81}{\centering\usebox\lstbox}

\end{table}

\myparagraph{Research Questions.} We wish to address the following questions.
\begin{description}
 \item[RQ1] Does \ADPDR{} advance the state-of-the-art performance of \emph{PDR} algorithms for probabilistic model checking?
 \item[RQ2] How does \ADPDR{}'s performance compare against \emph{non-PDR} algorithms for probabilistic model checking?
 \item[RQ3] Does the  theoretical framework of \ADPDR{} successfully guide the discovery of various heuristics with practical performance?
 \item[RQ4] Does \ADPDR{} successfully manage nondeterminism in MDPs
  (that is absent in MCs)?
\end{description}

\myparagraph{Experiments on
MCs (Table~\ref{tb:mc}).}
We used six
benchmarks:  Haddad\hyp{}Monmege  is from~\cite{HartmannsKPQR19}; the others are from~\cite{BatzJKKMS20,KoriCAV22}.
We compared
  \ADPDR (with three heuristics) against LT-PDR~\cite{KoriCAV22}, PrIC3 (with four  heuristics \emph{none}, \emph{lin.}, \emph{pol.}, \emph{hyb.}, see~\cite{BatzJKKMS20}), and Storm 1.5~\cite{DehnertJK017}.
Storm is a recent comprehensive toolsuite that implements different algorithms and solvers. Among them, our comparison is against \emph{sparse-numeric}, \emph{sparse-rational}, and \emph{sparse-sound}.
The \emph{sparse} engine uses explicit state space representation by sparse matrices; this is unlike another representative \emph{dd} engine that uses symbolic BDDs. (We did not use \emph{dd} since it often reported errors, and was overall slower than \emph{sparse}.)
 \emph{Sparse-numeric}  is a value-iteration (VI) algorithm;  \emph{sparse-rational}  solves linear (in)equations using rational arithmetic;
 \emph{sparse-sound} is a sound VI algorithm~\cite{QuatmannK18}.\footnote{There are another two sound algorithms in Storm:
 one that utilizes interval iteration~\cite{Baier0L0W17} and the other does optimistic VI~\cite{HartmannsK20}.  We have excluded them from the results since we observed that they returned incorrect answers.}

\begin{table}[tb]
  \scriptsize
  \caption{Experimental results on MDP benchmarks. The legend is the same as Table~\ref{tb:mc}, except that $P$ is now the maximum reachability probability.} \label{tb:mdp}
\begin{lrbox}{\lstbox}\begin{minipage}{\textwidth}
\centering
\begin{tabular}{cccccccccc}
    \toprule
    Benchmark & $|S|$ &
    $P$\
    & $\lambda$
    &\multicolumn{3}{c}{\ADPDR}
    &\multicolumn{3}{c}{Storm}
    \\\cmidrule(lr){5-7}\cmidrule{8-10}
    &&&&\verb|hCoB| &\verb|hCo01| &\verb|hCoS| &sp.-num &sp.-rat. &sp.-sd. \\
    \midrule\midrule
    \multirow{3}{*}{CDrive2} &\multirow{3}{*}{38} &\multirow{3}{*}{0.865} &0.9 &MO &0.172 &TO &\multirow{3}{*}{0.019} &\multirow{3}{*}{0.019} &\multirow{3}{*}{0.018}\\
    &&&\cellcolor{gray!25}0.75 &MO &0.058 &TO\\
    &&&\cellcolor{gray!25}0.5 &0.015 &0.029 &86.798 \\
    \midrule
    \multirow{3}{*}{TireWorld} &\multirow{3}{*}{8670} &\multirow{3}{*}{0.233} &0.9 &MO &3.346 &TO  &\multirow{4}{*}{0.070}&\multirow{4}{*}{0.164} &\multirow{4}{*}{0.069}\\
    &&&0.75 &MO &3.337 &TO \\
    &&&0.5 &MO &6.928 &TO \\
    &&&\cellcolor{gray!25}0.2 &4.246 &24.538 &TO \\
    \bottomrule
\end{tabular}
\end{minipage}\end{lrbox}
\scalebox{.85}{\centering\usebox\lstbox}
\end{table}

\myparagraph{Experiments on MDPs (Table~\ref{tb:mdp}).}
We used two benchmarks  from~\cite{HartmannsKPQR19}. We compared \ADPDR only against Storm, since RQ1 is already addressed using MCs (besides, PrIC3 did not run for MDPs).

\myparagraph{Discussion.} The experimental results suggest the following answers to the RQs.

\textbf{RQ1}. The performance advantage of \ADPDR{}, over both LT-PDR and PrIC3, was clearly observed throughout the benchmarks.
\ADPDR{} outperformed LT-PDR, thus confirming empirically the theoretical observation in Section~\ref{ssec:LTPDRvsADPDR}. The profit is particularly evident in those instances whose answer is positive.
\ADPDR{} generally outperformed PrIC3, too.
Exceptions are in ZeroConf, Chain and DoubleChain, where PrIC3 with  polynomial (pol.) and hybrid (hyb.) heuristics performs well. This seems to be thanks to the expressivity of the polynomial template in PrIC3, which is a possible enhancement we are yet to implement (currently our symbolic heuristic \verb|hCoS| uses only the affine template).

\textbf{RQ2}. The comparison with Storm is interesting.  Note first that Storm's \emph{sparse-numeric} algorithm is a VI algorithm that gives a guaranteed lower bound \emph{without guaranteed convergence}. Therefore  its positive answer to $P\le_{?}\lambda$ may not be correct. Indeed, for Haddad-Monmege with $|S|\sim 10^{3}$, it answered $P=0.5$ which is wrong ($(\dagger)$ in Table~\ref{tb:mc}). This is in contrast with PDR algorithms that discovers an explicit witness for  $P\le\lambda$ via their positive chain.

Storm's \emph{sparse-rational} algorithm is precise.
It was faster than PDR algorithms in many benchmarks, although \ADPDR was better or comparable in ZeroConf ($10^4$) and Haddad-Monmege ($41$), for $\lambda$ such that $P\le\lambda$ is true. We believe this suggests a general advantage of PDR algorithms, namely to accelerate the search for an invariant-like witness for safety.

Storm's \emph{sparse-sound} algorithm is a sound VI algorithm that returns correct answers aside numerical errors.
Its performance was similar to that of sparse-numeric, except for the two instances of Haddad-Monmege: sparse-sound returned correct answers but was much slower than sparse-numeric.
For these two instances, \ADPDR{} outperformed sparse-sound.

It seems that a big part of Storm's good performance is attributed to the sparsity of state representation. This is notable in the comparison of the two instances of Haddad-Monmege ($41$ vs.\ $10^3$): while Storm handles both of them easily, \ADPDR{} struggles a bit in the bigger instance. Our implementation can be extended to use sparse representation, too; this is future work.

\textbf{RQ3}. We derived the three heuristics (\verb|hCoB|, \verb|hCo01|, \verb|hCoS|) exploiting the theory of \ADPDR{}. The experiments show that each heuristic has its own strength. For example, \verb|hCo01| is slower than \verb|hCoB| for MCs, but it is much better for MDPs. In general, there is no silver bullet heuristic, so coming up with a variety of them is important. The experiments suggest that our theory of \ADPDR{} provides great help in doing so.

\textbf{RQ4}. Table~\ref{tb:mdp}  shows that \ADPDR{} can handle nondeterminism well: once a suitable heuristic is chosen, its performances on MDPs and on MCs of similar size are comparable.
It is also interesting that better-performing heuristics vary, as we discussed above.

\myparagraph{Summary.}
\ADPDR{} clearly outperforms  existing probabilistic PDR algorithms in many benchmarks. It also compares well with Storm---a highly sophisticated toolsuite---in a couple of benchmarks. These are notable especially given that \ADPDR{} currently lacks enhancing features such as richer symbolic templates and sparse representation (adding which is future work). Overall, we believe that \ADPDR{} \emph{confirms the potential of PDR algorithms in probabilistic model checking}. Through the three heuristics, we also observed the value of an abstract general theory in devising heuristics in PDR, which is probably true of verification algorithms in general besides PDR.

\bibliographystyle{splncs04}
\bibliography{mybib}

\begin{thebibliography}{10}
\providecommand{\url}[1]{\texttt{#1}}
\providecommand{\urlprefix}{URL }
\providecommand{\doi}[1]{https://doi.org/#1}

\bibitem{BaierK}
Baier, C., Katoen, J.: Principles of model checking. {MIT} Press (2008)

\bibitem{Baier0L0W17}
Baier, C., Klein, J., Leuschner, L., Parker, D., Wunderlich, S.: Ensuring the
  reliability of your model checker: Interval iteration for {Markov Decision
  Processes}. In: Majumdar, R., Kuncak, V. (eds.) Proc. of {CAV} 2017, {Part
  I}. Lecture Notes in Computer Science, vol. 10426, pp. 160--180. Springer
  (2017). \doi{10.1007/978-3-319-63387-9\_8}

\bibitem{BatzJKKMS20}
Batz, K., Junges, S., Kaminski, B.L., Katoen, J., Matheja, C., Schr{\"{o}}er,
  P.: {PrIC3}: Property directed reachability for {MDPs}. In: Lahiri, S.K.,
  Wang, C. (eds.) Proc. of {CAV} 2020, {Part II}. Lecture Notes in Computer
  Science, vol. 12225, pp. 512--538. Springer (2020).
  \doi{10.1007/978-3-030-53291-8\_27}

\bibitem{BonchiGGP18}
Bonchi, F., Ganty, P., Giacobazzi, R., Pavlovic, D.: Sound up-to techniques and
  complete abstract domains. In: Dawar, A., Gr{\"{a}}del, E. (eds.) Proc. of
  {LICS 2018}. pp. 175--184. {ACM} (2018). \doi{10.1145/3209108.3209169}

\bibitem{Bradley11}
Bradley, A.R.: {SAT}-based model checking without unrolling. In: Jhala, R.,
  Schmidt, D.A. (eds.) Proc. of {VMCAI} 2011. Lecture Notes in Computer
  Science, vol.~6538, pp. 70--87. Springer (2011).
  \doi{10.1007/978-3-642-18275-4_7}

\bibitem{CimattiG12}
Cimatti, A., Griggio, A.: Software model checking via {IC3}. In: Madhusudan,
  P., Seshia, S.A. (eds.) Proc. of {CAV} 2012. Lecture Notes in Computer
  Science, vol.~7358, pp. 277--293. Springer (2012).
  \doi{10.1007/978-3-642-31424-7_23}

\bibitem{DBLP:conf/sara/Cousot00}
Cousot, P.: Partial completeness of abstract fixpoint checking. In: Choueiry,
  B.Y., Walsh, T. (eds.) Proc. of {SARA} 2000. Lecture Notes in Computer
  Science, vol.~1864, pp. 1--25. Springer (2000).
  \doi{10.1007/3-540-44914-0\_1}

\bibitem{cousot21}
Cousot, P.: Principles of Abstract Interpretation. MIT Press (2021)

\bibitem{cousot77}
Cousot, P., Cousot, R.: Abstract interpretation: A unified lattice model for
  static analysis of programs by construction or approximation of fixpoints.
  In: Proc. of POPL 1977. p. 238–252. ACM (1977). \doi{10.1145/512950.512973}

\bibitem{DBLP:books/daglib/0023601}
Davey, B.A., Priestley, H.A.: Introduction to Lattices and Order, Second
  Edition. Cambridge University Press (2002)

\bibitem{DehnertJK017}
Dehnert, C., Junges, S., Katoen, J., Volk, M.: A storm is coming: {A} modern
  probabilistic model checker. In: Majumdar, R., Kuncak, V. (eds.) Proc. of
  {CAV} 2017, {Part II}. Lecture Notes in Computer Science, vol. 10427, pp.
  592--600. Springer (2017). \doi{10.1007/978-3-319-63390-9\_31}

\bibitem{EenMB11}
E{\'{e}}n, N., Mishchenko, A., Brayton, R.K.: Efficient implementation of
  property directed reachability. In: Bjesse, P., Slobodov{\'{a}}, A. (eds.)
  Proc. of {FMCAD} 2011. pp. 125--134. {FMCAD} Inc. (2011),
  \url{http://dl.acm.org/citation.cfm?id=2157675}

\bibitem{DBLP:journals/pacmpl/FeldmanSSW22}
Feldman, Y.M.Y., Sagiv, M., Shoham, S., Wilcox, J.R.: Property-directed
  reachability as abstract interpretation in the monotone theory. Proc. {ACM}
  Program. Lang.  \textbf{6}({POPL}),  1--31 (2022). \doi{10.1145/3498676}

\bibitem{GRS00}
Giacobazzi, R., Ranzato, F., Scozzari, F.: Making abstract interpretations
  complete. J. {ACM}  \textbf{47}(2),  361--416 (2000).
  \doi{10.1145/333979.333989}

\bibitem{Gurfinkel2015IC3PA}
Gurfinkel, A.: {IC3}, {PDR}, and friends (2015),
  \url{https://arieg.bitbucket.io/pdf/gurfinkel_ssft15.pdf}

\bibitem{HartmannsK20}
Hartmanns, A., Kaminski, B.L.: Optimistic value iteration. In: Lahiri, S.K.,
  Wang, C. (eds.) Proc. of {CAV} 2020, {Part II}. Lecture Notes in Computer
  Science, vol. 12225, pp. 488--511. Springer (2020).
  \doi{10.1007/978-3-030-53291-8\_26}

\bibitem{HartmannsKPQR19}
Hartmanns, A., Klauck, M., Parker, D., Quatmann, T., Ruijters, E.: The
  quantitative verification benchmark set. In: Vojnar, T., Zhang, L. (eds.)
  Proc. of {TACAS} 2019, {Part I}. Lecture Notes in Computer Science, vol.
  11427, pp. 344--350. Springer (2019). \doi{10.1007/978-3-030-17462-0\_20}

\bibitem{HoderB12}
Hoder, K., Bj{\o}rner, N.: Generalized property directed reachability. In:
  Cimatti, A., Sebastiani, R. (eds.) Proc. of {SAT} 2012. Lecture Notes in
  Computer Science, vol.~7317, pp. 157--171. Springer (2012).
  \doi{10.1007/978-3-642-31612-8\_13}

\bibitem{KoriCAV22}
Kori, M., Urabe, N., Katsumata, S., Suenaga, K., Hasuo, I.: The
  lattice-theoretic essence of property directed reachability analysis. In:
  Shoham, S., Vizel, Y. (eds.) Proc. of {CAV} 2022, {Part I}. Lecture Notes in
  Computer Science, vol. 13371, pp. 235--256. Springer (2022).
  \doi{10.1007/978-3-031-13185-1\_12}

\bibitem{KwiatkowskaNP11}
Kwiatkowska, M.Z., Norman, G., Parker, D.: {PRISM} 4.0: Verification of
  probabilistic real-time systems. In: Gopalakrishnan, G., Qadeer, S. (eds.)
  Proc. of {CAV} 2011. Lecture Notes in Computer Science, vol.~6806, pp.
  585--591. Springer (2011). \doi{10.1007/978-3-642-22110-1\_47}

\bibitem{LangeNNK20}
Lange, T., Neuh{\"{a}}u{\ss}er, M.R., Noll, T., Katoen, J.: {IC3} software
  model checking. Int. J. Softw. Tools Technol. Transf.  \textbf{22}(2),
  135--161 (2020). \doi{10.1007/s10009-019-00547-x}

\bibitem{Levy2004}
Levy, P.B.: Call-By-Push-Value: {A} Functional/Imperative Synthesis, Semantics
  Structures in Computation, vol.~2. Springer (2004)

\bibitem{maclane:71}
MacLane, S.: Categories for the Working Mathematician. Springer-Verlag, New
  York (1971), {Graduate Texts in Mathematics}, Vol. 5

\bibitem{Mil89}
Milner, R.: Communication and Concurrency. Prentice-Hall, Inc., USA (1989)

\bibitem{MouraB08}
de~Moura, L.M., Bj{\o}rner, N.S.: {Z3:} an efficient {SMT} solver. In:
  Ramakrishnan, C.R., Rehof, J. (eds.) Proc. of {TACAS} 2008. Lecture Notes in
  Computer Science, vol.~4963, pp. 337--340. Springer (2008).
  \doi{10.1007/978-3-540-78800-3_24}

\bibitem{QuatmannK18}
Quatmann, T., Katoen, J.: Sound value iteration. In: Chockler, H.,
  Weissenbacher, G. (eds.) Proc. of {CAV} 2018, {Part I}. Lecture Notes in
  Computer Science, vol. 10981, pp. 643--661. Springer (2018).
  \doi{10.1007/978-3-319-96145-3\_37}

\bibitem{SeufertS17}
Seufert, T., Scholl, C.: Sequential verification using reverse {PDR}. In:
  Gro{\ss}e, D., Drechsler, R. (eds.) Proc. of {MBMV} 2017. pp. 79--90. Shaker
  Verlag (2017)

\bibitem{SeufertS18}
Seufert, T., Scholl, C.: Combining {PDR} and reverse {PDR} for hardware model
  checking. In: Madsen, J., Coskun, A.K. (eds.) Proc. of {DATE} 2018. pp.
  49--54. {IEEE} (2018). \doi{10.23919/DATE.2018.8341978}

\bibitem{SeufertS19}
Seufert, T., Scholl, C.: {fbPDR}: In-depth combination of forward and backward
  analysis in property directed reachability. In: Teich, J., Fummi, F. (eds.)
  Proc. of {DATE} 2019. pp. 456--461. {IEEE} (2019).
  \doi{10.23919/DATE.2019.8714819}

\bibitem{Suda14}
Suda, M.: Property directed reachability for automated planning. In: Chien,
  S.A., Do, M.B., Fern, A., Ruml, W. (eds.) Proc. of {ICAPS} 2014. {AAAI}
  (2014). \doi{10.1613/jair.4231}

\end{thebibliography}

\newpage

\appendix


\section{Additional material}\label{app:additional}

\subsection{The Three Executions in Example~\ref{ex:MDPpositive}}\label{app:ExampleMDP}

In this appendix we report in full details the execution of {\ADPDR} for the max reachability problem in Example~\ref{ex:MDPpositive} with the simple initial heuristic (Fig.~\ref{fig:exsimpleinitialMDPs}), with the \verb|hCo01| heuristic (Fig.~\ref{fig:exgen}) and with the \verb|hCoB| heuristic (Fig.~\ref{fig:exgen2}).

In Fig.~\ref{fig:exgen} we exploit the scheduler $\xi\colon S \to A$ defined as $\xi \defeq [s_0\mapsto b,s_1\mapsto a,s_2\mapsto a,s_3\mapsto a]$, for which we illustrate
\begin{equation}\label{eq:xi}
b_\xi (\, \cvec{v_0}{v_1}{v_2}{v_3} \,) =\cvec{\frac{v_1+v_2}{2}}{\frac{v_0+2\cdot v_3}{3}}{v_2}{1}
\quad \text{ and }  \quad
\mathcal{F}^1_\xi \defeq \{ d \mid b_\xi(d ) \in p^\downarrow \} = \{\,\cvec{v_0}{v_1}{v_2}{v_3} \mid {\tiny{ \!\begin{bmatrix} 0 & 1 & 1 &0 \end{bmatrix}\!}}\cdot \cvec{v_0}{v_1}{v_2}{v_3} \leq {\tiny{ \!\begin{bmatrix} \frac{4}{5}  \end{bmatrix}\!} } \}
\end{equation}

\begin{figure}
  {\footnotesize
  \begin{align*}
    & (\cvec{0}{0}{0}{0}\cvec{1}{1}{1}{1} \| \varepsilon)_{2,2} \tr{\mathit{Ca}} (\cvec{0}{0}{0}{0}\cvec{1}{1}{1}{1} \| p^\downarrow )_{2,1}
    \tr{\mathit{Co}} (\cvec{0}{0}{0}{0}\cvec{0}{0}{0}{1} \| \varepsilon)_{2,2}
    \\
    \tr{U}\tr{\mathit{Ca}} & (\cvec{0}{0}{0}{0}\cvec{0}{0}{0}{1} \cvec{1}{1}{1}{1} \| p^\downarrow )_{3,2}
    \tr{\mathit{Co}} (\cvec{0}{0}{0}{0}\cvec{0}{0}{0}{1} \cvec{0}{\frac{2}{3}}{0}{1} \| \varepsilon)_{3,3}
    \\
    \tr{U}\tr{\mathit{Ca}} & (\cvec{0}{0}{0}{0}\cvec{0}{0}{0}{1} \cvec{0}{\frac{2}{3}}{0}{1} \cvec{1}{1}{1}{1} \| p^\downarrow )_{4,3}
    \tr{\mathit{Co}} (\cvec{0}{0}{0}{0}\cvec{0}{0}{0}{1} \cvec{0}{\frac{2}{3}}{0}{1} \cvec{\frac{1}{3}}{\frac{2}{3}}{0}{1} \| \varepsilon )_{4,4}
    \\
    \tr{U}\tr{\mathit{Ca}} & (\cvec{0}{0}{0}{0}\cvec{0}{0}{0}{1} \cvec{0}{\frac{2}{3}}{0}{1} \cvec{\frac{1}{3}}{\frac{2}{3}}{0}{1}\cvec{1}{1}{1}{1} \| p^\downarrow )_{5,4}
    \tr{\mathit{Co}} (\cvec{0}{0}{0}{0}\cvec{0}{0}{0}{1} \cvec{0}{\frac{2}{3}}{0}{1} \cvec{\frac{1}{3}}{\frac{2}{3}}{0}{1}\cvec{\frac{1}{3}}{\frac{7}{9}}{0}{1} \| \varepsilon )_{4,4}
    \\
    \tr{U}\tr{\mathit{Ca}} & (\cvec{0}{0}{0}{0}\cvec{0}{0}{0}{1} \cvec{0}{\frac{2}{3}}{0}{1} \cvec{\frac{1}{3}}{\frac{2}{3}}{0}{1}\cvec{\frac{1}{3}}{\frac{7}{9}}{0}{1}\cvec{1}{1}{1}{1} \| p^\downarrow )_{6,5}
    \tr{\mathit{Co}} (\cvec{0}{0}{0}{0}\cvec{0}{0}{0}{1} \cvec{0}{\frac{2}{3}}{0}{1} \cvec{\frac{1}{3}}{\frac{2}{3}}{0}{1}\cvec{\frac{1}{3}}{\frac{7}{9}}{0}{1}\cvec{\frac{7}{18}}{\frac{7}{9}}{0}{1} \| \varepsilon )_{6,6}
    \\
    \tr{U}\tr{\mathit{Ca}} & (\cvec{0}{0}{0}{0}\cvec{0}{0}{0}{1} \cvec{0}{\frac{2}{3}}{0}{1} \cvec{\frac{1}{3}}{\frac{2}{3}}{0}{1}\cvec{\frac{1}{3}}{\frac{7}{9}}{0}{1}\cvec{\frac{7}{18}}{\frac{7}{9}}{0}{1}\cvec{1}{1}{1}{1} \| p^\downarrow )_{7,6}
    \tr{\mathit{Co}} (\cvec{0}{0}{0}{0}\cvec{0}{0}{0}{1} \cvec{0}{\frac{2}{3}}{0}{1} \cvec{\frac{1}{3}}{\frac{2}{3}}{0}{1}\cvec{\frac{1}{3}}{\frac{7}{9}}{0}{1}\cvec{\frac{7}{18}}{\frac{7}{9}}{0}{1}\cvec{\frac{7}{18}}{\frac{43}{54}}{0}{1} \| \varepsilon )_{7,7}
    \cdots
    \end{align*}
  }
\caption{The non-terminating execution of {\ADPDR} with the simple initial heuristics for the max reachability problem of Example \ref{ex:MDPpositive}. The elements of the positive chain, with the exception of the last one $x_{n-1}$ are those of the initial chain.}\label{fig:exsimpleinitialMDPs}
\end{figure}

\begin{figure}
{\footnotesize\begin{align*}
    & (\emptyset\cvec{0}{0}{0}{0}\cvec{1}{1}{1}{1} \| \varepsilon)_{3,3}
    \tr{\mathit{Ca}} (\emptyset\cvec{0}{0}{0}{0}\cvec{1}{1}{1}{1} \| p^\downarrow )_{3,2} \tr{\mathit{Co}} (\emptyset\cvec{0}{0}{0}{0}\cvec{\frac{2}{5}}{0}{0}{1} \| \varepsilon)_{3,3} & b(\cvec{0}{0}{0}{0})=\cvec{0}{0}{0}{1} \quad \mathcal{Z}_2 = \{\cvec{\frac{2}{5}}{0}{0}{1}, \cvec{\frac{2}{5}}{1}{0}{1},\cvec{\frac{2}{5}}{0}{1}{1}, \cvec{\frac{2}{5}}{1}{1}{1} \} \\
    %
    %
    \tr{U}  \tr{\mathit{Ca}} & (\emptyset\cvec{0}{0}{0}{0}\cvec{\frac{2}{5}}{0}{0}{1} \cvec{1}{1}{1}{1}\| p^\downarrow)_{4,3} \tr{\mathit{Co}} (\emptyset\cvec{0}{0}{0}{0}\cvec{\frac{2}{5}}{0}{0}{1} \cvec{\frac{2}{5}}{1}{0}{1} \| \varepsilon)_{4,4} & b(\cvec{\frac{2}{5}}{0}{0}{1})= \cvec{\frac{2}{5}}{\frac{4}{5}}{0}{1} \quad \mathcal{Z}_3 = \{ \cvec{\frac{2}{5}}{1}{0}{1}, \cvec{\frac{2}{5}}{1}{1}{1} \}\\
    %
    %
    %
    \tr{U} \tr{\mathit{Ca}} & (\emptyset\cvec{0}{0}{0}{0}\cvec{\frac{2}{5}}{0}{0}{1} \cvec{\frac{2}{5}}{1}{0}{1} \cvec{1}{1}{1}{1} \| p^\downarrow)_{5,4}  \tr{D}  (\emptyset\cvec{0}{0}{0}{0}\cvec{\frac{2}{5}}{0}{0}{1} \cvec{\frac{2}{5}}{1}{0}{1} \cvec{1}{1}{1}{1} \| \mathcal{F}^1_\xi p^\downarrow)_{5,3} & b(\cvec{\frac{2}{5}}{1}{0}{1})= \cvec{max(\frac{2}{5},\frac{1}{2})}{\frac{4}{5}}{0}{1} \not \in p^\downarrow  \\
    \tr{\mathit{Co}} &  (\emptyset\cvec{0}{0}{0}{0}\cvec{\frac{2}{5}}{0}{0}{1} \cvec{\frac{2}{5}}{\frac{4}{5}}{0}{1} \cvec{1}{1}{1}{1} \| p^\downarrow)_{5,4}  &b(\cvec{\frac{2}{5}}{0}{0}{1})= \cvec{\frac{2}{5}}{\frac{4}{5}}{0}{1} \quad \mathcal{Z}_3 = \{ \cvec{1}{\frac{4}{5}}{0}{1} \}\\
        \tr{\mathit{Co}} &  (\emptyset\cvec{0}{0}{0}{0}\cvec{\frac{2}{5}}{0}{0}{1} \cvec{\frac{2}{5}}{\frac{4}{5}}{0}{1} \cvec{\frac{2}{5}}{1}{0}{1} \| \varepsilon)_{5,5}  & b(\cvec{\frac{2}{5}}{\frac{4}{5}}{0}{1})= \cvec{\frac{2}{5}}{\frac{4}{5}}{0}{1} \quad \mathcal{Z}_4 =\{ \cvec{\frac{2}{5}}{1}{0}{1}, \cvec{\frac{2}{5}}{1}{1}{1} \} \\
        \tr{U} \tr{\mathit{Ca}} &(\emptyset\cvec{0}{0}{0}{0}\cvec{\frac{2}{5}}{0}{0}{1} \cvec{\frac{2}{5}}{\frac{4}{5}}{0}{1} \cvec{\frac{2}{5}}{1}{0}{1} \cvec{1}{1}{1}{1} \| p^\downarrow)_{6,5} \tr{D} (\emptyset\cvec{0}{0}{0}{0}\cvec{\frac{2}{5}}{0}{0}{1} \cvec{\frac{2}{5}}{\frac{4}{5}}{0}{1} \cvec{\frac{2}{5}}{1}{0}{1} \cvec{1}{1}{1}{1} \| \mathcal{F}^1_\xi p^\downarrow)_{6,4} & b(\cvec{\frac{2}{5}}{1}{0}{1})= \cvec{max(\frac{2}{5},\frac{1}{2})}{\frac{4}{5}}{0}{1} \not \in p^\downarrow \\
 \tr{\mathit{Co}} &   (\emptyset\cvec{0}{0}{0}{0}\cvec{\frac{2}{5}}{0}{0}{1} \cvec{\frac{2}{5}}{\frac{4}{5}}{0}{1} \cvec{\frac{2}{5}}{\frac{4}{5}}{0}{1} \cvec{1}{1}{1}{1} \|  p^\downarrow)_{6,4}     & b(\cvec{\frac{2}{5}}{\frac{4}{5}}{0}{1})= \cvec{\frac{2}{5}}{\frac{4}{5}}{0}{1}\quad \mathcal{Z}_3 = \{ \cvec{1}{\frac{4}{5}}{0}{1} \}
 \end{align*}}
\caption{On the left, the execution of {\ADPDR}$_{\texttt{hCo01}}$ for the max reachability problem of Example \ref{ex:MDPpositive}: in the last state, it returns true since $x_3=x_4$. On the right, the data explaining the choices of (Conflict) and (Decide). Note that, in the two (Decide) steps, the guard $b(x_{k-1}) \notin Y_k$ holds because of the possibility of choosing the label $b$ in state $s_0$. This explain why  $Z$  is taken as $\mathcal{F}^1_\xi(p\downarrow)$ for the scheduler $\xi$ defined in~\eqref{eq:xi}.}\label{fig:exgen}
\end{figure}

\begin{figure}
 {\footnotesize \begin{align*}
    & (\emptyset\cvec{0}{0}{0}{0}\cvec{1}{1}{1}{1} \| \varepsilon)_{3,3}
    \tr{\mathit{Ca}} (\emptyset\cvec{0}{0}{0}{0}\cvec{1}{1}{1}{1} \| p^\downarrow )_{3,2} \tr{\mathit{Co}} (\emptyset\cvec{0}{0}{0}{0}\cvec{\frac{2}{5}}{0}{0}{1} \| \varepsilon)_{3,3} & {\tiny b(\cvec{0}{0}{0}{0})=\cvec{0}{0}{0}{1} \quad \mathcal{Z}_2 = \{\cvec{\frac{2}{5}}{0}{0}{1}, \cvec{\frac{2}{5}}{1}{0}{1},\cvec{\frac{2}{5}}{0}{1}{1}, \cvec{\frac{2}{5}}{1}{1}{1} \}  } \\
    %
    %
    \tr{U}  \tr{\mathit{Ca}} & (\emptyset\cvec{0}{0}{0}{0}\cvec{\frac{2}{5}}{0}{0}{1} \cvec{1}{1}{1}{1}\| p^\downarrow)_{4,3} \tr{\mathit{Co}}  (\emptyset\cvec{0}{0}{0}{0}\cvec{\frac{2}{5}}{0}{0}{1} \cvec{\frac{2}{5}}{\frac{4}{5}}{0}{1} \| \varepsilon)_{4,4}  & {\tiny b(\cvec{\frac{2}{5}}{0}{0}{1})= \cvec{\frac{2}{5}}{\frac{4}{5}}{0}{1} \quad \mathcal{Z}_3 = \{ \cvec{\frac{2}{5}}{1}{0}{1}, \cvec{\frac{2}{5}}{1}{1}{1} \} }\\
    %
    %
    %
    \tr{U} \tr{\mathit{Ca}} & (\emptyset\cvec{0}{0}{0}{0}\cvec{\frac{2}{5}}{0}{0}{1} \cvec{\frac{2}{5}}{\frac{4}{5}}{0}{1} \cvec{1}{1}{1}{1} \| p^\downarrow)_{5,4}  \tr{\mathit{Co}}  (\emptyset\cvec{0}{0}{0}{0}\cvec{\frac{2}{5}}{0}{0}{1} \cvec{\frac{2}{5}}{\frac{4}{5}}{0}{1} \cvec{\frac{2}{5}}{\frac{4}{5}}{0}{1} \| p^\downarrow)_{5,4} & {\tiny b(\cvec{\frac{2}{5}}{\frac{4}{5}}{0}{1})= \cvec{\frac{2}{5}}{\frac{4}{5}}{0}{1} \quad \mathcal{Z}_4 = \mathcal{Z}_3 }
 \end{align*}}
\caption{On the left, the execution of {\ADPDR}$_{\texttt{hCoB}}$ for the max reachability problem of Example \ref{ex:MDPpositive}: in the last state, it returns true since $x_3=x_4$. On the right, the data explaining the three choices of (Conflict). }\label{fig:exgen2}
\end{figure}

 \subsection{A Summary About the Max Reachability Problem}\label{app:summarymaxreach}
As a courtesy to the reader, we provide here a short description of the max reachability problems for MDPs. The reader is referred to \cite{BaierK} for further details.

An MDP $(A, S, s_\iota, \delta)$ mixes nondeterministic and probabilistic computations. The notion of a \emph{scheduler}, also known as \emph{adversary}, \emph{policy} or \emph{strategy}, is used to resolve nondeterministic choices.
Below, we write $S^+$ for the set of non-empty sequence over $S$, intuitively representing runs of the MDP.
A scheduler is a function $\alpha \colon S^+ \to A$ such that $\alpha(s_0 s_1\dots s_n)\in \mathit{Act}(s_n)$: given the states visited so far, the scheduler decides which action to trigger among the enabled ones so that the MDP behaves as a Markov chain. A scheduler $\alpha$ is called \emph{memoryless} if it always selects the same action in a given state, namely, if $\alpha(s_0 s_1\dots s_n)=\alpha(s_n)$ for any sequence $s_0 s_1\dots s_n\in S^+$. Memoryless schedulers can thus be represented just as functions $\alpha \colon S \to A$ such that $\alpha(s)\in \mathit{Act}(s)$ for any $s\in S$.

Given an MDP, the \emph{max reachability problem} requires to check whether the probability of reaching some bad states $\beta \subseteq S$ is less than or equal to a given threshold $\lambda \in [0, 1]$ for all possible schedulers.
Thus, to solve this problem, one should compute the supremum over infinitely many schedulers. Notably, it is known that there always exists one memoryless scheduler that maximizes the probabilities to reach $\beta$  (see e.g. \cite{BaierK}). As the memoryless schedulers are finitely many (although their number can grow exponentially), the supremum can thus be replaced by a maximum.

The fact that the problem stated in this way coincides with the lattice theoretic problem that we illustrated in Section \ref{sec:MDP}, it is well known: see e.g. \cite{BaierK}.


\section{Proofs of Section~\ref{sec:APDR}}\label{app:APDR}

In this appendix, we illustrate the proofs for the various results in Section~\ref{sec:APDR}. After the proofs for the invariants (Appendix~\ref{app:invariants}), we show soundness (Appendix~\ref{app:soundness}) and discuss two results about positive and negative sequences (Appendix~\ref{app:semilattices}) that will be used for the proof of Proposition~\ref{prop:multipleLTPDR}. The proofs for the results in Sections~\ref{sec:progression},~\ref{sec:heuristics} and~\ref{sec:termination} are in Appendixes~\ref{app:prog},~\ref{app:heuristics} and~\ref{app:neg}.


\subsection{Proofs about invariants}\label{app:invariants}


In this appendix we show that the properties in Fig.~\ref{fig:invariants} are invariants of {\APDR}, namely they hold in all reachable states. The proofs are fairly standard and some of them are similar in the spirit to those in~\cite{KoriCAV22}. However, we illustrate them in full details as we believe that the reader may find them helpful. \smallskip

For a state $s$ and a property $(Q)$, we will write $s\models (Q)$ to mean that $(Q)$ holds in $s$.
In order to show that $(Q)$ is an invariant we will prove
\begin{itemize}
\item[(a)] $s_0 \models (Q)$ and
\item[(b)] if $s \models (Q)$ and $s \tr{ } s'$, then $s'\models (Q)$.
\end{itemize}
Hereafter, we will fix $s=( \vec{x} \| \vec{y} )_{n,k}$ and $s'=( \vec{x}' \| \vec{y}' )_{n',k'}$. As usual we will write $x_j$ and $y_j$ for the elements of $\vec{x}$ and $\vec{y}$. For the elements of $\vec{x}'$ and $\vec{y}'$, we will write $x_j'$ and $y_j'$. Through the proofs, we will avoid to repeat every time in (b) that $s \models (Q)$, and we will just write $\stackrel{\tiny{(Q)}}{=}$ or $\stackrel{\tiny{(Q)}}{\tiny{\sqsubseteq}}$ whenever using such hypothesis. Moreover in (b) we will avoid to specify those cases that are trivial: for instance, for the properties that only concerns the positive chain $\vec{x}$, e.g., \eqref{eq:x0bot}, \eqref{eq:positivechain}, \eqref{eq:Ix1}, \eqref{eq:xP}  and \eqref{eq:positiveF}, it is enough to check the property (b) for $s \tr{U} s'$ and $s \tr{\mathit{Co}} s'$, since $s \tr{D} s'$ and $s \tr{\mathit{Ca}} s'$ only modify the negative sequence $\vec{y}$.


\medskip

We begin with proving \eqref{eq:x0bot} and \eqref{eq:positivechain}.

\paragraph{Proof of \eqref{eq:x0bot}$\colon$ $ x_0 = \bot $}
\begin{itemize}
\item[(a)] In $s_0$, $x_0= \bot$.
\item[(b)] If $s \tr{U} s'$, then $x_0' = x_0 \stackrel{\tiny{\eqref{eq:x0bot}}}{=}\bot$. \\
If $s \trz{\mathit{Co}}{z} s'$, then $x_0'=x_0 \sqcap z \stackrel{\tiny{\eqref{eq:x0bot}}}{=} \bot \sqcap z = \bot$. \qed
\end{itemize}

\paragraph{Proof of \eqref{eq:positivechain}$\colon$  }$\forall j\in[0, n-2] \text{, } x_j \sqsubseteq x_{j+1}$
\begin{itemize}
\item[(a)] In $s_0$, since $n=2$ one needs to check only the case $j=0$: $x_{0} = \bot \sqsubseteq \top =x_1 $.
\item[(b)] If $s \tr{U} s'$, then $x_j' =x_j \stackrel{\tiny{\eqref{eq:positivechain}}}{\sqsubseteq}  x_{j+1} =x_{j+1}'$ for all $j\in[0,n-2]$. For $j=n-1$, $x_{n-1}' =x_{n-1} \sqsubseteq \top =x_{n}'$. Since $n'=n+1$, then $\forall j\in [0,n'-2] \text{, } x_j'\sqsubseteq x_{j+1}'$.\\
If $s \trz{\mathit{Co}}{z} s'$, then $x_j' = x_j \sqcap z \stackrel{\tiny{\eqref{eq:positivechain}}}{\sqsubseteq} x_{j+1} \sqcap z \sqsubseteq x_{j+1}'$ for all $j\in[0,k]$.
For all $j\in [i+1,n-2]$, $x_j' = x_j \stackrel{\tiny{\eqref{eq:positivechain}}}{\sqsubseteq} x_{j+1}  =x_{j+1}'$. Thus, $\forall j\in [0,n'-2] \text{, } x_j'\sqsubseteq x_{j+1}'$.  \qed
\end{itemize}

%

The proof for the next invariant follows a different pattern and crucially relies on the following observation.
\begin{lemma}\label{lemma:i=1}
When $k=1$, then the algorithm either returns or does (Conflict).
\end{lemma}
\begin{proof}
Since $k=1$ then $\vec{y}\neq \varepsilon$.
By \eqref{eq:x0bot} $x_0=\bot$. Since $f$ is a left adjoint, then $f(x_0) = \bot$. Thus $f(x_0) \sqsubseteq y_1$ and if the state is not conclusive then (Conflict) is enabled.
\qed
\end{proof}

\paragraph{Proof of \eqref{eq:invi}$\colon$}$1\leq k \leq n$
\begin{itemize}
\item To prove that $1 \leq k$, observe that $k$ is initialised at $2$ and that it is only decremented by $1$. When $k=1$, by Lemma \ref{lemma:i=1}, we have that the algorithms either returns or does (Conflict) and thus increment $k$.
\item To prove that $k \leq n$, observe that $k$ is incremented only by $1$. When $k=n$, the algorithm does either (Unfold) or (Candidate). In the latter case, $k$ is decremented. In the former, both $n$ and $k$ are incremented.
\qed
\end{itemize}

It is worth to remark that the proofs of the three invariants \eqref{eq:x0bot}, \eqref{eq:invi}  and \eqref{eq:positivechain} do not rely on the properties of the chosen element $z\in L$. Such properties are instead fundamental for proving the invariants of the positive chain
(\eqref{eq:Ix1}, \eqref{eq:xP}, \eqref{eq:positiveF} and \eqref{eq:positiveG}), and the invariants of the negative sequence (\eqref{eq:Pepsilon} and \eqref{eq:negativeG}).

\paragraph{Proof of \eqref{eq:Ix1}$\colon$}$i \sqsubseteq x_1$
\begin{itemize}
\item[(a)] In $s_0$, $x_1 = \top \sqsupseteq i$.
\item[(b)] If $s \tr{U} s'$, then $x_1'=x_1 \stackrel{\tiny{\eqref{eq:Ix1}}}{\sqsupseteq} i$. \\
If $s \trz{\mathit{Co}}{z} s'$, since $z\sqsupseteq i$,  then $x_1' = x_1\sqcap z \stackrel{\tiny{\eqref{eq:Ix1}}}{\sqsupseteq} i\sqcap i = i$.
 \qed
\end{itemize}

\paragraph{Proof of \eqref{eq:xP}$\colon$}$x_{n-2}\sqsubseteq p$
\begin{itemize}
\item[(a)] In $s_0$, $n=2$ and $x_{0}= \bot \sqsubseteq p$.
\item[(b)] If $s \tr{U} s'$, since (Unfold) is applied only if $x_{n-1}\sqsubseteq p$ and $n'=n+1$, then $x'_{n'-2}=x_{n-1} \sqsubseteq p$. \\
If $s \tr{\mathit{Co}} s'$, since $n'=n$, then $x_{n'-2}'= x_{n-2}'  \sqsubseteq x_{n-2} \stackrel{\tiny{\eqref{eq:xP}}}{\sqsubseteq} p$.   \qed
\end{itemize}

\paragraph{Proof of \eqref{eq:positiveF}$\colon$}$\forall j\in[0, n-2] \text{, } f(x_j) \sqsubseteq x_{j+1}$
\begin{itemize}
\item[(a)] In $s_0$, since $n=2$  one needs to check only the case $j=0$: $f(x_0) \sqsubseteq \top = x_1$.
\item[(b)] If $s \tr{U} s'$, then $f(x_j') =f(x_j) \stackrel{\tiny{\eqref{eq:positivechain}}}{\sqsubseteq} x_{j+1} =x_{j+1}'$ for all $j\in[0,n-2]$. For $j=n-1$, $f(x_{n-1}') = f(x_{n-1}) \sqsubseteq \top = x_{j+1}'$. Since $n'=n+1$, then $\forall j\in [0,n'-2] \text{, } f(x_j')\sqsubseteq x_{j+1}'$. \\
If $s \trz{\mathit{Co}}{z} s'$, since $f(x_{k-1} \sqcap z) \sqsubseteq z$, then by \eqref{eq:positivechain} and monotonicity of $f$ it holds that $\forall j\in [0,k-1]$, $f(x_{j} \sqcap z) \sqsubseteq z$. Since $f(x_j \sqcap z) \sqsubseteq f(x_j) \stackrel{\tiny{\eqref{eq:positiveF}}}{\sqsubseteq} x_{j+1}$, it holds that $f(x_j \sqcap z) \sqsubseteq x_{j+1} \sqcap z $ for all $j\in[0, k-1]$. With this observation is immediate to conclude that $\forall j\in[0,n'-2] \text{, }f(x_j')\sqsubseteq x_{j+1}'$. \qed
\end{itemize}

\paragraph{Proof of \eqref{eq:positiveG}$\colon$}$\forall j\in[0, n-2] \text{, }x_j \sqsubseteq g(x_{j+1})$

\noindent
Follows immediately from \eqref{eq:positiveF} and $f \dashv g$. \qed

\medskip

Next we prove that \eqref{eq:Pepsilon} and \eqref{eq:negativeG} are invariant. We will omit the cases of (Conflict) and (Unfold) since they are trivial. Indeed, the negative sequence $\vec{y}$ is truncated in the rule (Conflict) but if the two invariants holds for $\vec{y}$ then they obviously hold also for its tail $\mathsf{tail}(\vec{y})$. Moreover the number $n$ is modified by (Unfold) but, the two invariants trivially holds when $\vec{y} = \varepsilon$.

\paragraph{Proof of \eqref{eq:Pepsilon}$\colon$ }If $\vec{y}\neq \varepsilon$ then $p \sqsubseteq y_{n-1}$
\begin{itemize}
\item[(a)] In $s_0$, $\vec{y}= \varepsilon$. Thus \eqref{eq:Pepsilon} trivially holds.
\item[(b)] If $s \trz{\mathit{Ca}}{z} s'$, since $p\sqsubseteq z$, then $p \sqsubseteq z = y_{n-1}'$.\\
If $s \tr{D} s'$, then $p \stackrel{\tiny{\eqref{eq:Pepsilon}}}{\sqsubseteq} y_{n-1}=y'_{n-1}$. \qed
\end{itemize}

\paragraph{Proof of \eqref{eq:negativeG}$\colon$}$\forall j\in[k,n-2] \text{, } g(y_{j+1}) \sqsubseteq y_j$
\begin{itemize}
\item[(a)] In $s_0$, $k=2$ and $n=2$. Thus \eqref{eq:negativeG} trivially holds.
\item[(b)] If $s \tr{\mathit{Ca}} s'$, then $k'=n-1$ and thus \eqref{eq:negativeG} trivially holds.\\
If $s \trz{D}{z} s'$, since  $z\sqsupseteq g(y_k)$ and $k'=k-1$, then $y_{k'}' = y_{k-1}' = z \sqsupseteq g(y_k) = g(y_k')= g(y_{k'+1}')$. For $j\in[k'+1,n-2]$, namely for $j\in[k,n-2]$, it holds that $y'_j= y_j\stackrel{\tiny{\eqref{eq:negativeG}}}{\sqsupseteq} g(y_{j+1})=g(y_{j+1}')$. Thus, $\forall j \in [k',n-2] \text{, }g(y_{j+1}') \sqsubseteq y_j'$. \qed
\end{itemize}


At this point it is worth to mention that the proofs of the invariants illustrated so far only exploit the second constraints on the chosen element $z\in L$ in the rules (Candidate), (Decide) and (Conflict). The proof of the next invariant, necessary to prove progression, relies on the first constraints on $z$.

\paragraph{Proof of \eqref{eq:positivenegative}$\colon$}$\forall j \in [k, n - 1] \text{, } x_j \not\sqsubseteq y_j$
\begin{itemize}
\item[(a)] In $s_0$, $k=n$ and thus \eqref{eq:positivenegative} trivially holds.
\item[(b)] If $s \tr{U} s'$, then $k'=n'$ and thus \eqref{eq:positivenegative} trivially holds.\\
If $s \trz{\mathit{Ca}}{z} s'$, since $x_{n-1} \not\sqsubseteq z$, $x'_{n-1}=x_{n-1}$ and $k'=n'-1=n-1$, then $x'_{n'-1} = x_{n-1}\not\sqsubseteq z = y'_{n'-1}$.  \\
If $s \trz{D}{z} s'$, since $x_{k-1} \not\sqsubseteq z$, then $x_{k-1}'=x_{k-1} \not\sqsubseteq z = y'_{k-1}$. Moreover, $\forall j \in [k, n - 1]$,
$x_j' = x_j \stackrel{\tiny{\eqref{eq:positivenegative}}}{ \not\sqsubseteq} y_j = y_j'$. Thus, $\forall j \in [k', n' - 1] \text{, } x_j' \not\sqsubseteq y_j'$.\\
If $s \tr{\mathit{Co}} s'$, then $k'=k+1$ and $n'=n$. Observe that for $j \in [k + 1, n - 1]$, $x'_j = x_j \stackrel{\tiny{\eqref{eq:positivenegative}}}{ \not\sqsubseteq} y_j =y_j'$. Thus $\forall j \in  [k', n' - 1] \text{,  }x_j' \not\sqsubseteq y_j' $.
\qed
\end{itemize}

We conclude with the proofs of \eqref{eq:positiveinitialfinal}, \eqref{positivefinal} and \eqref{negativefinal} that are simple arguments based on the invariants proved above.

\paragraph{Proof of \eqref{eq:positiveinitialfinal}$\colon$}$\forall j \in [0, n-1] \text{, } (f \sqcup i)^j (\bot) \sqsubseteq x_j \sqsubseteq (g \sqcap p)^{n-1-j} (\top)$
\begin{itemize}
\item We prove $(f \sqcup i)^j \bot \sqsubseteq x_j$ by induction on $j \in [0, n-1]$.
For $j=0$, $(f \sqcup i)^j  \bot = \bot \sqsubseteq x_j$. For $j \in [1,n-1]$,
\begin{align*}
(f \sqcup i)^j \bot
    & = (f \sqcup i) (\, (f \sqcup i)^{j-1} \bot \,) & \tag{def.}\\
    & = f (\, (f \sqcup i)^{j-1} \bot \,)  \sqcup i  & \tag{def.}    \\
    & \sqsubseteq f(x_{j-1}) \sqcup i \tag{Induction hypothesis} &  \\
    & \sqsubseteq x_j \sqcup i & \tag{\eqref{eq:positiveF}} \\
    & = x_j  & \tag{\eqref{eq:Ix1} and \eqref{eq:positivechain}}
\end{align*}
\item In order to prove $x_j \sqsubseteq (g \sqcap p)^{n-1-j} \top$ for all $j \in [0, n-1]$, it is convenient to prove the equivalent statement  $x_{n-1-j} \leq (g \sqcap p)^{j} \top$ for all $j \in [0, n-1]$.
For $j=0$, $x_{n-1} \sqsubseteq \top = (g \sqcap p)^{0} \top$. For $j \in [1,n-1]$,
\begin{align*}
(g \sqcap p)^{j} \top
    & = (g \sqcap p) (\, (g \sqcap p)^{j-1} \top \,) & \tag{def.}\\
    & = g (\, (g \sqcap p)^{j-1} \top \,)  \sqcap p  & \tag{def.}    \\
    & \sqsupseteq g(x_{n-1-j+1}) \sqcap p \tag{Induction hypothesis} &  \\
    & \sqsupseteq x_{n-1-j} \sqcap p & \tag{\eqref{eq:positiveG}} \\
    & = x_{n-1-j}  & \tag{\eqref{eq:xP} and \eqref{eq:positivechain}}
\end{align*} \qed
\end{itemize}

\paragraph{Proof of \eqref{positivefinal}$\colon$}$\forall j \in [1, n-1] \text{, } x_{j-1} \sqsubseteq g^{n-1-j}(p)$

\noindent
By \eqref{eq:chainsadjoint}, for all $l\in \Nat$, $(g \sqcap p)^{l+1}\top = \bigsqcap_{j\leq l}g^j(p)$. Thus, using \eqref{eq:positiveinitialfinal}, it holds that for all $j\in [1,n-1]$, $x_{j-1} \sqsubseteq (g \sqcap p)^{n-j} \top = \bigsqcap_{j\leq n-j-1}g^j(p) \sqsubseteq g^{n-j-1}(p)$. \qed

\paragraph{Proof of \eqref{negativefinal}$\colon$}$\forall j\in[k,n-1]\text{, }g^{n-1-j}(p) \sqsubseteq y_j$

\noindent
In order to prove \eqref{negativefinal}, we prove the equivalent statement $g^j(p) \sqsubseteq y_{n-1-j}$ for all $j\in [0,n-1-k]$.
For $j=0$, $g^0(p)=p \sqsubseteq y_{n-1}$. The last inequality holds by \eqref{eq:Pepsilon}. For $j\in [1,n-1-k]$,
\begin{align*}
g^{j}(p)
    & = g (\, g^{j-1} (p) \,) & \tag{def.}\\
    & \sqsupseteq g(y_{n-1-(j-1)} )  \tag{Induction hypothesis} &  \\
& = g(y_{n-j} )  \tag{def.} &  \\
    & \sqsupseteq y_{n-1-j}  & \tag{\eqref{eq:negativeG}}
\end{align*}
This concludes the proofs of all main invariants of {\APDR}.
\qed

\subsection{Proof of Theorem~\ref{th:soundness}: Soundness for {\APDR}}\label{app:soundness}

\begin{proof}[Proof of Theorem~\ref{th:soundness}]
The first point exploits Knaster-Tarski. The second Kleene.
\begin{enumerate}
\item Observe that {\APDR} returns true if $x_{j+1} \sqsubseteq x_j$. By \eqref{eq:positiveF}, we thus have $f(x_j) \sqsubseteq x_{j+1} \sqsubseteq x_j$. Moreover, by \eqref{eq:Ix1} and \eqref{eq:positivechain}, it holds that $i \sqsubseteq x_j$ and $x_j\sqsubseteq p$. Therefore, it holds that
\[(f\sqcup i) x_j \sqsubseteq  x_j \sqsubseteq p.\]
By \eqref{eq:coinductionproofprinciple}, we have that $\mu (f\sqcup i) \sqsubseteq p$.
\item Observe that {\APDR} returns false if $i \not \sqsubseteq y_1$. By \eqref{negativefinal},  $g^{n-2}(p) \sqsubseteq y_1$. Thus  $i \not \sqsubseteq g^{n-2}(p)$. Moreover
\begin{align*}
g^{n-2}p & \sqsupseteq \bigsqcap_{j\in \omega} g^{j}(p) & \tag{def.}\\
&  = \nu (g \sqcap p) & \tag{\eqref{eq:kleeneadjoint}}
\end{align*}
Thus $i \not \sqsubseteq \nu (g \sqcap p) $. By \eqref{eq:iff}, $\mu (f \sqcup i) \not \sqsubseteq p$.\qed
\end{enumerate}
\end{proof}



\subsection{Positive Chains and Negative Sequences Form Semilattices}\label{app:semilattices}
We show that positive chains form a join-semilattice and negative sequences a meet-semilattices.
We recall that joins and meets are defined point-wise, i.e., for all positive chains $\vec{x^1}, \vec{x^2}$ and negative sequences  $\vec{y^1}, \vec{y^2}$, $(\vec{x^1\sqcup x^2})_j \defeq x^1_j\sqcup x^2_j$ and $(\vec{y^1\sqcap y^2})_j \defeq y^1_j\sqcap y^2_j$. The meet semilattice of negative sequences has as bottom element the sequence defined for all $j\in [k,n-1]$ as $g^{n-1-j}(p)$, see \eqref{negativefinal} in Fig.~\ref{fig:invariants}. The join semilattice of positive chains has as top element the chain defined for all $j\in [0,n-1]$ as $(g \sqcap p)^{n-1-j} (\top)$ and as bottom element the chain $(f \sqcup i)^j (\bot)$, see \eqref{eq:positiveinitialfinal}.
Note that, if  $\vec{y^1}$ and $\vec{y^2}$ are conclusive, also $\vec{y^1\sqcap y^2}$ is conclusive, while such property does not necessarily hold for positive chains.

\begin{lemma}\label{prop:joinpositive} Let $I$ be a set. For all $m \in I$, let $\vec{x^m}=x^m_0, \dots, x^m_{n-1}$ be a positive chain.
Then, the chain $\vec{\bigsqcup_{m\in I} x^m}$ defined for all $j \in [0, n-1]$ as
\[(\vec{\bigsqcup_{m\in I} x^m})_j \defeq \bigsqcup_{m\in I}x^m_{j}\] is a positive chain.
\end{lemma}
\begin{proof}
Since $  i \sqsubseteq x^m_{1}$ for all $m\in I$, then $  i \sqsubseteq \bigsqcup_{m\in I}x^m_{1} $.

Since $  x^m_{n-2} \sqsubseteq p$ for all $m\in I$, then $ \bigsqcup_{m\in I}x^m_{n-2}  \sqsubseteq p$.

To show that $f( (\bigsqcup_{m\in I}\vec{x^m})_{j})  \sqsubseteq (\bigsqcup_{m\in I}\vec{x^m})_{j+1}$ we just observe the following
\begin{align*}
f((\bigsqcup_{m\in I}\vec{x^m})_{j})   &= f(\bigsqcup_{m\in I}x^m_j) & \tag{def.}\\
&  = \bigsqcup_{m\in I} f(x^m_{j} ) & \tag{$f \dashv g$}\\
&  \sqsubseteq \bigsqcup_{m\in I} x^m_{j+1} & \tag{\eqref{eq:positiveF}}\\
& =  (\bigsqcup_{m\in I}\vec{x^m})_{j+1} &\tag{def.}
\end{align*}
Thus \eqref{eq:Ix1}, \eqref{eq:xP} and  \eqref{eq:positiveF} hold for $\vec{\bigsqcup_{m\in I} x^m}$.
\qed
\end{proof}

\begin{lemma}\label{prop:meetneg} Let $I$ be a set. For all $m \in I$, let $\vec{y^m}=y^m_k, \dots, y^m_{n-1}$ be a negative sequence. Then, the sequence $\vec{\bigsqcap_{m\in I} y^m}$ defined for all $j = 0, \dots n-1$ as
\[(\vec{\bigsqcap_{m\in I} y^m})_j \defeq \bigsqcap_{m\in I}y^m_{j}\] is a negative sequence. Moreover, if  $\vec{y^m}$ is conclusive for all $m \in I$, then also $\vec{\bigsqcap_{m\in I} y^m}$ is conclusive.
\end{lemma}
\begin{proof}
Since $  p \sqsubseteq y^m_{n-1}$ for all $m\in I$, then $  p \sqsubseteq \bigsqcap_{m\in I}y^m_{n-1} $.

To show that $g(\bigsqcap_{m\in I}\vec{y^m})_{j+1}  \sqsubseteq (\bigsqcap_{m\in I}\vec{y^m})_{j}$ we proceed as follows
\begin{align*}
g(\bigsqcap_{m\in I}\vec{y^m})_{j+1}   &= g(\bigsqcap_{m\in I}y^m_j) & \tag{def.}\\
&  = \bigsqcap_{m\in I} g(y^m_{j+1} ) & \tag{$f \dashv g$}\\
&  \sqsubseteq \bigsqcap_{m\in I} y^m_{j} & \tag{\eqref{eq:negativeG}}\\
& =  (\bigsqcap_{m\in I}\vec{y^m})_{j} &\tag{def.}
\end{align*}

For conclusive observe that, since  $i \not \sqsubseteq y_1^m$ for all $m \in I$, then  $i \not \sqsubseteq \bigsqcap_{m\in I}y_1^m = (\bigsqcap_{m\in I}\vec{y^m})_{1}$.
\qed
\end{proof}

The fact that the bottom element of such meet-semilattice is given by the sequence $g^{n-1-j}(p)$ is exactly the statement of invariant \eqref{negativefinal}.
The top and bottom element for positive chains are defined as in invariant \eqref{eq:positiveinitialfinal}.


\subsection{Proof of Section~\ref{sec:progression}: Progression}\label{app:prog}

\begin{proof}[Proof of Proposition~\ref{prop:CanonicalChoice}]
For each rule, we prove that if the guard of the rule is satisfied then the choice of $z$ satisfies the required constraints.
\begin{enumerate}
\item The guard of (Candidate) is $x_{n-1} \not \sqsubseteq p$. By choosing $z=p$, one has that $x_{n-1} \not \sqsubseteq z$ and $p \sqsubseteq z$ are trivially satisfied;
\item The guard of (Decide) is $f(x_{k-1}) \not \sqsubseteq y_k$ thus, by $f \dashv g$, $x_{k-1} \not \sqsubseteq g(y_k)$. By choosing $z= g(y_k)$, one has that $x_{k-1} \not \sqsubseteq z$ and $g(y_k) \sqsubseteq z$;
\end{enumerate}
The proofs for the choices in (Conflict) are more subtle. First of all, observe that if $k=1$, then $i\sqsubseteq y_1$ otherwise the algorithm would have returned false. Moreover, for $k\geq 2$, we have that
$i \sqsubseteq x_{k-1} \sqsubseteq y_k$: the first inequality holds by \eqref{eq:Ix1} and the second by \eqref{positivefinal} and \eqref{negativefinal}. In summary,
\begin{equation}\label{eq:yigeqI}
\text{for all } j \geq 1\text{, } i \sqsubseteq y_j\text{.}
\end{equation}
We can then proceed as follows.
\begin{enumerate}\setcounter{enumi}{2}
\item The guard of (Conflict) is $f(x_{k-1}) \sqsubseteq y_k$. By choosing $z=y_k$, one has that $z \sqsubseteq y_k$ trivially holds. For $(f \sqcup i)(x_{k-1} \sqcap z) \sqsubseteq z$ observe that
\begin{align*}
(f \sqcup i)(x_{k-1} \sqcap z) & = f(x_{k-1} \sqcap z) \sqcup i & \tag{def.}\\
&  \sqsubseteq f(x_{k-1} ) \sqcup i  & \tag{monotonicity}\\
&  \sqsubseteq z \sqcup i & \tag{guard}\\
& = z & \tag{\eqref{eq:yigeqI}}
\end{align*}

\item The guard of (Conflict) is $f(x_{k-1}) \sqsubseteq y_k$. By choosing $z=(f \sqcup i)(x_{k-1})$, one has that $(f \sqcup i)(x_{k-1} \sqcap z) \sqsubseteq (f \sqcup i)(x_{k-1}) =z$ holds by monotonicity.
For $z \sqsubseteq y_k$, by using the guard and \eqref{eq:yigeqI}, we have that $z= (f \sqcup i)(x_{k-1}) = f(x_{k-1}) \sqcup i \sqsubseteq y_k$.
\qed
\end{enumerate}
\end{proof}

\begin{proof}[Proof of Proposition~\ref{prop:progres}]
	Let us consider the following partial order on positive chains: given two sequences $\vec{x}= x_0, \dots x_{n-1}$ and $\vec{x}' = x'_0, \dots, x'_{n'-1}$, we say $\vec{x} \preceq \vec{x}'$ if
	\begin{equation*}
		n \le n' \land x_j \sqsupseteq x'_j \text{ for each } j\in[0,n - 1]
	\end{equation*}
	We extend the order to states by letting $(\vec{x} \| \vec{y} )_{n, k} \preceq (\vec{x}' \| \vec{y}' )_{n', k'}$ with $\vec{x} \prec \vec{x}'$ or $\vec{x} = \vec{x}'$ and $k \ge k'$.

	We prove the first statement by showing that applying a rule strictly increases the state in that partial order.
	As before, we use non-primed variables such as $\vec{x}$ for values before the application of a rule, and primed variables such as $\vec{x}'$ after.

	For (Unfold), we have that $n < n' = n + 1$ and $x_j = x'_j$ for each $j\in[0,n - 1]$.

	For (Candidate), we have $\vec{x}' = \vec{x}$ and $k' = n - 1 < n = k$.

	For (Decide), we have $\vec{x}' = \vec{x}$ and $k' = k - 1 < k$.

	For (Conflict), $n = n'$, and
	\[
	x'_j = \begin{cases*}
		x_j & if $j > k$ \\
		x_j \sqcap z & if $j \le k$
	\end{cases*}
	\]
	So for $j\in[k+1,n - 1]$ we have $x_j = x'_j$, and for $j\in[0,k]$ we have $x_j \sqsupseteq x_j \sqcap z = x'_j$. So $\vec{x} \preceq \vec{x}'$. Assume by contradiction that $x'_k = x_k$. Since $x'_k = x_k \sqcap z$, this is equivalent to $x_k \sqsubseteq z$. The choice of $z$ in (Conflict) satisfies $z \sqsubseteq y_k$, that would imply $x_k \sqsubseteq z \sqsubseteq y_k$. However, this is a contradiction, since by \eqref{eq:positivenegative} we know $x_k \not \sqsubseteq y_k$. Hence $x_k \sqsupset x'_k$, meaning $\vec{x} \prec \vec{x}'$.
\qed
\end{proof}

The following lemma will be useful later on to prove negative termination.

\begin{lemma}\label{lem:decreasing}
If $s_0 \ttr{}  ( \vec{x} \| \vec{y} )_{n,k} \ttr{}  ( \vec{x}' \| \vec{y}' )_{n',k'}$, then $n'\geq n$ and, for all $j\in [0,n-1]$, $x_j \sqsupseteq x_j'$.
\end{lemma}
\begin{proof}
Follows from proof of Proposition~\ref{prop:progres}.
	\qed
\end{proof}


\subsection{Proofs of Section~\ref{sec:heuristics}: Heuristics}\label{app:heuristics}
In this appendix, after illustrating the proof of  Proposition~\ref{prop:negativesequencefinalchain}, we show  two results Propositions \eqref{prop:heuristic-final-chain} and \eqref{prop:heuristic-initial-chain}, about simple intial and final heuristics, briefly mentioned in Section~\ref{sec:heuristics}. These two results will not be used by other proofs.

\begin{proof}[Proof of Proposition~\ref{prop:negativesequencefinalchain}]
	As for the invariants, we prove this equality by induction showing
	\begin{itemize}
		\item[(a)] it holds for $s_0$ and
		\item[(b)] if it holds for $s$ and $s \tr{ } s'$, then it holds for $s'$.
	\end{itemize}

	In the initialization and after (Unfold), since $k = n$ there is no $j \in [k, n-1]$.

	For (Conflict), since the property holds on $\vec{y}$ it also holds on $\vec{y}' = \mathsf{tail}(\vec{y})$.

	For (Candidate), $\vec{y}' = p$ and $k' = n - 1$, so the thesis holds because $y_{n-1} = p = g^{n-1-(n-1)} p$.

	For (Decide), $\vec{y}' = g(y_k), \vec{y}$ and $k' = k - 1$. For all $j \in [k' + 1, n-1]$ the thesis holds because $y'_j = y_j$. For $j = k'$, we have $y_{k'} = g(y_k) = g(g^{n - 1 -k}( p)) = g^{n - 1 - k'}(p)$.
	\qed
\end{proof}


\begin{proposition}\label{prop:heuristic-final-chain}
Assume $p \neq \top$ and let $h\colon \states \to L$ be any simple final heuristic.
If $s_0\ttr{}\tr{U} ( \vec{x} \| \vec{y} )_{n,k}$
then the second inequality in \eqref{eq:positiveinitialfinal} holds as an equality, namely for all $j\in[1,n-1]$,
$x_j=(g\sqcap p)^{n-1-j}(\top)$.
\end{proposition}
\begin{proof}
	To prove this property, first we prove by induction that the following invariants hold:
	\begin{itemize}
		\item[(a)] either $x_1 = (g \sqcap p)^{n-1} \top$ or $x_1 = (g \sqcap p)^{n-2} \top$,
		\item[(b)] if $x_1 = (g \sqcap p)^{n-1} \top$, for all $j \in [1, k-1]$, $x_j = (g \sqcap p)^{n-j} \top$ and for all $j \in [k, n-1]$, $x_j = (g \sqcap p)^{n-1-j} \top$,
		\item[(c)] if $x_1 = (g \sqcap p)^{n-2} \top$, for all $j \in [1, n-1]$, $x_j = (g \sqcap p)^{n-1-j} \top$,
		\item[(d)] if $k = 1$ then $x_1 = (g \sqcap p)^{n-2} \top$.
	\end{itemize}
	Note that by (a) exactly one of of the consequences of (b) and (c) hold, and when $k = 1$ (d) prescribes it must be (c). In the rest of the proof, we say that (b) or (c) hold meaning that $x_1$ is $(g \sqcap p)^{n-1} \top$ or $(g \sqcap p)^{n-2} \top$ respectively, so that the respective consequence holds, too.

	At initialization, $\vec{x} = \bot, \top$ and $n = 2$, so (a) $x_1 = (g \sqcap p)^{n-2} \top$, and (c) holds.

	After (Unfold), $\vec{x}' =  \vec{x}, \top$ and $n' = n+1$. Since we applied (Unfold), it must be the case that $\vec{y} = \varepsilon$, so $k = n$, and $x_{n-1} \sqsubseteq p$. By (a) either $x_{n-1} = (g \sqcap p) \top = p$, or $x_{n-1} = (g \sqcap p)^{0} \top = \top$. But since $x_{n-1} \sqsubseteq p$, it can't be the latter. Thus (b) holds before (Unfold).
	After the rule, (a) holds because $x'_1 = x_1 = (g \sqcap p)^{n-1} \top = (g \sqcap p)^{n'-2} \top$; (c) holds too because for all $j \in [1, n'-1] = [1, n]$, $x'_j = x_j = (g \sqcap p)^{n-j} \top = (g \sqcap p)^{n'-1-j} \top$ (where we used that (b) holds before the rule) and for $j = n'$, $x'_j = \top = (g \sqcap p)^0 \top$. (d) holds because $k' = n' > 1$.

	For (Candidate), since we applied it, it must be the case that $x_{n-1} \not\sqsubseteq p$. If (b) held before the rule, it would mean that $x_{n-1} = (g \sqcap p) \top = p \sqsubseteq p$, so (c) holds.
	After the rule, (a) and (c) still hold because $\vec{x}' = \vec{x}$ and $n' = n$. (d) holds because (c) holds.

	For (Decide), since we applied it, it must be the case that $f (x_{k-1}) \not\sqsubseteq y_k = g^{n-1-k} (p)$. This, by $f \dashv g$, is equivalent to $x_{k-1} \not\sqsubseteq g^{n-k} (p)$. If (b) held before the rule, it would mean that $x_{k-1} = (g \sqcap p)^{n-k+1} \top \sqsubseteq g^{n-k} (p)$, so (c) holds.
	After the rule, (a) and (c) still hold because $\vec{x}' = \vec{x}$ and $n' = n$. (d) holds because (c) holds.

	For (Conflict), let us distinguish two cases. 
	
	If $k = 1$, (c) holds before the rule because of (d). The choice in the rule is $z = y_1 = g^{n-2} (p)$, so after (Conflict) $\vec{x}' = \vec{x} \sqcap_1 g^{n-2} (p)$ and $k' = k+1$. (a) holds because $x'_1 = x_1 \sqcap g^{n-2} (p) = (g \sqcap p)^{n-2} \top \sqcap g^{n-2} (p) = (g \sqcap p)^{n-1} \top$. (b) holds because for $j \in [1, k' -1] = [1, 1]$, $x'_1 = (g \sqcap p)^{n-1} \top$ and for $j \in [k', n-1]$, $x'_j = x_j = (g \sqcap p)^{n-1-j} \top$. (d) holds because $k' = 2 > 1$.

	If $k > 1$, since we applied (Conflict), it must be the case that $f (x_{k-1}) \sqsubseteq y_k = g^{n-1-k} (p)$. This, by $f \dashv g$, is equivalent to $x_{k-1} \sqsubseteq g^{n-k} (p)$. If (c) held before the rule, it would mean that $x_{k-1} = (g \sqcap p)^{n-k} \top \not\sqsubseteq g^{n-k} (p)$, so (b) holds.
	The choice in the rule is $z = y_k = g^{n-1-k} (p)$. After (Conflict), (a) holds because $x'_1 = x_1 \sqcap g^{n-1-k} (p) = (g \sqcap p)^{n-1} \top \sqcap g^{n-1-k} (p) = (g \sqcap p)^{n-1} \top$. (b) holds because for $j \in [1, k'-2] = [1, k-1]$, $x'_j = x_j \sqcap g^{n-1-k} (p) = (g \sqcap p)^{n-j} \top \sqcap g^{n-1-k} (p) = (g \sqcap p)^{n-j} \top$; for $j = k'-1 = k$, $x'_j = x_j \sqcap g^{n-1-k} (p) = (g \sqcap p)^{n-1-k} \top \sqcap g^{n-1-k} (p) = (g \sqcap p)^{n-k} \top$; for $j \in [k', n-1]$, $x'_j = x_j = (g \sqcap p)^{n-1-j} \top$. (d) holds because $k' = k + 1 > 1$.

	This concludes the proof of the invariants. To prove the original statement, it is enough to observe that right after (Unfold) (c) holds.
	\qed
\end{proof}

In Proposition~\ref{prop:heuristic-initial-chain}, note that the condition holds for all $j$ but $n-1$, that is the last element of the sequence can be different. Actually, this is $\top$ when added with (Unfold), and is then lowered to $(f\sqcup i)^{n-1}\bot$ right before the next (Unfold).

\begin{proposition}\label{prop:heuristic-initial-chain}
Assume $p \neq \top$ and let $h\colon \states \to L$ be any simple initial heuristic.  For all $( \vec{x} \| \vec{y} )_{n,k} \in \states^h$, the first inequality in \eqref{eq:positiveinitialfinal} holds as an equality for all $j\in[0,n-2]$, namely
$x_j=(f\sqcup i)^{j}(\bot)$.
\end{proposition}
\begin{proof}
	We prove this property by showing the exact sequence of steps the algorithm performs with this heuristics. To do so, assume $m$ is the smallest integer such that $f^m (i) \not \sqsubseteq p$. Note that, if $\mu (f \sqcup i) \sqsubseteq p$, there is no such $m$; if that is the case we say $m = +\infty$. Also note that, for all $n < m$ it holds $f^n (i) \sqsubseteq p$, hence also $f (f \sqcup i)^n \bot = f \bigsqcup\limits_{j < n} f^j (i) = \bigsqcup\limits_{j < n} f^{j+1} (i) \sqsubseteq p$, and for $n \sqsubseteq m$ we have $(f \sqcup i)^n \bot = \bigsqcup\limits_{j < n} f^j (i) \sqsubseteq p$.

	Intuitively, while $n < m + 2$ the algorithm performs a cycle of (Unfold), then (Candidate), then (Conflict) and then (Unfold) again. When it reaches $n = m + 2$, it enters a sequence of (Decide) that eventually lead to return false. Of course, if $m = + \infty$, this sequence of (Decide) never happens.

	To prove this formally, we prove by induction that after initialization and every (Unfold) applied in this sequence, for all $j \in [0,n-2]$, $x_j=(f\sqcup i)^{j}\bot$. We do so by showing this holds for initialization, and that if we assume this to be true after initialization or (Unfold), (a) if $n < m + 2$, the algorithm does exactly (Candidate), then (Conflict), then (Unfold) and the invariant holds again, and (b) if $n = m + 2$, the algorithm does (Candidate) then (Decide) until $k = 1$, then returns false. In doing so, we also show that the invariant holds after every rule, that is exactly the thesis.

	For initialization, as $k = n = 2$ and $x_0 = \bot$ the property holds.

	For (a), suppose the algorithm just did (Unfold) or initialization. Then, by the invariant, for all $j \in [0,n-2]$, $x_j=(f\sqcup i)^{j}\bot$, and both after initialization and (Unfold), $x_{n-1} = \top$.
	Applying the algorithm, the sequence of states is then
	\begin{align*}
		&( \bot, (f\sqcup i)\bot, \dots, (f\sqcup i)^{n-2}\bot, \top \| \varepsilon )_{n,n} \\
		\trz{\mathit{Ca}}{h} &( \bot, (f\sqcup i)\bot, \dots, (f\sqcup i)^{n-2}\bot, \top \| p )_{n,n-1} \\
		\trz{\mathit{Co}}{h} &( \bot, (f\sqcup i)\bot, \dots, (f\sqcup i)^{n-2}\bot, (f\sqcup i)^{n-1}\bot \| \varepsilon )_{n,n} \\
		\trz{U}{h} &( \bot, (f\sqcup i)\bot, \dots, (f\sqcup i)^{n-2}\bot, (f\sqcup i)^{n-1}\bot, \top \| \varepsilon )_{n+1,n+1}
	\end{align*}
	where the choice for (Candidate) is $p$ and the choice for (Conflict) is $(f\sqcup i) x_{n-2} = (f\sqcup i)^{n-1} \bot$.
	The condition to apply (Candidate) is $x_{n-1} = \top \not\sqsubseteq p$, which follows from $p \neq \top$.
	The condition to apply (Conflict) is $f (x_{n-2}) = f (f\sqcup i)^{n-2}\bot \sqsubseteq p$, which follows from $n - 2 < m$.
	The condition to apply (Unfold) is $x_{n-1} = (f\sqcup i)^{n-1}\bot \not\sqsubseteq p$, which follows from $n - 1 \le m$.
	The invariant clearly holds for all three states traversed.

	For (b), suppose again the algorithm just did (Unfold) or initialization, so $\vec{x} = \langle \bot, (f\sqcup i)\bot, \dots, (f\sqcup i)^{n-2}\bot, \top \rangle$. Recalling that $n = m + 2$, the sequence of states is
	\begin{align*}
		&( \bot, (f\sqcup i)\bot, \dots, (f\sqcup i)^{m}\bot, \top \| \varepsilon )_{m+2,m+2} \\
		\trz{\mathit{Ca}}{h} &( \bot, (f\sqcup i)\bot, \dots, (f\sqcup i)^{m}\bot, \top \| p )_{m+2,m+1} \\
		\trz{D}{h} &( \bot, (f\sqcup i)\bot, \dots, (f\sqcup i)^{m}\bot, \top \| g (p), p )_{m+2,m} \\
		\trz{D}{h} &\dots \\
		\trz{D}{h} &( \bot, (f\sqcup i)\bot, \dots, (f\sqcup i)^{m}\bot, \top \| g^{m} (p), \dots, g (p), p )_{m+2,1}
	\end{align*}
	where the choice for (Candidate) is $p$ and the choice for (Decide) is $g (y_{k})$.
	The condition to apply (Candidate) is again $x_{n-1} = \top \not\sqsubseteq p$.
	The condition to apply (Decide) for $k$ is $f (x_{k-1}) \not \sqsubseteq y_k$, that is $f (f\sqcup i)^{k-1}\bot \not \sqsubseteq g^{m+1-k} (p)$. This holds because $f (f\sqcup i)^{k-1}\bot \sqsubseteq f^{k-1} (i)$ and, by $f \dashv g$, $f^{k-1} (i) \not\subseteq g^{m+1-k} (p)$ if and only if $f^{m+1-k} f^{k-1} (i) = f^{m} (i) \not\sqsubseteq p$.
	Lastly, when $k = 1$, we have $i \not\sqsubseteq g^m (p)$ again by $f \dashv g$, so the algorithm returns false.
	The invariant clearly holds for all the $m$ states traversed.
	\qed
\end{proof}


\subsection{Proofs of Section~\ref{sec:termination}: Negative Termination}\label{app:neg}

The following lemma is the key to prove termination.

\begin{lemma}\label{lemma:differentz}
If $s_0 \ttr{} s \trz{D}{z}  \ttr{ } s' \trz{D}{z'} $ and $s$ and $s'$ carry the same index $(n,k)$, then $z' \neq z$.
Similarly, if $s_0 \ttr{} s \trz{\mathit{Ca}}{z}  \ttr{ } s' \trz{\mathit{Ca}}{z'} $ and $s$ and $s'$ carry the same index $(n,k)$, then $z' \neq z$.
\end{lemma}

\begin{proof}
Since $s$ and $s'$ carry the same index $(n, k)$ and the algorithm only increases $n$, then $n$ is never increased in the steps between $s$ and $s'$. Thus (Unfold) is never executed. On the other hand, $k$ will be increased by (Conflict) and decreased in (Candidate) and (Decide).

We prove the proposition for (Decide), the case of (Candidate) is analogous. Let us fix $s=( \vec{x} \| \vec{y} )_{n,k}$. The state immediately after $s$ would be $( \vec{x} \| z, \vec{y} )_{n,k-1}$.
Observe that before arriving to the state $( \vec{x}' \| \vec{y}' )_{n,k}$,  the $z$ inserted by the (Decide) right after $s$ should be removed by some (Conflict), as that's the only rule that can remove elements from $\vec{y}$.
The state before such (Conflict) will be of the form $(\vec{x}'' \| z, \vec{y} )_{n,k+1}$ for some positive chain $\vec{x}''$. Now let $z''$ be the element chosen by such (Conflict). It holds that $z'' \sqsubseteq y_{k+1} = z$. The state after the (Conflict) will be $(\vec{x}''\sqcap_k z'' \| \vec{y} )_{n,k}$.  In this state and, by Corollary \ref{lem:decreasing} in any of the following states, the $(k-1)$-th element of the positive chain is below $z$. In particular, for $s'=( \vec{x}' \| \vec{y}' )_{n,k}$, we have that $x_{k-1}' \sqsubseteq z$. Since $( \vec{x}' \| \vec{y}' )_{n,k} \trz{D}{z'} $, $x_{k-1}'\not \sqsubseteq z'$. Thus $z' \neq z$.
\qed
\end{proof}

\begin{proof}[Proof of Theorem~\ref{thm:negativetermination}]
Let us fix an $n$. Since the possible choices of $z\in h(\mathit{CaD}^h_{n,k})$ are finitely many, by Lemma \ref{lemma:differentz}  we can apply (Candidate) or (Decide) only a finite amount of times for every $k$.
Since by invariant~\eqref{eq:invi}, $1 \le k \le n$, we only have a finite amount of different values of $k$, so (Candidate) and (Decide) occur only finitely many times with the same $n$.

Since both (Candidate) and (Decide) decrease $k$, (Conflict) increase $k$ and $1 \le k \le n$, then also (Conflict) occurs only finitely many times with the same $n$.
Therefore in any infinite computation of \texttt{A-PDR}$_h$ (Unfold), which is the only rule that increase $n$, should occur infinitely many times.

But when $\mu (f \sqcup i) \not\sqsubseteq p$, by \eqref{eq:Kleenefpthm}, there is some $j\in \Nat$ such that $(f \sqcup i)^{j} \bot \not\sqsubseteq p$. Since (Unfold) can be applied only when $(f \sqcup i)^{n-1} \bot \sqsubseteq x_{n-1} \sqsubseteq p$, then it can applied only a finite amount of time.
\qed
\end{proof}

\begin{proof}[Proof of Corollary~\ref{cor:negativetermiantion}]
By Proposition \ref{prop:negativesequencefinalchain}, $h$ maps any reachable state $s$ such that $s\tr{D}$ into $g^{n-1-k}(p)$ and any reachable state $s$ such that $s\tr{\mathit{Ca}}$ into $p$. Thus $h(\mathit{CaD}^h_{n,k})$ has cardinality 2. By Theorem \ref{thm:negativetermination}, if  $\mu(f\sqcup i) \not \sqsubseteq p$, then {\APDR}$_h$ terminates.
\qed
\end{proof}


\section{Proofs of Section~\ref{sec:downset}}
In this appendix, we illustrate the proofs for the various results in Section~\ref{sec:downset}. After illustrating the proofs of Proposition~\ref{prop:prob_down_up} in Appendix \ref{app:AAA}, we will show the proofs for the results in Sections \ref{ssec:ADPDR} and \ref{ssec:LTPDRvsADPDR} in Appendixes \ref{app:ADPDR} and \ref{app:LTPDRvsADPDR}.


\begin{figure}[t]
\centering
\underline{{\APDR} $(\bot^\downarrow, b^\downarrow,b_r^{\downarrow}, p^{\downarrow})$}
\begin{codeNT}
<INITIALISATION>
  $( \vec{X} \| \vec{Y} )_{n,k}$ := $(\emptyset,L\|\varepsilon)_{2,2}$
<ITERATION>
  case $( \vec{X} \| \vec{Y} )_{n,k}$ of
	   $\vec{Y}=\varepsilon$ And $X_{n-1} \subseteq p^{\downarrow}$     :                    %(Unfold)
			$( \vec{X} \| \vec{Y} )_{n,k}$ := $( \vec{X}, L \| \varepsilon )_{n+1,n+1}$
	   $\vec{Y}=\varepsilon$ And $X_{n-1} \not \subseteq p^{\downarrow}$    :                     %(Candidate)
			choose $Z\in L^{\downarrow}$ st  $X_{n-1} \not \subseteq Z$ And  $p^\downarrow \subseteq Z$;
			$( \vec{X} \| \vec{Y} )_{n,k}$ := $( \vec{X} \| Z )_{n,n-1}$
	   $\vec{Y} \neq \varepsilon$ And $b^\downarrow(X_{k-1}) \not \subseteq Y_k$ :                        %(Decide)
			choose $Z \in L^\downarrow$ st $X_{k-1} \not \subseteq Z$ And $b_r^{\downarrow}(Y_k) \subseteq Z$;
			$(\vec{X} \| \vec{Y} )_{n,k}$ := $(\vec{X} \| Z , \vec{Y} )_{n,k-1}$
	   $\vec{Y} \neq \varepsilon$ And $b^\downarrow(X_{k-1}) \subseteq Y_k$ :                        %(Conflict)
			choose $Z \in L^\downarrow$ st $Z \subseteq Y_k$ And $(b^\downarrow \cup \bot^\downarrow)(X_{k-1} \cap Z) \subseteq Z$;
			$(\vec{X} \| \vec{Y} )_{n,k}$ := $(\vec{X} \cap_k Z \| \mathsf{tail}(\vec{Y}) )_{n,k+1}$
  endcase
<TERMINATION>
	if $\exists j\in [0,n-2]\,.\, X_{j+1} \subseteq X_j$ then return true		 % $\vec{X}$ conclusive
	if $\bot^\downarrow \not \subseteq Y_1$ then return false							% $\vec{Y}$ conclusive
\end{codeNT}
\caption{{\APDR} algorithm checking $\mu (b^\downarrow \cup \bot^\downarrow) \subseteq p^\downarrow$.}
\label{fig:APDRonDownset}
\end{figure}

\begin{figure}[t]
\centering
\underline{{\APDR'} $(\bot^\downarrow, b^\downarrow,b_r^{\downarrow}, p^{\downarrow})$}
\begin{codeNT}
<INITIALISATION>
  $( \vec{X} \| \vec{Y} )_{n,k}$ := $(\emptyset,\top^\downarrow\|\varepsilon)_{2,2}$
<ITERATION>
  case $( \vec{X} \| \vec{Y} )_{n,k}$ of								% $\vec{X}$ has the form $\emptyset,x_1^\downarrow,\dots,x_{n-1}^\downarrow$
	   $\vec{Y}=\varepsilon$ And $x_{n-1} \sqsubseteq p$     :                    %(Unfold)
			$( \vec{X} \| \vec{Y} )_{n,k}$ := $( \vec{X}, \top^\downarrow \| \varepsilon )_{n+1,n+1}$
	   $\vec{Y}=\varepsilon$ And $x_{n-1} \not \sqsubseteq p$    :                     %(Candidate)
			choose $Z\in L^{\downarrow}$ st  $x_{n-1} \not \in Z$ And  $p\in Z$;
			$( \vec{X} \| \vec{Y} )_{n,k}$ := $( \vec{X} \| Z )_{n,n-1}$
	   $\vec{Y} \neq \varepsilon$ And $b^\downarrow(x^\downarrow_{k-1}) \not \subseteq Y_k$ :                        %(Decide)
			choose $Z \in L^\downarrow$ st $x_{k-1} \not \in Z$ And $b_r^{\downarrow}(Y_k) \subseteq Z$;
			$(\vec{X} \| \vec{Y} )_{n,k}$ := $(\vec{X} \| Z , \vec{Y} )_{n,k-1}$
	   $\vec{Y} \neq \varepsilon$ And $b^\downarrow(x^\downarrow_{k-1}) \subseteq Y_k$ :                        %(Conflict)
			choose $z \in L$ st $z \in Y_k$ And $(b^\downarrow \cup \bot^\downarrow)(x^\downarrow_{k-1} \cap z^\downarrow) \subseteq z^\downarrow$;
			$(\vec{X} \| \vec{Y} )_{n,k}$ := $(\vec{X} \cap_k z^\downarrow \| \mathsf{tail}(\vec{Y}) )_{n,k+1}$
  endcase
<TERMINATION>
	if $\exists j\in [0,n-2]\,.\, x^\downarrow_{j+1} \subseteq x^\downarrow_j$ then return true		 % $\vec{X}$ conclusive
	if $Y_1=\emptyset$ then return false							% $\vec{Y}$ conclusive
\end{codeNT}
\caption{{\APDR} algorithm checking $\mu (b^\downarrow \cup \bot^\downarrow) \subseteq p^\downarrow$, where we restrict the elements of the positive chain to be principals. Note that: in (Unfold) the condition $x_{n-1} \sqsubseteq p$ is equivalent to $x^\downarrow_{n-1} \subseteq p^{\downarrow}$; and similarly for their negation in (Candidate), where moreover the condition $x^\downarrow_{n-1} \not \subseteq Z$ is equivalent to $x\not\in Z$; same for (Decide); finally in (Conflict) the condition $z \in L$ is equivalent to $z^\downarrow \in L^\downarrow$ and the condition $z \in Y_k$ is equivalent to $z^\downarrow \subseteq Y_k$.}
\label{fig:APDRonDownset2}
\end{figure}

\subsection{Proof of Proposition~\ref{prop:prob_down_up}}\label{app:AAA}

\begin{proof}[Proof of Proposition~\ref{prop:prob_down_up}]
A simple inductive argument using \eqref{eq:EMlaw} confirms that
\begin{equation}\label{eq:hdownn}
(b^n x)^\downarrow = (b^\downarrow)^n x^\downarrow
\end{equation}
for all $x\in L$. The following sequence of logical equivalences
  \begin{align*}
    \mu b \sqsubseteq p &\Leftrightarrow \forall n \in \mathbb{N}.~b^n \bot \sqsubseteq p \tag{by \eqref{eq:Kleenefpthm}} \\
    &\Leftrightarrow \forall n \in \mathbb{N}.~(b^n \bot)^\downarrow \subseteq p^\downarrow \tag{mon. of $(-)^\downarrow$, $\bigsqcup$ and $\bigsqcup (-)^{\downarrow}=id$} \\
    &\Leftrightarrow \bigcup_{n \in \mathbb{N}} (b^n \bot)^\downarrow  \subseteq p^\downarrow \tag{def. of $\bigcup$} \\
    &\Leftrightarrow \bigcup_{n \in \mathbb{N}} (b^\downarrow)^n \bot^\downarrow \subseteq p^\downarrow \tag{by \eqref{eq:hdownn}} \\
    &\Leftrightarrow \mu (b^\downarrow \cup \bot^\downarrow) \subseteq p^\downarrow \tag{by \eqref{eq:kleeneadjoint}}.
  \end{align*}
concludes the proof of the main statement.
\qed
\end{proof}

\subsection{Proofs of Section \ref{ssec:ADPDR}: From {\APDR} to {\ADPDR}}\label{app:ADPDR}

\begin{proof}[Proof of Theorem~\ref{th:ADPDR}]
The algorithm {\ADPDR} differs from the instance of {\APDR} on $(\bot^\downarrow, b^\downarrow,b_r^{\downarrow}, p^{\downarrow})$ for two main reasons: first we restrict the elements of the positive chain to be principals, second we optimize the initial state of the algorithm.

To prove that the properties of {\APDR} can be extended to {\ADPDR}, we first show the instance of {\APDR} on $(\bot^\downarrow, b^\downarrow,b_r^{\downarrow}, p^{\downarrow})$ for the lower set domain $(L^\downarrow,\subseteq)$ in Fig.~\ref{fig:APDRonDownset}. Clearly, as instance of {\APDR}, the algorithm in Fig.~\ref{fig:APDRonDownset} inherits all its properties  about soundness, progression and negative termination.
 Note that, in (Candidate), the condition $p^\downarrow \subseteq Z$ is equivalent to $p\in Z$, because $Z\in L^\downarrow$.
Moreover, we note that the negative termination condition $\bot^\downarrow \not \subseteq Y_1$ amounts to $Y_1 = \emptyset$.

To restrict the elements of the positive chain to be principals we need to add the condition $\exists z \in L.~Z=z^\downarrow$ in rule (Conflict), which is thus modified as follows w.r.t. Fig.~\ref{fig:APDRonDownset}: 

\begin{codeNT}
$\vec{Y} \neq \varepsilon$ And $b^\downarrow(x^\downarrow_{k-1}) \subseteq Y_k$ :                        %(Conflict)
	choose $Z \in L^\downarrow$ st $Z \subseteq Y_k$ And $(b^\downarrow \cup \bot^\downarrow)(X_{k-1} \cap Z) \subseteq Z$
		 								   And $\exists z\in L. Z=z^\downarrow$;
	$(\vec{X} \| \vec{Y} )_{n,k}$ := $(\vec{X} \cap_k Z \| \mathsf{tail}(\vec{Y}) )_{n,k+1}$
\end{codeNT}

Let us call {\APDR'} such algorithm. All the executions of {\APDR'} are also possible in {\APDR}, thus all the invariants of {\APDR} holds for {\APDR'}. 
The invariants suffice to prove Theorem~\ref{th:soundness}, Proposition~\ref{prop:progres} and Theorem~\ref{thm:negativetermination}.

The elements of the positive chain are introduced by (Unfold) and modified by (Conflict).
By choosing $Z=z^\downarrow$ in (Conflict) it follows that all the elements of the positive chain are also principals, with the only exception of $X_0=\emptyset$. Indeed, every new element of the positive chain has that form (in (Unfold) we take $\top_{L^\downarrow}=\top^\downarrow$) and the meet of two principals $x_j^\downarrow$ and $z^\downarrow$ in (Conflict) is itself the principal $(x_j\sqcap z)^\downarrow$ generated by the meet of $x_j$ and $z$.

Regarding the canonical choices of Proposition~\ref{prop:CanonicalChoice}, 

\medskip

\noindent \begin{minipage}{.5\linewidth}
\begin{enumerate}\setcounter{enumi}{0}
\item in (Candidate) $Z=p^\downarrow$;
\item in (Decide) $Z= b_r^\downarrow(Y_k)$;
\end{enumerate}
\end{minipage}
\begin{minipage}{.5\linewidth}
\begin{enumerate}\setcounter{enumi}{2}
\item in (Conflict) $Z = Y_k$;
\item in (Conflict) $Z = (b^\downarrow \cup \bot^\downarrow)(X_{k-1})$.
\end{enumerate}
\end{minipage}

\medskip

\noindent
we have that choice 3 is not necessarily possible, because we cannot assume that $Y_k=y_k^\downarrow$ for some $y_k\in L$, but 1, 2 and 4 remain valid choices: in fact,
1 and 2 deal with the negative sequence for which we have no restriction; for 4, if $X_{k-1} = x_{k-1}^\downarrow$ for some $x_{k-1}\in L$, then
$Z=(b^\downarrow \cup \bot^\downarrow)(X_{k-1})
=(b^\downarrow \cup \bot^\downarrow)(x_{k-1}^\downarrow)
=b^\downarrow(x_{k-1}^\downarrow) \cup \bot^\downarrow
\stackrel{\eqref{eq:EMlaw}}{=}b(x_{k-1})^\downarrow \cup \bot^\downarrow
=b(x_{k-1})^\downarrow$
is still a principal.
Since choices 1, 2 and 4 guarantees the existence of a simple heuristic (i.e., the initial one) we also have that Corollary~\ref{cor:negativetermiantion} about negative termination is valid for {\APDR'}.

We now take advantage of the  shape of the positive chain to present the code of {\APDR'} as reported in Fig.~\ref{fig:APDRonDownset2}: we exploit the fact that $\vec{X}=\emptyset, x_1^\downarrow,\dots,x_{n-1}^\downarrow$ to make some simple code transformations described in the caption. Now we note that $b^\downarrow(x^\downarrow_{k-1})=b(x_{k-1})^\downarrow$, so that the conditions $b^\downarrow(x^\downarrow_{k-1}) \not \subseteq Y_k$ in (Decide) and $b^\downarrow(x^\downarrow_{k-1}) \subseteq Y_k$ in (Conflict) are equivalent to $b(x_{k-1})\not \in Y_k$ and to $b(x_{k-1}) \in Y_k$, respectively.
Moreover,  we have
$(b^\downarrow \cup \bot^\downarrow)(x^\downarrow_{k-1} \cap z^\downarrow)
= (b^\downarrow \cup \bot^\downarrow)((x_{k-1} \sqcap z)^\downarrow)
= (b \sqcup \bot)(x_{k-1} \sqcap z)^\downarrow
= (b (x_{k-1} \sqcap z)\sqcup \bot)^\downarrow
= b (x_{k-1} \sqcap z)^\downarrow$,
so that $(b^\downarrow \cup \bot^\downarrow)(x^\downarrow_{k-1} \cap z^\downarrow) \subseteq z^\downarrow$
is equivalent to 
$b (x_{k-1} \sqcap z)^\downarrow\subseteq z^\downarrow$
and therefore also to
$b(x_{k-1} \sqcap z) \sqsubseteq z$.

Then, the only difference between {\APDR'} and {\ADPDR} is the initialization condition:  {\ADPDR} starts from the state reached after the following three steps of  {\APDR'}:
\[
(\emptyset,\top^\downarrow\|\varepsilon)_{2,2}
\tr{\mathit{Ca}} (\emptyset,\top^\downarrow\|p^\downarrow)_{2,1}
\tr{\mathit{Co}} (\emptyset,\bot^\downarrow\|\varepsilon)_{2,2}
\tr{U} (\emptyset,\bot^\downarrow,\top^\downarrow\|\varepsilon)_{3,3}\textrm{.}
\]

Since the three steps apply the canonical choices for (Candidate) and (Conflict), we conclude that {\ADPDR} is sound and that it enjoys progression and negative termination.
\qed
\end{proof}

\subsection{Proofs of Section~\ref{ssec:LTPDRvsADPDR}: LT-PDR vs {\ADPDR}}\label{app:LTPDRvsADPDR}

\begin{figure}[t]
\centering
\underline{LT-PDR $(b,p)$}
\begin{codeNT}
<INITIALISATION>
  $( \vec{x} \| \vec{c} )_{n,k}$ := $(\bot,b(\bot)\|\varepsilon)_{2,2}$
<ITERATION>
  case $( \vec{x} \| \vec{c} )_{n,k}$ of
	   $\vec{c}=\varepsilon$ And $x_{n-1} \sqsubseteq p$     :                    %(Unfold)
			$( \vec{x} \| \vec{c} )_{n,k}$ := $( \vec{x}, \top \| \varepsilon )_{n+1,n+1}$
	   $\vec{c}=\varepsilon$ And $x_{n-1} \not \sqsubseteq p$    :                     %(Candidate)
			choose $z\in L$ st  $z \sqsubseteq x_{n-1}$ And  $z \not\sqsubseteq p$;
			$( \vec{x} \| \vec{c} )_{n,k}$ := $( \vec{x} \| z )_{n,n-1}$
	   $\vec{c} \neq \varepsilon$ And $c_k \sqsubseteq b(x_{k-1})$ :                        %(Decide)
			choose $z \in L$ st $z \sqsubseteq x_{k-1}$ And $c_k \sqsubseteq b(z)$;
			$(\vec{x} \| \vec{c} )_{n,k}$ := $(\vec{x} \| z , \vec{c} )_{n,k-1}$
	   $\vec{c} \neq \varepsilon$ And $c_k \not \sqsubseteq b(x_{k-1})$ :                        %(Conflict)
			choose $z \in L$ st $c_k \not \sqsubseteq z$ And $b(x_{k-1} \sqcap z) \sqsubseteq z$;
			$(\vec{x} \| \vec{c} )_{n,k}$ := $(\vec{x} \sqcap_k z \| \mathsf{tail}(\vec{c}) )_{n,k+1}$
  endcase
<TERMINATION>
	if $\exists j\in [0,n-2]\,.\, x_{j+1} \sqsubseteq x_j$ then return true 		% $\vec{x}$ is conclusive
	if $k = 1$ then return false							% $\vec{c}$ is conclusive
\end{codeNT}
\caption{LT-PDR algorithm checking $\mu b \sqsubseteq p$, adapted from \cite{KoriCAV22}.}
\label{fig:LTPDR2}
\end{figure}

\begin{proof}[Proof of Theorem~\ref{th:LT-PDR-instance-ADPDR}]
For the scope of this proof, we call $\states = \{( \vec{x} \| \vec{Y} )_{n,k} \}$ the set of states of {\ADPDR}, and $\states' = \{ (\vec{x}'\| \vec{c}' )_{n',k'} \}$ that of LT-PDR, where we use non-primed variables for the former and primed for the latter. 
%
The function $\mathcal{R} \colon \states' \rightarrow \states$ is defined for all states $s' = ( \vec{x}' \| \vec{c}' )_{n',k'} \in \states'$, as \[\mathcal{R}(s') = ( \vec{x} \| \vec{Y})_{n,k} \in \states\] 
where
\[ n = n' + 1, \qquad k = k' + 1, \qquad \vec{x} = \emptyset, \vec{x}'\qquad \text{ and } \qquad \vec{Y} = \negation{x'}\text{.}\]
We prove that $\mathcal{R}$  is a simulation \cite{Mil89}, that is for all $s',t' \in \states'$, if $s' \rightarrow t'$ then $\mathcal{R}(s') \rightarrow \mathcal{R}(t')$. The pseudo-code of LT-PDR is reported in Fig.~\ref{fig:LTPDR2}.

\medskip

First, we remark that, for any $z, x \in L$, $x \not\in \lnot (\{ z \}^{\uparrow})$ if and only if $z \sqsubseteq x$. Moreover, note that indices $\vec{x}'$ in $s'$ and $\vec{x}$ in $\mathcal{R}(s')$ are off-setted by one: $x_j = x'_{j-1}$. However, as $n = n' + 1$ we have, for instance, $x_{n-1} = x'_{n'-1}$ (and analogously for $k'$).

Consider now a state $s' = ( \vec{x}' \| \vec{c}' )_{n',k'} \in \states'$ and $\mathcal{R}(s') = ( \vec{x} \| \vec{Y} )_{n,k} \in \states$. Suppose that LT-PDR can perform a transition from $s'$. This must be determined by one of the four rules of LT-PDR, possibly performing some choice of $z \in L$. We show that {\ADPDR} is able to simulate this transition, using the same rule and performing a corresponding choice. We do so by cases on the rule used by LT-PDR.
	
	\begin{itemize}
		\item If LT-PDR applies rule (Unfold), we have $\vec{c}' = \varepsilon$ and $x'_{n'-1} \sqsubseteq p$ so that
		$$s'=( \vec{x}' \| \varepsilon )_{n',n'} \tr{U} ( \vec{x}',\top \| \varepsilon)_{n'+1,n'+1}=t'.$$
		Then, for $( \vec{x} \| \vec{Y} )_{n,k} = \mathcal{R}(s') = ( \emptyset,\vec{x}' \| \varepsilon )_{n'+1,n'+1}$ it holds $\vec{Y} = \varepsilon$ and $x_{n-1} = x'_{n'-1} \sqsubseteq p$, so {\ADPDR} can apply (Unfold) too and
		$$\mathcal{R}(s')=( \emptyset,\vec{x}' \| \varepsilon )_{n'+1,n'+1} \tr{U} ( \emptyset,\vec{x}',\top \| \varepsilon)_{n'+2,n'+2}=\mathcal{R}(t').$$
		
		\item If LT-PDR applies rule (Candidate), we have $\vec{c}' = \varepsilon$ and $x'_{n'-1} \not \sqsubseteq p$, so that $z\in L$ is chosen such that $z\sqsubseteq x'_{n'-1}$ and $z\not\sqsubseteq p$ to derive
		$$s'=( \vec{x}' \| \varepsilon )_{n',n'} \tr{\mathit{Ca}}_z ( \vec{x}' \| z)_{n',n'-1}=t'.$$
		Then, for $( \vec{x} \| \vec{Y} )_{n,k} = \mathcal{R}(s') = ( \emptyset,\vec{x}' \| \varepsilon )_{n'+1,n'+1}$ it holds $\vec{Y} = \varepsilon$ and $x_{n-1}=x'_{n'-1} \not \sqsubseteq p$, so that {\ADPDR} can apply (Candidate) too. Moreover we can choose $Z \defeq \lnot (\{ z \}^{\uparrow})$, because $z \sqsubseteq x'_{n'-1}$ implies $x_{n-1}=x'_{n'-1} \not \in Z$, and $z \not \sqsubseteq p$ implies $p \in Z$. By doing so we derive
		$$\mathcal{R}(s')=( \emptyset,\vec{x}' \| \varepsilon )_{n'+1,n'+1} \tr{\mathit{Ca}}_z ( \emptyset,\vec{x}' \| Z)_{n'+1,n'}=\mathcal{R}(t').$$
		
		\item If LT-PDR applies rule (Decide), we have $\vec{c}'\neq \varepsilon$ and $c'_{k'} \sqsubseteq b(x'_{k'-1})$,  so that $z\in L$ is chosen such that $z \sqsubseteq x'_{k'-1}$ and $c'_{k'} \sqsubseteq b(z)$ to derive
		$$s'=( \vec{x}' \| \vec{c}' )_{n',k'} \tr{D}_z ( \vec{x}' \| z,\vec{c}')_{n',k'-1}=t'.$$
		Then, for $( \vec{x} \| \vec{Y} )_{n,k} = \mathcal{R}(s') = ( \emptyset,\vec{x}' \| \negation{c'} )_{n'+1,k'+1}$ it holds $\negation{c'} \neq \varepsilon$ and $b(x_{k-1})=b(x'_{k'-1})\not\in \negation{c'}_k=\lnot (\{ c'_{k-1} \}^{\uparrow})$, because the latter is implied by $c'_{k'} \sqsubseteq b(x'_{k'-1})$. Thus {\ADPDR} can apply (Decide) too. Moreover we can choose $Z \defeq \lnot (\{ z \}^{\uparrow})$. In fact $z \sqsubseteq x'_{k'-1}$ implies $x_{k-1}=x'_{n'-1} \not \in \lnot (\{ z \}^{\uparrow})=Z$. Moreover $b^{\downarrow}_r(Y_{k}) \subseteq Z$ if and only if $x \not \in Z$ implies $b(x) \not \in Y_{k}$. Because $Z = \lnot (\{ z \}^{\uparrow})$ and $\negation{c'}_{k} = \lnot (\{ c'_{k'} \}^{\uparrow})$, this implication is equivalent to requiring that $z \sqsubseteq x$ implies $c'_{k'} \sqsubseteq b(x)$, which is true as $c'_{k'} \sqsubseteq b(z)$ and $b$ is monotone. With this choice of $Z$ we derive
		$$\mathcal{R}(s')=( \emptyset,\vec{x}' \| \vec{Y} )_{n'+1,k'+1} \tr{D}_Z ( \emptyset,\vec{x}' \| Z,\vec{Y})_{n'+1,k'}=\mathcal{R}(t').$$
		
		\item If LT-PDR applies rule (Conflict), we have $\vec{c}'\neq \varepsilon$ and $c'_{k'} \not\sqsubseteq b(x'_{k'-1})$,  so that $z\in L$ is chosen such that $c'_{k'} \not \sqsubseteq z$ and $b(x'_{k'-1}\sqcap z)\sqsubseteq z$ to derive
		$$s'=( \vec{x}' \| \vec{c}' )_{n',k'} \tr{\mathit{Co}}_z ( \vec{x}'\sqcap_{k'} z \| \mathsf{tail}(\vec{c}'))_{n',k'+1}=t'.$$
		Then, for $( \vec{x} \| \vec{Y} )_{n,k} = \mathcal{R}(s') = ( \emptyset,\vec{x}' \| \negation{c'} )_{n'+1,k'+1}$ it holds $\negation{c'} \neq \varepsilon$ and $b(x_{k-1})=b(x'_{k'-1})\in Y_k=\lnot (\{ c'_{k-1} \}^{\uparrow})$, because the latter is implied by $c'_{k'} \not\sqsubseteq b(x'_{k'-1})$. Thus {\ADPDR} can apply (Conflict) too. Moreover we can choose the same $z$ as LT-PDR. In fact $c'_{k'} \not \sqsubseteq z$ implies $z \in \lnot (\{ c'_{k'} \}^{\uparrow})=\lnot (\{ c'_{k-1} \}^{\uparrow})=\negation{c'}_{k}$, and $b(x_{k-1} \sqcap z)=b(x'_{k'-1}\sqcap z) \sqsubseteq z$ is also an hypothesis in LT-PDR. With this choice of $z$ we derive
		$$\mathcal{R}(s')=( \emptyset,\vec{x}' \| \vec{Y} )_{n'+1,k'+1} \tr{\mathit{Co}}_z ( \emptyset,\vec{x}' \sqcap_{k} z\| \mathsf{tail}(\vec{Y}))_{n'+1,k'+2}=\mathcal{R}(t').$$
	\end{itemize}
	This concludes the proof that $\mathcal{R}$ is a simulation. However, to complete the proof that {\ADPDR} simulates LT-PDR, we have to take care of initial and final states.

	We observe that the initial states $s'_0 = (\bot,b(\bot)\|\varepsilon)_{2,2} \in \states'$ of LT-PDR and $s_0 = (\emptyset,\bot,\top\|\varepsilon)_{3,3} \in \states$ of {\ADPDR} are not related: $\mathcal{R}(s'_0) \neq s_0$. To solve this issue, first observe that if $b(\bot) \not \sqsubseteq p$, both algorithms return false in a few steps. If instead $b(\bot) \sqsubseteq p$, {\ADPDR} can reach the state $\mathcal{R}(s'_0)$ from $s_0$ in just two steps:
	\[
		s_0 = (\emptyset,\bot,\top\|\varepsilon)_{3,3}
		\tr{\mathit{Ca}} ( \emptyset, \bot, \top \| p^{\downarrow} )_{3,2}
		\tr{\mathit{Co}}_{ b(\bot) } ( \emptyset, \bot, b(\bot) \| \varepsilon )_{3,3} = \mathcal{R}(s'_0)
	\]

	Lastly, we discuss the termination conditions of the two algorithms. 
	
	When LT-PDR terminates from a state $s'$ returning true, $s'$ satisfies $x'_{j+1} \sqsubseteq x'_j$ for some $j$, so also {\ADPDR} terminates from $\mathcal{R}(s')$ returning true. 
	
	Instead, when LT-PDR terminates from $s'$ returning false, the condition $k' = 1$ does not imply that {\ADPDR} terminates from $s = \mathcal{R}(s')$. However, the latter algorithm can always apply (Decide) from $s$: as proved in \cite{KoriCAV22}, the termination condition $k' = 1$ of LT-PDR implies $c'_1 \sqsubseteq b(\bot)$, which in turn means $b(x_1) = b(x'_0) = b(\bot) \not \in Y_2 = \lnot (\{ c'_1 \}^{\uparrow})$. Moreover, we can choose $Z \defeq \emptyset$: for all $x \in L$ we have $c'_1 \sqsubseteq b(\bot) \sqsubseteq b(x)$, so $b(x) \not \in Y_2$. After this step, we get that $Y_1 = \emptyset$, so {\ADPDR} returns false, too.
	\qed
\end{proof}

\begin{proposition}\label{prop:negLTPDR}
Let $\vec{c}$ be a Kleene sequence. Then $\negation{c}$ is a negative sequence. 
\end{proposition}
\begin{proof}[Proof of Proposition \ref{prop:negLTPDR}]
First, we show that $p \in \negation{c}_{n-1}$. Since $c_{n-1}\not \sqsubseteq p$, by (C1), then $p \not\in \{c_{n-1}\}^\uparrow$. Thus $p \in \neg ( \{c_{n-1}\}^\uparrow)$, that is $p \in \negation{c}_{n-1}$.

Then, we show that $b_r^\downarrow(\negation{c}_{j+1}) \subseteq \negation{c}_j$. 
\begin{align*}
b_r^\downarrow(\negation{c}_{j+1})  &= b_r^\downarrow(\neg(\{c_{j+1}\}^\uparrow)) & \tag{def.}\\
&  = \{x \mid b(x) \notin (\{c_{j+1}\}^\uparrow \} & \tag{def.}\\
&  =  \{x \mid c_{j+1} \not \sqsubseteq b(x) \}  & \tag{def.}\\
&  \subseteq  \{x \mid b(c_j) \not \sqsubseteq b(x) \}  & \tag{(C2)}\\
&  \subseteq  \{x \mid c_j \not \sqsubseteq x \}  & \tag{mon. of $b$}\\
&  =  \neg(\{c_j\}^\uparrow)  & \tag{def.}\\
& =  \negation{c}_j &\tag{def.}
\end{align*}
\qed
\end{proof}

\begin{proof}[Proof of Proposition \ref{prop:multipleLTPDR}]
Since each of the $\vec{c^m}$ is a Kleene sequence, then for all $m\in M$, then $\negation{c^m}$ is, by Proposition \ref{prop:negLTPDR}, a negative sequence.
Thus, by Lemma \ref{prop:meetneg} (in Appendix~\ref{app:semilattices}), their intersection is also a negative sequence.
\qed
\end{proof}

%
%
%

\begin{proof}[Proof of Proposition \ref{prop:LTPDRfinal}]
Since $\vec{c} = c_0, \dots, c_{n-1}$ is a Kleene sequence, then $\negation{c}=\lnot (\{ c_k \}^{\uparrow}), \dots, \lnot (\{ c_{n-1} \}^{\uparrow})$ is, by Proposition \ref{prop:negLTPDR}, a negative sequence.
 Thus, by \eqref{negativefinal}, for all $j\in [k,n-1]$,$(b_r^\downarrow)^{n-1-j}(p^\downarrow) \subseteq  \neg(\{c_j\}^\uparrow)$. Thus  $\neg (b_r^\downarrow)^{n-1-j}(p^\downarrow) \supseteq \{c_j\}^\uparrow$ and thus $c_j \in \neg (b_r^\downarrow)^{n-1-j}(p^\downarrow) = \neg (Y_j)$.
\qed
\end{proof}

\section{Proofs of Section \ref{sec:MDP}}

In this appendix we illustrate the proofs about the heuristics introduced in Section \ref{sec:MDP} to deal with the max reachability problem for MDPs.


\begin{proof}[Proof of Corollary \ref{cor:ADPDRtermination}]
Using Theorem \ref{thm:negativetermination}, it is enough to prove that, for all indexes $(n,k)$, the set of all possible choices for (Candidate) and (Decide), denoted by $h(\mathit{CaD}^h_{n,k})$, is finite. For simplicity, in this proof let $\mathit{CaD}^h_{n,k} = \mathit{Ca}^h_{n,k} \uplus D^h_{n,k}$, where 
$\mathit{Ca}^h_{n,k} \defeq \{ s\in \mathit{CaD}^h_{n,k} \mid s \tr{\mathit{Ca}}\}$ is the set of reachable $(n,k)$-indexed states that trigger (Candidate) and 
$D^h_{n,k} \defeq \{ s\in \mathit{CaD}^h_{n,k} \mid s\tr{D}\}$  is the set of reachable $(n,k)$-indexed states that trigger (Decide).
We prove that both images $h(\mathit{Ca}^h_{n,k})$ and $h(D^h_{n,k})$ are finite so that also $h(\mathit{CaD}^h_{n,k})$ is such.
\begin{itemize}
\item In (Candidate), any heuristics $h$ defined as in \eqref{eq:secondterminatingheuristics} always choose $z=p^\downarrow$. 
Therefore the image $h(\mathit{Ca}^h_{n,k})$ has cardinality $1$ for all $(n,k)$.
\item In (Decide), when $k=n-1$, any heuristic $h$ defined as in  \eqref{eq:secondterminatingheuristics} may select any lower set $\{d \mid b_\alpha(d)\in p^\downarrow\}$ for some function $\alpha\colon S \to A$. Since, there are $|S|^{|A|}$ of such functions $\alpha$, then there are at most $|S|^{|A|}$ of such lower set. Thus the set $h(D^h_{n,n-1})$ has at most cardinality $|S|^{|A|}$. 

When $k=n-2$, the heuristic $h$ may select any lower set $\{d \mid b_{\alpha_2}(b_{\alpha_2}(d))\in p^\downarrow\}$ for any two functions $\alpha_1,\alpha_2\colon S \to A$. Thus, the set $h(D^h_{n,n-2})$ has at most cardinality $|(|S^{|A|})^2$.

One can easily generalise these cases to arbitrary $j\in [1,n]$, and prove with a simple inductive argument that the cardinality of $h(D^h_{n,n-j})$ is at most $(|S|^{|A|})^j$. Since both $S$ and $A$ are finite, then the set $h(D^h_{n,k})$ is finite for all indices $(n,k)$. 
\qed
\end{itemize}
\end{proof}

\begin{proof}[Proof of Proposition \ref{prop:genlegit}]
We need to prove that the choices of \verb|hCoB| and \verb|hCo01| for (Candidate), (Decide) and (Conflict) respect the constraints imposed by {\ADPDR}.

\begin{itemize}
\item For (Candidate), both \verb|hCoB| and \verb|hCo01| take $Z = p^\downarrow$. We need to prove that $x_{n-1} \not \in Z$ and $p\in Z$. 
\begin{itemize}
\item For the former, recall that the guard of (Candidate) is $x_{n-1} \not \sqsubseteq p$. Thus $x_{n-1}\not \in p^\downarrow = Z$.
\item The second is trivial: $p\in p^\downarrow =Z$. 
\end{itemize}

%
\item For (Decide), both \verb|hCoB| and \verb|hCo01|  take $Z=\{d\mid b_\alpha(d) \in Y_k\}$ for an $\alpha$ such that $b_\alpha(x_{k-1})\notin Y_k$. We need to prove $x_{k-1}\notin Z$ and $b_r^\downarrow(Y_K)\subseteq Z$.
\begin{itemize}
\item Since $b_\alpha(x_{k-1})\notin Y_k$, then $x_{k-1}\notin \{d\mid b_\alpha(d) \in Y_k\}=Z$.
\item To see that $b_r^\downarrow(Y_k)\subseteq Z$, it is enough to observe that $b^\downarrow_r (Y_k) = \{d \mid b(d) \in Y_k\}
= \bigcap_{\alpha}\{d \mid b_{\alpha} (d)\in Y_k\} \subseteq \{d \mid b_{\alpha} (d)\in Y_k\} = Z$.
\end{itemize}

\item For (Conflict), we start with \verb|hCoB| and show later \verb|hCo01|. 

Let us consider first the case  $\mathcal{Z}_k=\emptyset$, then \verb|hCoB| chooses $z_B=b(x_{k-1})$. The proof is the same as for Proposition \ref{prop:CanonicalChoice}.4.

Let us consider now the case $\mathcal{Z}_k \neq \emptyset$, and let $z_k = \bigwedge \mathcal{Z}_k$, so that
\begin{equation*}
z_{B} \defeq
\begin{cases*}
		z_k(s)  & if $r_s \neq 0$ \\
		b(x_{k-1})(s) & if $r_s = 0$
\end{cases*}
\end{equation*}
 We first prove  $z_k\in Y_k$ and $b(x_{k-1} \wedge z_k) \leq z_k$.
\begin{itemize}
\item Since $\mathcal{Z}_k \neq \emptyset$, then there should be at least a $d \in \mathcal{Z}_k$. By definition, $d$ is a (convex) generator of $Y_k$ and thus $d\in Y_k$. Since $z_k=\bigwedge \mathcal{Z}_k \leq d$ and since $Y_k$ is, by definition, a lower set, then $z_k\in Y_k$.
\item By definition of $\mathcal{Z}_k$, $b(x_{k-1})\leq d$, for all $d\in \mathcal{Z}_k$. Thus $b(x_{k-1})\leq \bigwedge \mathcal{Z}_k =z_k$. Therefore $b(x_{k-1} \wedge z_k) \leq b(x_{k-1}) \leq z_k$.
\end{itemize}
Now let us show that $z_B\in Y_k$ and $b(x_{k-1} \wedge z_B) \leq z_B$.
\begin{itemize}
\item Since $Y_k=\{d\in [0,1]^S \mid \sum_{s\in S}(r_s \cdot d(s)) \leq r \}$ and $z_B$ differs from $z_k$ only when $r_s=0$, we have $\sum_{s\in S}(r_s \cdot z_B(s)) = \sum_{s\in S}(r_s \cdot z_k(s)) \leq r$, because we already proved that $z_k\in Y_k$.
\item We know that $b(x_{k-1})\leq z_k$. Then, for any $s\in S$ it follows that $b(x_{k-1})(s)\leq z_k(s) = z_B(s)$ if $r_s \neq 0$, and $b(x_{k-1})(s) = z_B(s)$  if $r_s = 0$. Therefore $b(x_{k-1})\leq z_B$, from which we get $b(x_{k-1} \wedge z_B) \leq b(x_{k-1}) \leq z_B$.
\end{itemize}

Finally, we focus on \verb|hCo01|: we need to prove $z_{01}\in Y_k$ and $b(x_{k-1} \wedge z_{01}) \leq z_{01}$.
\begin{itemize}
\item Since $Y_k=\{d\in [0,1]^S \mid \sum_{s\in S}(r_s \cdot d(s)) \leq r \}$ and $z_{01}$ differs from $z_B$ only when $r_s=0$, we have $\sum_{s\in S}(r_s \cdot z_{01}(s)) = \sum_{s\in S}(r_s \cdot z_B(s)) \leq r$, because we already proved that $z_B\in Y_k$.
\item We know that $b(x_{k-1})\leq z_B$. 
Then, for any $s\in S$ it follows that $b(x_{k-1})(s) \leq \lceil z_B(s)\rceil = z_{01}(s)$  if $r_s = 0$ and $\mathcal{Z}_k\neq \emptyset$, and that $b(x_{k-1})(s)\leq z_B(s) = z_{01}(s)$ otherwise. Therefore $b(x_{k-1})\leq z_{01}$, from which it readily follows $b(x_{k-1} \wedge z_{01}) \leq b(x_{k-1}) \leq z_{01}$.
\qed
\end{itemize}
\end{itemize}
\end{proof}

\end{document}